\documentclass[runningheads,a4paper]{llncs}
\usepackage{amssymb,amsmath,pifont,footnote,epic,eepic,gastex}
\usepackage{latexsym}
\usepackage{array}
\usepackage{stmaryrd}
\usepackage{algorithmic,algorithm}

\newif\ifarxiv
\arxivtrue

\usepackage[top=3cm, bottom=2.5cm, left=3cm, right=2.5cm] {geometry}

\newtheorem{thm}{Theorem} 

\newtheorem{lem}{Lemma}

\newtheorem{rem}{Remark}
\newtheorem{cor}{Corollary}
\newtheorem{Obs}{Observation}

\newtheorem{examp}{Example}

\newcommand{\pfbox}{\quad\hspace*{\fill}$\Box$}
\newcommand{\NewChange}[1]{#1}

\newcommand{\Z}{\mathbb{Z}}
\newcommand{\Q}{\mathbb{Q}}

\newcommand{\Set}[1]{\{ #1 \}}
\newcommand{\Range}[2]{#1 ,\dots, #2}
\newcommand{\RangeSet}[2]{\Set{ \Range{#1}{#2}}}

\newcommand{\Automat}[1]{\ensuremath{\mathcal{#1}}}

\newcommand{\Heading}[1]{\vspace{-0.25cm}\paragraph{\bf{#1}}}
\newcommand{\BeginProof}{\vspace{-0.25cm}\begin{proof}}

\newcommand{\Comment}[1]{}

\newcommand{\Avg}{\mathsf{Avg}}
\newcommand{\LimInfAvg}{\mathsf{LimInfAvg}}
\newcommand{\LimSupAvg}{\mathsf{LimSupAvg}}

\newcommand{\Nat}{\ensuremath{\mathbb{N}}}

\begin{document}

\pagestyle{plain}

\def\One{\mathit{One}}
\makesavenoteenv{tabular}
\newcommand{\CMPSup}{\ensuremath{\bigwedge \textrm{MeanPayoffSup} }}
\newcommand{\CMPInf}{\ensuremath{\bigwedge \textrm{MeanPayoffInf} }}
\newcommand{\DMPSup}{\ensuremath{\bigvee \textrm{MeanPayoffSup} }}
\newcommand{\DMPInf}{\ensuremath{\bigvee \textrm{MeanPayoffInf} }}

\def\Skip{\mathit{skip}}
\def\Pop{\mathit{pop}}
\def\Push{\mathit{push}}
\def\Com{\mathsf{Com}}
\def\com{\mathit{com}}
\def\SH{\mathsf{SH}}
\def\ASH{\mathsf{ASH}}
\def\Top{\mathsf{Top}}
\newcommand{\wps}{\mathcal{A}}
\newcommand{\atuple}[1]{\langle #1 \rangle}
\newcommand{\ov}{\overline}
\newcommand{\Gr}{\mathsf{Gr}}
\newcommand{\CycleFree}{\mathsf{CycleFree}}
\newcommand{\Manipulated}{\mathsf{Man}^{\tau}_{\epsilon}}
\newcommand{\False}{\mathsf{FALSE}}
\newcommand{\True}{\mathsf{TRUE}}
\def\Calls{\mathsf{Calls}}
\def\Returns{\mathsf{Retns}}
\newcommand{\LocalHistory}{\mathsf{LocalHistory}}


\makeatletter

\begingroup \catcode `|=0 \catcode `[= 1
\catcode`]=2 \catcode `\{=12 \catcode `\}=12
\catcode`\\=12 |gdef|@xcomment#1\end{comment}[|end[comment]]
|endgroup

\def\@comment{\let\do\@makeother \dospecials\catcode`\^^M=10\def\par{}}

\def\begincomment{\@comment\@xcomment}

\makeatother

\ifarxiv
\newenvironment{comment}{\begincomment}{}
\fi

\title{The Complexity of Mean-Payoff Pushdown Games\thanks{A preliminary 
version of the paper appeared in LICS (Logic in Computer Science) 2012.}$^,$\thanks{The research 
was supported by Austrian Science Fund (FWF) Grant No P 23499-N23,
FWF NFN Grant No S11407-N23 (RiSE), ERC Start grant (279307: Graph Games), 
Microsoft faculty fellows award, the Israeli Centers of Research Excellence 
(I-CORE) program, (Center  No. 4/11), and the RICH Model Toolkit (ICT COST Action IC0901).}
} 
\author{Krishnendu Chatterjee \and Yaron Velner}
\institute{IST Austria \and  School of Computer Science and Engineering, The Hebrew University, Israel}

\maketitle
\pagestyle{plain}

\begin{abstract}
Two-player games on graphs are central in many problems in formal verification 
and program analysis such as synthesis and verification of open systems. 
In this work we consider solving recursive game graphs (or pushdown game 
graphs) that model the control flow of sequential programs with 
recursion. 
While pushdown games have been studied before with qualitative objectives, 
such as reachability and $\omega$-regular objectives, 
in this work we study for the first time such games with the most well-studied 
quantitative objective, namely, mean-payoff objective.
In pushdown games two types of strategies are relevant: (1)~global strategies, 
that depend on the entire global history; and (2)~modular strategies, that 
have only local memory and thus do not depend on the context of invocation,
but only on the history of the current invocation of the module. 
Our main results are as follows:
(1)~One-player pushdown games with mean-payoff objectives under global 
strategies are decidable in polynomial time.
(2)~Two-player pushdown games with mean-payoff objectives under global 
strategies are undecidable.
(3)~One-player pushdown games with mean-payoff objectives under modular 
strategies are NP-hard.
(4)~Two-player pushdown games with mean-payoff objectives under modular 
strategies can be solved in NP (i.e., both one-player and two-player 
pushdown games with mean-payoff objectives under modular strategies are 
NP-complete).
We also establish the optimal strategy complexity by showing that global 
strategies for mean-payoff objectives require infinite memory even in 
one-player pushdown games; and memoryless modular strategies are 
sufficient in two-player pushdown games.
Finally we also show that all the problems have the same complexity if 
the stack boundedness condition is added, where along with the mean-payoff
objective the player must also ensure that the stack height is bounded.
\end{abstract}

\ifarxiv
\else

\fi

\section{Introduction}

\noindent{\bf Quantitative verification and synthesis of pushdown systems.}
\emph{Pushdown automata} (a.k.a pushdown systems) are one of the most basic and 
fundamental computation models.
They extend finite-state systems with a single unbounded stack, to model 
the call stack of first-order recursive programs, where the control states
hold valuations of the program's global variables, and stack characters encode 
the local variable valuations.
The fundamental result for the model-checking of pushdown systems was established 
by B\"uchi in~\cite{Buchi64}, who showed how to compute the set of all reachable stack configurations.
Since then, other problems of pushdown model-checking (e.g., wrt 
safety, LTL or $\omega$-regular properties) as well as games over pushdown systems 
have been extensively studied and led to various efficient implementations and 
applications to program analysis~\cite{mooly,alur12,chen2002,henz2002,Wal01,Wal00,AlurParity}.
While the traditional model-checking problem asked for Boolean answers (whether a property
is satisfied or not), recent trends explore more quantitative properties, such 
as performance measure. 
One of the most fundamental and well-studied quantitative objective is the 
{\em mean-payoff or long-run average} objective. 
In this work we study the verification and synthesis problem of pushdown systems 
with mean-payoff objectives.
We start with the description of finite-state game graphs, and then pushdown systems
and its various equivalent models, followed by existing results and then 
our contribution.

\smallskip\noindent{\bf Games on graphs.} 
Two-player games played on finite-state graphs provide the mathematical 
framework to analyze several important problems in computer science as 
well as mathematics.  
In particular, when the vertices of the graph represent the states 
of a reactive system and the edges represent the transitions, then the synthesis problem
(Church's problem) asks for the construction of a winning strategy in a 
game played on the graph~\cite{BuchiLandweber69,RamadgeWonham87,PnueliRosner89,McNaughton93}.
Game-theoretic formulations have also proved useful for the 
verification~\cite{AHK02}, refinement~\cite{FairSimulation}, and compatibility 
checking \cite{InterfaceTheories} of reactive systems.    
Games played on graphs are dynamic games that proceed for an infinite 
number of rounds.
The vertex set of the graph is partitioned into player-1 vertices
and player-2 vertices. 
The game starts at an initial vertex, and if the current vertex is a 
player-1 vertex, then player~1 chooses an outgoing edge, and 
if the current vertex is a player-2 vertex, then player~2 does 
likewise.
This process is repeated forever, and gives rise to 
an outcome of the game, called a {\em play}, that consists of the 
infinite sequence of states that are visited. 
Two-player games on finite-state graphs with qualitative objectives
such as reachability, liveness and $\omega$-regular conditions formalized
by the canonical parity objectives, strong fairness objectives, etc., 
have been extensively studied in the literature~\cite{GH82,EJ88,EJ91,Zie98,Thomas97,WilkeBook}. 

\smallskip\noindent{\bf The extensions.} 
The study of two-player finite-state games with qualitative objectives has 
been extended in two orthogonal directions in the literature: 
(1)~two-player infinite-state games with qualitative objectives; and 
(2)~two-player finite-state games with quantitative objectives.
One of the most well-studied models of infinite-state games with qualitative 
objectives is pushdown games (or games on recursive state machines) that 
can model reactive systems with recursion (or model the control flow of 
sequential programs with recursion).
Pushdown games with reachability and parity objectives have been studied 
in~\cite{Wal01,Wal00,AlurReach,AlurParity} (also see~\cite{EY05,EY09,BBKO11,BBFK08} and the recent work of~\cite{EY15} for sample research in stochastic pushdown games).
The most well-studied quantitative objective is the \emph{mean-payoff} objective,
where a reward is associated with every transition and the goal of one of the 
players is to maximize the long-run average of the rewards (and the goal of
the opponent is to minimize).
Two-player finite-state games with mean-payoff objectives have been studied
in~\cite{EM79,ZP95,LigLip69}, and more recently applied in synthesis of reactive systems with 
quality guarantee~\cite{BCHJ09,C_acm_15} and robustness~\cite{BGHJ09}, as well as 
interprocedural analysis~\cite{CPV15}.
Moreover recently many quantitative logics and automata theoretic formalisms
have been proposed with mean-payoff objectives in their heart to express 
properties such as reliability requirements, and resource bounds of reactive
systems~\cite{CDH10,BCHK11,DM10,ornaACM}.
Thus pushdown games with mean-payoff objectives would be a central 
theoretical question for model checking of quantitative logics
(specifying reliability and resource bounds) on reactive systems 
with recursion feature.
Several applications of mean-payoff pushdown systems in the context of program analysis, such as, resource usage of containers, static profiling of programs for frequency of function calls, are given in~\cite[Section~3,Section~5]{CPV15}.
 
\smallskip\noindent{\bf Pushdown mean-payoff games.} 
In this work we study for the first time pushdown games with mean-payoff objectives
(to the best of our knowledge mean-payoff objectives have not been studied in
the context of pushdown games).
In pushdown games two types of strategies are relevant and studied in the
literature. 
The first are the \emph{global} strategies, where a global strategy can 
choose the successor state depending on the entire global history of the 
play (where history is the finite sequence of configurations of the current 
prefix of a play).
The second are the \emph{modular} strategies, and modular strategies are 
understood more intuitively in the model of games on recursive state machines.
A \emph{recursive state machine} (RSM) consists of a set of component machines 
(or modules). 
Each module has a set of \emph{nodes} (atomic states) and \emph{boxes} 
(each of which is mapped to a module), a well-defined interface consisting of 
\emph{entry} and \emph{exit} nodes, and edges connecting nodes/boxes. 
An edge entering a box models the invocation of the module associated with the 
box and an edge leaving the box represents return from the module.
In the game version the nodes are partitioned into player-1 nodes and 
player-2 nodes.
Due to recursion the underlying global state-space is infinite and isomorphic
to pushdown games. 
The polynomial-time equivalence of pushdown games and recursive games has been established
in~\cite{AlurReach}. 
A modular strategy is a strategy that has only local memory, and thus, 
the strategy does not depend on the context of invocation of the module,
but only on the history within the current invocation of the module. 
In other words, modular strategies are appealing because they are 
stackless strategies, decomposable into one for each module. 
In this work we will study pushdown games with mean-payoff objectives
for both global and modular strategies.

\smallskip\noindent{\bf Previous results.}
Pushdown games with qualitative objectives were studied in~\cite{Wal01,Wal00}.
It was shown in~\cite{Wal01} that solving pushdown games (i.e., 
determining the winner in pushdown games) 
with reachability objectives under global strategies is EXPTIME-hard, 
and pushdown games with parity objectives under 
global strategies can be solved in EXPTIME.
Thus it follows that pushdown games with reachability and 
parity objectives under global strategies are EXPTIME-complete.
The notion of modular strategies in games on recursive state machines
was introduced in~\cite{AlurReach,AlurParity}. 
It was shown that the modular strategies problem is NP-complete 
in pushdown games with reachability and parity objectives in general~\cite{AlurReach,AlurParity}.
The results of~\cite{AlurReach} also presents more refined complexity results 
in terms of the number of exit nodes, showing that if 
every module has single exit, then the problem is polynomial for 
reachability objectives~\cite{AlurReach} and in NP $\cap$ coNP for parity objectives~\cite{AlurParity}.

\smallskip\noindent{\bf Our contributions.} 
In this work we present a complete characterization of the computational 
and strategy complexity of pushdown games and pushdown systems (one-player 
pushdown games or pushdown automata) with mean-payoff objectives. 
Solving a pushdown system (resp. pushdown game) with respect to a mean-payoff
objective is to decide whether there exists a path that (resp. a winning strategy 
to ensure that every path possible given the strategy) satisfies the mean-payoff objective.
Our main results for computational complexity are as follows.
\begin{enumerate}
\item \emph{Global strategies.} 
We show that pushdown systems (one-player pushdown games) 
with mean-payoff objectives under global strategies 
can be solved in polynomial time, whereas solving pushdown games 
with mean-payoff objectives under global strategies is 
undecidable.

\item \emph{Modular strategies.} 
Solving pushdown systems with single exit nodes with mean-payoff objectives
under modular strategies is NP-hard, and pushdown games 
with mean-payoff objectives under modular strategies can be solved
in NP.
Thus both pushdown systems and pushdown games with mean-payoff objectives
under modular strategies are NP-complete.
\end{enumerate}
Our results are shown in Table~\ref{tab-results}.
First observe that our hardness result for modular strategies 
is different from the NP-hardness of~\cite{AlurReach} because the hardness
result of~\cite{AlurReach} shows hardness for \emph{games} with reachability 
objectives and requires that the number of modules with multiple 
exit nodes are not bounded
(in fact if every module of the recursive game has a single exit, 
then the problem is in PTIME for reachability and in NP $\cap$ coNP for parity 
objectives).
In contrast we show that for mean-payoff objectives the problem 
is NP-hard even for pushdown systems (only one player), where every module
has a single exit node, under modular strategies. 
Second, we also observe the very different complexity of global 
and modular strategies for mean-payoff objectives in pushdown systems 
vs pushdown games: 
the global strategies problem is computationally inexpensive (in PTIME) as
compared to the modular strategies problem (which is NP-complete) in pushdown systems;
whereas  the global strategies problem is computationally infeasible 
(undecidable) as compared to the modular strategies problem (which is NP-complete) 
in pushdown games.
Also observe that in contrast to finite-state game graphs 
where the complexities for mean-payoff and parity objectives match, 
for pushdown systems and games, the complexities 
of parity and mean-payoff objectives are very different.
Along with the computational complexities, 
we also establish the optimal strategy complexity showing 
that global winning strategies for mean-payoff objectives in general 
require infinite memory even in pushdown systems; 
whereas memoryless or positional (independent of history) strategies
suffice for modular strategies for mean-payoff objectives
in pushdown games (see Table~\ref{tab-strategy}).
Finally we also study the stack boundedness conditions where the goal
of one player along with maximizing the mean-payoff objective is 
also to ensure that the height of the stack is bounded.
We show that all the complexities for the additional stack boundedness 
condition along with mean-payoff objectives 
are the same in pushdown systems and games as without the stack boundedness
condition.

\smallskip\noindent{\bf Technical contributions.} 
Our key technical contributions are as follows.
For pushdown systems under global strategies we show that the mean-payoff objective problem 
can be solved by only considering additional stack height that is polynomial.
We then show that the stack height bounded problem can be solved in 
polynomial time using a dynamic programming style algorithm.
For pushdown games under global strategies our undecidability result is obtained by a reduction 
from the universality problem of \emph{weighted} automata (which is 
undecidable~\cite{Krob,ABK11}).
For modular strategies we first show the existence of cycle independent 
modular strategies, and then show that memoryless modular strategies 
are sufficient. 
Given memoryless modular strategies and our polynomial-time algorithm for
pushdown systems, we obtain the NP upper bound for the modular strategies
problem.
Our NP-hardness result for modular strategies is a reduction from the
3-SAT problem.

\smallskip\noindent{\bf Organization.}
Our paper is organized as follows.
In Section~\ref{sect:PushDownProcesses} (resp. Section~\ref{sec:games})
we present the results for pushdown systems
(resp. pushdown games) under global strategies. 
In Section~\ref{sec:modular} we present the results for modular
strategies.

\begin{table}[t]
\begin{center}
\begin{small}
\ifarxiv
\else
\tbl{Computational complexity of pushdown systems and pushdown games with mean-payoff objectives.\label{tab-results}}{
\fi
\begin{tabular}{|c|c|c|}
\hline
                & \ Global strategies \ & \ Modular strategies\ \\ 
\hline
\ Pushdown systems \ &    PTIME   & \ NP-complete (NP-hard for single exit) \   \\
						&	(Theorem~\ref{thm:MPSupMPInfInPTIME} and Theorem~\ref{thm:LeqIsInP})			& (Theorem~\ref{thm:ModularStratInNP} and Theorem~\ref{thm:np-hard}) \ \\
\hline
\ Pushdown games \ & \  Undecidable \ &\ NP-complete \ \\
& (Theorem~\ref{thrm:game}) & (Theorem~\ref{thm:ModularStratInNP} and Theorem~\ref{thm:np-hard})  \ \\
\hline
\end{tabular}
\ifarxiv
\else
}
\fi
\end{small}
\end{center}
\ifarxiv
\caption{Computational complexity of pushdown systems and pushdown games with mean-payoff objectives.}\label{tab-results}
\fi
\end{table}

\begin{table}[t]
\begin{center}
\begin{small}
\ifarxiv
\else
\tbl{Strategy complexity of pushdown systems and pushdown games with mean-payoff objectives.\label{tab-strategy}}{
\fi
\begin{tabular}{|c|c|c|}
\hline
                & \ Global strategies \ & \ Modular strategies\ \\ 
\hline
\ Pushdown systems \ &    Infinite   & Memoryless   \\
& (Example~\ref{ex:illustration1}) & (Lemma~\ref{lem:ModularMemorylessDeterminacy}) \\
\hline
\ Pushdown games \ & \  Infinite \ &\ Memoryless  \ \\
& (Example~\ref{ex:illustration1}) & (Lemma~\ref{lem:ModularMemorylessDeterminacy}) \\
\hline
\end{tabular}
\ifarxiv
\else
}
\fi
\end{small}
\end{center}
\ifarxiv
\caption{Strategy complexity of pushdown systems and pushdown games with mean-payoff objectives.}\label{tab-strategy}
\fi
\end{table}

\section{Mean-Payoff Pushdown Graphs}\label{sect:PushDownProcesses}
In this section we consider pushdown graphs (or pushdown systems) with 
mean-payoff objectives.
\NewChange{We start with the basic definitions of pushdown systems and valid paths in pushdown systems, and then give the overview of the solution and introduce basic notations.}

\smallskip\noindent{\em Stack alphabet and commands.}
Let $\Gamma$ denote a finite set of \emph{stack symbols} (called the stack 
alphabet), and 
$\Com(\Gamma) = \Set{\Skip,\Pop} \cup \Set{\Push(z) \mid z\in\Gamma}$ denote the set of 
\emph{stack commands} over $\Gamma$.
Intuitively, the command $\Skip$ does nothing, $\Pop$ deletes the top element of the stack, 
$\Push(z)$ puts $z$ on the top of the stack.
For a stack command $\com$ and a stack string $\alpha \in \Gamma^+$ 
we denote by $\com(\alpha)$ the stack string obtained by executing the 
command $\com$ on $\alpha$.

\smallskip\noindent{\bf Weighted pushdown systems.}
A \emph{weighted pushdown system (WPS)} (or a weighted pushdown graph) is a tuple:
\[
\wps = \atuple{ Q, \Gamma, q_0\in Q, E \subseteq  (Q\times\Gamma) \times (Q\times \Com(\Gamma)), w:E \to \Z},
\]
where $Q$ is a finite set of \emph{states} with $q_0$ as the initial state; 
$\Gamma$ is a finite \emph{stack alphabet} and we assume that there is a special 
initial stack symbol $\bot \in \Gamma$; $E$ describes the set of 
edges or transitions of the pushdown system; and $w$ is a weight function 
that assigns integer weights to every edge (and the weights are encoded in binary).
A \emph{configuration} of a WPS is a pair $(\alpha,q)$ where $\alpha\in \Gamma ^+$ 
is a stack string and $q\in Q$.
We assume that $\bot$ can be neither put nor removed from the stack,  
and thus all configurations must contain $\bot$ at the bottom of the stack.
For a stack string $\alpha$ we denote by $\Top(\alpha)$ the top symbol of the stack.
The initial configuration of the WPS is $(\bot,q_0)$.
We use $W$ to denote the maximal absolute weight of the edge
weights.

\smallskip\noindent{\em Successor configurations, paths, and ultimately periodic paths.}
Given a WPS $\wps$, a configuration $c_{i+1}=(\alpha_{i+1},q_{i+1})$
is a \emph{successor} configuration of a configuration $c_{i}=(\alpha_{i},q_{i})$,
if there is an edge $(q_i,\gamma_i,q_{i+1},\com) \in E$ such that 
$\com(\alpha_i)=\alpha_{i+1}$, where $\gamma_i=\Top(\alpha_i)$.
A \emph{path} $\pi$ is a sequence of configurations.
A path $\pi = \atuple{c_1, \dots, c_{n+1}}$ is a \emph{valid path} if for all 
$1 \leq i \leq n$ the configuration $c_{i+1}$ is a successor configuration of $c_i$
(and the notation is similar for infinite paths).
In the sequel we shall refer only to valid paths.
Let $\pi = \atuple{c_1,c_2,\dots,c_i,c_{i+1},\dots}$ be a path.
We denote by $\pi[j] = c_j$ the $j$-th configuration of the path and 
by $\pi[i_1,i_2] = \atuple{c_{i_1}, c_{i_1 + 1}, \dots,c_{i_2}}$ the 
segment of the path from the $i_1$-th to the $i_2$-th configuration.
A path can equivalently be defined as a sequence $\atuple{c_1 e_1 e_2 \dots e_n}$,
where $c_1$ is the first configuration and $e_i$ are valid transitions.
A path $\pi$ is \emph{ultimately periodic} if there exists a finite sequence 
$\xi_1 \in E^*$ and a non-empty finite sequence $\xi_2 \in E^+$ of transitions 
such that the path consists of $\xi_1$ followed by $\xi_2$ forever, i.e., 
$\pi=\atuple{c_1 \xi_1 (\xi_2)^\omega}$. 
A configuration $c_r$ is \emph{reachable} if there is a finite path that 
begins at the initial configuration and ends in $c_r$.
Similarly, a path $\pi$ is reachable if its first configuration is reachable.

\smallskip\noindent{\em Average weights of paths.}
For a finite path $\pi$, we denote by $w(\pi)$ the sum of the weights of the
edges in $\pi$ and $\Avg(\pi) = \frac{w(\pi)}{|\pi|}$, where $|\pi|$ 
is the length of $\pi$, denotes the average of the weights.
For an infinite path $\pi$, we denote by $\LimSupAvg(\pi)$ (resp. 
$\LimInfAvg(\pi)$)  the limit-sup (resp. limit-inf) of the averages (long-run 
average or mean-payoff objectives), i.e., $\lim\sup(\Avg(\pi[1,i]))_{i\geq 1}$
(resp. $\lim\inf(\Avg(\pi[1,i]))_{i\geq 1}$).
We say that $\pi$ is a \emph{positive path} if $w(\pi) > 0$, and \emph{negative, non-negative} and \emph{non-positive} paths are similarly defined.

\noindent{\bf Mean-payoff objectives with strict and non-strict inequalities.}
For a given integer $r$, the mean-payoff objective $\LimInfAvg\bowtie r$ 
(resp. $\LimSupAvg\bowtie r$) defines the set of infinite paths $\pi$ such 
that $\LimInfAvg(\pi)\bowtie r$ (resp. $\LimSupAvg(\pi)\bowtie r$),
where $\bowtie \in \Set{\geq,>}$. 
Mean-payoff objectives with integer threshold $r$ can be transformed 
to threshold~0 by subtracting $r$ from all transition weights. 
Hence in this work w.l.o.g we will consider the mean-payoff objectives
(i)~$\LimInfAvg>0$ (resp. $\LimSupAvg>0$), and call them mean-payoff objectives
with strict inequality; and
(ii)~$\LimInfAvg\geq0$ (resp. $\LimSupAvg\geq0$), and call them mean-payoff 
objectives with non-strict inequality.
We are interested in solving WPSs with mean-payoff objectives, i.e., to
decide if there is a path that satisfies the objective.

\begin{figure}[!tb]
\begin{center}
\begin{picture}(48,28)(0,0)
\node[Nmarks=i, iangle=180](n0)(10,12){$q_1$}
\node[Nmarks=n](n1)(40,12){$q_2$}
\drawloop[ELside=l,loopCW=y, loopdiam=6](n0){$\Push(\gamma),-1$}
\drawloop[ELside=l,loopCW=y, loopdiam=6](n1){$\Pop(\gamma),1$}
\drawedge[ELpos=50, ELside=l, ELdist=0.5, curvedepth=6](n0,n1){$\Skip,-1$}
\drawedge[ELpos=50, ELside=l, curvedepth=6](n1,n0){$\Skip,-1$}
\end{picture}
\caption{WPS $\wps$ to witness ultimately periodic words might not suffice
for mean-payoff objectives with non-strict inequality.
\NewChange{The $\Pop$ transition is valid only when the stack is not empty.
All other transitions are valid for all stack tops}}\label{fig:ex-wps}
\end{center}
\end{figure}
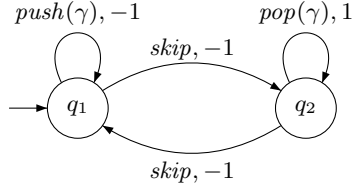
\begin{examp}\label{ex:illustration1}
We now illustrate with an example that \NewChange{ultimately periodic witness paths for non-strict inequality mean-payoff objectives need not exist, and the only witness paths are non-ultimately periodic ones.}
Consider the WPS $\wps$ with two states $Q = \Set{q_1,q_2}$, with two symbol 
stack alphabet $\Gamma = \Set{\bot,\gamma}$, and the edge set 
$E=\Set{e_1,e_2,\ldots,e_5}$ is described as follows:
$e_1 = (q_1,\bot,q_1,\Push(\gamma)), e_2 = (q_1,\gamma,q_1,\Push(\gamma)),
e_3 = (q_1,\gamma,q_2,\Skip), e_4 = (q_2,\gamma,q_2,\Pop),$ and 
$e_5 = (q_2,\bot,q_1,\Skip)$.
The weight function is as follows: $w(e_4)=1$, and all other edge weights are $-1$.
(See Figure~\ref{fig:ex-wps} for a pictorial description).
For $i\geq 1$, consider the path segment $\rho_i=c_1 e_1 e_2^{i-1} e_3 e_4^i e_5$ that executes 
the edge $e_1$, followed by $(i-1)$-times the edge $e_2$, then the edge $e_3$, 
followed by $i$-times the edge $e_4$, and finally the edge $e_5$.
It is straightforward to verify that 
for the infinite path $\pi = (\bot,q_1) \rho_1 \rho_2 \rho_3 \dots $
we have that $\LimSupAvg(\pi)=\LimInfAvg(\pi) = 0$.
However for every valid path $\pi = c_1 \xi_1 (\xi_2)^\omega$, where $\xi_1 \in E^*$ and 
$\xi_2 \in E^+$ it must be the case that either (i)~$\xi_2 = e_2$ and 
then $\LimInfAvg(\pi) =\LimSupAvg(\pi) = -1$ or that 
(ii)~$\xi_2$ is a cycle with length $|\xi_2|$ and 
has weight at most~$-1$,
and hence $\LimInfAvg(\pi) \leq \LimSupAvg(\pi) \leq -\frac{1}{|\xi_2|}<0$.
\hfill\qed
\end{examp}

\smallskip\noindent{\bf Overview of the solution.}
We first characterize the \emph{pumpable paths} in a pushdown graph that 
determine the possible mean-payoff values of the graph.
In a finite-state graph a path with non-negative mean-payoff exists if and 
only if there is a finite pumpable path (namely, a cycle) with non-negative 
weight.
For infinite-state graphs, and pushdown graphs in particular, the latter does 
not hold.
However we show that a path with a strictly positive mean-payoff exists if and 
only if there is a finite pumpable path with positive weight.
For this purpose we first characterize the pumpable paths in a pushdown graph 
and in Section~\ref{subsec:wps1} we obtain a polynomial algorithm to detect 
pumpable paths with a positive weight, and hence we get a polynomial-time 
algorithm to detect a path with a positive mean-payoff in a pushdown graph.
In Section~\ref{subsec:wps2} we show a reduction from the problem of detecting paths 
with non-negative mean-payoff to the problem of detecting paths with positive mean-payoff 
and a polynomial-time algorithm for mean-payoff pushdown graphs is obtained.

\smallskip\noindent{\em Notations.}
We shall use (i)~$\gamma$ or $\gamma_i$ for an element of $\Gamma$;
(ii)~$e$ or $e_i$ for a transition (equivalently an edge) from $E$;
(iii)~$\alpha$ or $\alpha_i$ for a string from $\Gamma^*$.
For a path $\pi = \atuple{c_1,  c_2, \dots} = \atuple{c_1 e_1 e_2 \dots}$ we denote by
(i)~$q_i$: the state of configuration $c_i$, and 
(ii)~$\alpha_i$: the stack string of configuration $c_i$.

\smallskip\noindent{\em Stack height and additional stack height of paths.}
For a path $\pi = \atuple {(\alpha_1,q_1),\dots,(\alpha_n,q_n)}$, the \emph{stack height}
of $\pi$ is the maximal height of the stack in the path, i.e., 
$\SH(\pi) = \max\Set{|\alpha_1|,\dots,|\alpha_n|}$.
The \emph{additional stack height} of $\pi$ is the additional height 
of the stack in the segment of the path, i.e., the additional 
stack height $\ASH(\pi)$ is 
$\SH(\pi) - \max\Set{|\alpha_1|,|\alpha_n|}$.

\smallskip\noindent{\em Pumpable pair of paths.}
Let $\pi = \atuple{c_1 e_1 e_2 \dots}$ be a finite or infinite path.
A \emph{pumpable pair of paths} for $\pi$ is a pair of non-empty sequence of edges:
$(p_1,p_2)=(e_{i_1}e_{i_1+1}\dots e_{i_1 + n_1},  e_{i_2}e_{i_2+1}\dots e_{i_2 + n_2})$, 
for $n_1,n_2\geq 0$, $i_1 \geq 0$ and $i_2 > i_1+n_1$ such that for every $j\geq 0$ the path
$\pi_{(p_1,p_2)}^j$ obtained by pumping the pair $p_1$ and $p_2$ of paths $j$ times each is a valid path, i.e., 
for every $j \geq 0$ we have
\[
\pi_{(p_1,p_2)}^j = \atuple{
c_1 e_1 e_2 \dots e_{i_1 - 1} (e_{i_1} e_{i_1 + 1}\dots e_{i_1 + i_n})^j e_{i_1 + i_n + 1} \dots e_{i_2 - 1} (e_{i_2} e_{i_2 + 1} \dots e_{i_2 + n_2})^j e_{i_2 + n_2} \dots}
\]
is a valid path.
We will show that large additional stack height implies the 
existence of a pumpable pair of paths. To prove the results we need the notion
of \emph{local minimum} of paths.

\smallskip\noindent{\em Local minimum of a path and non-decreasing paths.}
Let $\pi = \atuple{c_1, c_2, \dots}$ be a path.
A configuration $c_i = (\alpha_i,q_i)$ is a \emph{local minimum} (aka ``stair position'' in the 
literature) if for every $j\geq i$ we have $\alpha_i \sqsubseteq \alpha_j$ (i.e., the stack string 
$\alpha_i$ is a prefix string of $\alpha_j$), and we say that $\pi$ is a \emph{non-decreasing path} if $c_1$ is a local minimum.
One basic fact about local minima of a path is as follows: 
Every infinite path has infinitely many local minima.
We discuss the proof of the basic fact and some properties of local minima.
Consider a path $\pi=\atuple{c_1,c_2,\dots}$. 
If there is a finite integer $j$ such that from some point on 
(say after the $i$-th index) the stack height is always at least $j$, 
and the stack height is $j$ infinitely often,
then every configuration after the $i$-th index with stack height $j$ is a 
local minimum (and there are infinitely many of them).
Otherwise, for every integer $j$, there exists an index $i$, such that
for every index after $i$ the stack height exceeds $j$, and then 
for every $j$, the last configuration with stack height $j$ is a 
local minimum and we have infinitely many local minima. 
This shows the basic fact about infinitely many local minima of a 
path.
We now discuss a property of consecutive local minima in a path.
If we consider a path and the sequence of local minima, and let
$c_i$ and $c_j$ be two consecutive local minima. 
Then either $c_i$ and $c_j$ have the same stack height, or else
$c_j$ is obtained from $c_i$ with one push operation.
In the following lemma we establish that if the additional stack height 
of a path exceeds $(|Q|\cdot |\Gamma|)^2$, then there is a pumpable pair 
of paths.
Also note that if the additional stack height of a path is at least $d$, then it 
means that there are at least $d+1$ configurations in the path.

\begin{lem}\label{lemm:EveryDeepPathHasAPumpablePair}
Let $\pi$ be a finite path such that $\ASH(\pi) =d \geq (|Q|\cdot|\Gamma|)^2$.
Then $\pi$ has a pumpable pair of paths.
\end{lem}
\begin{proof}
We first select a subpath of $\pi$, denoted by $\pi^*$, such that 
$\pi^* = \atuple{c_1 ^*, \dots, c_p^*, \dots, c_n ^*}$ and the following conditions hold: 
(i)~$c_1^*$ is a local minimum in $\pi^*$, 
(ii)~$|\alpha_1^*| = |\alpha_n^*|$, and 
(iii)~$|\alpha_p^*| = |\alpha_1 ^*| + d$.
The subpath is selected as follows: consider a configuration 
$c_p^*$ in $\pi$ where the stack height is maximal, and $c_1^*$ is the 
closest configuration before $c_p^*$ where the stack height is exactly 
$d$ less than the stack height of $c_p^*$, and similarly $c_n^*$ is the 
closest configuration after $c_p^*$ where the stack height is exactly 
$d$ less than that of $c_p^*$ (see Figure~\ref{fig-mp1}).
Clearly all the three conditions are satisfied.
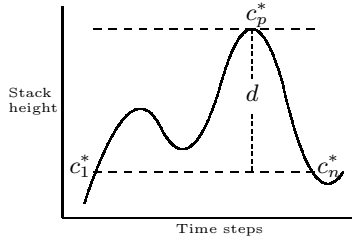
\begin{figure}[t]
\setlength{\unitlength}{0.4cm}
\begin{center}

\begin{picture}(10,8)(0,1)

\put(-0.5, 1){ \line(0, 1){7} }
\put(-0.5, 1){ \line(1, 0){9.5} }

\qbezier(0.5,1.5)(2,6)(3,4)

\qbezier(3,4)(4,2)(5,5.5)

\qbezier(5,5.5)(6,9)(7,5.5)

\qbezier(7,5.5)(8,1)(9,2.5)

\multiput(0.80,2.5)(0.5,0){15}{\line(1,0){0.25}}
\multiput(0.80,7.25)(0.5,0){15}{\line(1,0){0.25}}


\multiput(6,2.5625)(0.0,0.25){8}{\line(0,1){0.125}}

\multiput(6,5.625)(0.0,0.25){7}{\line(0,1){0.125}}

\put(5.8,4.875){$d$}

\put(5.8,7.55){$c_p^*$}

\put(0,2.5){$c_1^*$}
\put(8.15,2.5){$c_n^*$}

\put(-2,5){\tiny{Stack}}
\put(-2,4.5){\tiny{height}}

\put(3.5,0.5){\tiny{Time steps}}

\end{picture}
\end{center}

\caption{Subpath construction.}
  \label{fig-mp1}
\end{figure}

Let $c'_{j}$ (resp. $c''_{j}$) be the closest configuration before (resp. after) $c_p^*$ 
such that the stack height of $c'_{j}$ (resp. $c''_{j}$) is $|\alpha_1^*| + j$, 
for $j\geq 0$.
Since $d\geq (|Q|\cdot |\Gamma|)^2$ it follows from the pigeonhole principle that
there exist $j_1, j_2$ such that $j_1 < j_2$, and the states and the top stack symbol of $c'_{{j_1}}$, and 
$c'_{{j_2}}$ are identical, and the states and the top stack symbol of $c''_{{j_1}}$ and $c''_{{j_2}}$ 
are identical.
It is straightforward to verify that the sequence $p_1$ of edges between $c'_{{j_1}}$ and $c'_{{j_2}}$ 
along with the sequence $p_2$ of edges between $c''_{{j_2}}$ and $c''_{{j_1}}$ 
form a pumpable pair, 
i.e., $(p_1,p_2)$ is a pumpable pair for $\pi$.
\hfill\qed
\end{proof}

In the following lemma we establish the connection of additional stack height 
and the existence of a pumpable pair of paths with positive weights.

\begin{lem}\label{lem:HeavyPathWithBigDepthImpliesPosPump}
Let $c_1$ and $c_2$ be two configurations. If there exists $n \in \Z$ such that the minimal additional stack height (denoted by $d$) 
of the paths from $c_1$ to $c_2$ with weight at least $n$ is at least $(|Q|\cdot|\Gamma|)^2$, then there exists a path from $c_1$ to $c_2$ with 
additional stack height $d$ that contains a pumpable pair $(p_1, p_2)$ with $w(p_1) + w(p_2) >0$.
\end{lem}
\begin{proof}
Let $\pi$ be one of the shortest paths (in terms of length) from $c_1$ to $c_2$ with weight at least $n$ and additional stack height $d$. 
By Lemma~\ref{lemm:EveryDeepPathHasAPumpablePair}, the path $\pi$ has a 
pumpable pair $(p_1, p_2)$ of paths. 
Assume towards a contradiction that the weight of the pair is not positive 
(i.e., consider that  $w(p_1) + w(p_2) \leq 0$). 
Then, we can remove the pair and obtain a path $\pi' = \pi^0_{(p_1, p_2)}$ from $c_1$ to $c_2$ with $w(\pi ') \geq w(\pi)$. 
The path $\pi'$ is shorter in length than $\pi$, since either $p_1$ or $p_2$ is not an empty path.
Moreover, the additional stack height of $\pi '$ is at most $d$.
\NewChange{If $\ASH(\pi ') = d$, then this 
this yields a contradiction to the assumption that the length of $\pi$ is 
minimal.
Otherwise, if $\ASH(\pi ') < d$, we have a contradiction to the minimality of $d$.}
\hfill\qed
\end{proof}

\smallskip\noindent{\em Ignoring finite plays.} 
For technical convenience, we will only consider infinite plays, and 
consider that finite plays do not satisfy the mean-payoff objective.
Thus if there are no transitions from a state, then we consider it as a 
losing sink state (a state with a self-loop with negative weight).

\subsection{Objectives $\LimInfAvg> 0$ and $\LimSupAvg> 0$}\label{subsec:wps1}
In this section we consider limit-average (or mean-payoff) objectives
with strict inequality. 
We show that WPSs with such objectives can be solved in polynomial
time. 
A crucial concept in the proof is the notion of good cycles, and we define 
them below.

\smallskip\noindent{\bf Good cycle.}
A finite path $\pi =\atuple{c_1, \dots, c_n}$ is a \emph{good cycle} if the following 
conditions hold:
\begin{enumerate}
\item $w(\pi) > 0$ (the weight of the path is positive); and
\item $\pi$ is a non-decreasing path; and
\item let $c_1 = (\alpha_1,q_1)$ and $c_n = (\alpha_n, q_n)$, then $q_1 = q_n$ and $\Top(\alpha_1) = \Top(\alpha_n)$.
\end{enumerate}

We first prove two lemmas and the intuitive descriptions of them are 
as follows: 
In the first lemma (Lemma~\ref{lemm:BigWeightMeanPosCycle}) we show that for 
every WPS, for every natural number $d$, there exists a natural number $n$ 
such that if there is a path with weight at least $n$ and additional
stack height at most $d$, then there is a good cycle in the WPS.
The second lemma (Lemma~\ref{lemm:MaxDepthBigThenWeightUnbounded}) is similar to the first lemma,
and shows that if the additional stack height of the path is large, 
then it is possible to construct paths with arbitrarily 
large weights.
Using the above two lemmas we then show that if the 
weight of a finite path is sufficiently large, then either 
a good cycle exists or paths with arbitrarily large weights
can be constructed (Lemma~\ref{lem:BigWeightImpliesOmegaOrCycle}). 
Finally we prove the key lemma (Lemma~\ref{lem:GoodCycleImpliesGoodMP}) that 
establishes the equivalence of the existence of a path satisfying mean-payoff 
objectives with strict inequality and the existence of a good cycle.

\smallskip\noindent{\em Restrictions to reachable configurations.}
In the sequel of this whole section we will only consider reachable 
configurations and reachable paths, and not explicitly 
mention the reachability property.
For good cycles, we often mention the reachable property explicitly.

\begin{lem}\label{lemm:BigWeightMeanPosCycle}
Let $\wps$ be a WPS.
For every $d\in\Nat$ there exists $n_{\wps,d} \in \Nat$ such that the 
following assertion holds:
If there exists a non-decreasing path $\pi = \atuple{c_1, \dots, c_r}$ such that (i)~$w(\pi) \geq n_{\wps,d}$ and (ii)~$\ASH(\pi) \leq d$;
then $\wps$ has a reachable good cycle.
\end{lem}
\begin{proof}
Let us define $n_{\wps,d} = W\cdot(|Q|\cdot|\Gamma|^{d + |Q|\cdot|\Gamma|}+1)$.
We prove the result by induction over $\SH(c_r)-\SH(c_1)$. 
Note that since $\pi$ is a non-decreasing path, we have $\SH(c_r)-\SH(c_1) \geq 0$.

\smallskip\noindent{\em Base case.}
In the base case we prove the result for $\SH(c_r)-\SH(c_1) \leq |Q|\cdot |\Gamma|$.
Let $G_d$ be a graph that contains all the configurations that occur in one of the non-decreasing paths from $c_1$ to $c_r$, 
for which $c_1$ is a local minimum and that has additional stack height at most $d$. 
The set of edges of $G_d$ (and their weights) is induced by the set of transitions in $\wps$.
Note that $\ASH(\pi) \leq d$ and $\SH(c_r)-\SH(c_1) \leq |Q|\cdot |\Gamma|$ implies that $G_d$ contains only configurations with stack height between $\SH(c_1)$ and $\SH(c_1) + d + |Q|\cdot |\Gamma|$.
Hence the graph  $G_d$ is a finite graph with at most $|Q|\cdot|\Gamma|^{d + |Q|\cdot|\Gamma|}$ states, and the maximal absolute weight is at most $W$ (the maximal absolute weight of $\wps$).
A reachable positive cycle in $G_d$ implies the existence of a reachable good cycle in $\wps$,
and if no positive cycle is reachable, then the weight of each path is bounded by $(|G_d| + 1) \cdot W$. 
Thus with \NewChange{$n_{\wps,d}= W\cdot(|Q|\cdot|\Gamma|^{d + |Q|\cdot|\Gamma|}+1) \geq (|G_d| + 1) \cdot W$} we obtain the desired result.

\smallskip\noindent{\em Inductive case.}
We now prove the result for $\SH(c_r)-\SH(c_1) > |Q|\cdot |\Gamma|$.
Since $\SH(c_r)-\SH(c_1) > |Q|\cdot |\Gamma|$, then $\pi$ has at least $|Q|\cdot |\Gamma| + 1$ local minima with different stack heights.
Hence, by the pigeonhole principle, there must be two configurations in $\pi$, namely $c_i$ and $c_j$ (for $1\leq i \leq j\leq r$), such that $c_i$ and $c_j$ are local minima with different stack heights and $\Top(c_i) = \Top(c_j)$.
We denote by $\pi_1$ the sequence of transitions that is induced by $\pi$ from $c_1$ to $c_i$, by $\pi_2$ the sequence of transitions from $c_i$ to $c_j$, 
and by $\pi_3$ the sequence of transitions from $c_j$ to $c_r$.
We note that by definition $\pi = \pi_1\pi_2\pi_3$.
Moreover, since $c_i$ and $c_j$ are local minima, the path $\pi_1\pi_3$ is a valid path that begins in configuration $c_1$ and ends in some configuration $c_\ell$ such that $\SH(c_\ell) - \SH(c_1) < \SH(c_r)-\SH(c_1)$.
Hence, by the induction hypothesis, if $w(\pi_1\pi_3)\geq n_{\wps,d}$, then $\wps$ has a good cycle.
Otherwise, if $w(\pi_1\pi_3) < n_{\wps,d}$ and $w(\pi) \geq n_{\wps,d}$, then it must be the case that $w(\pi_2) > 0$.
Thus, $\wps$ has a good cycle, namely, the path $\pi_2$.
The desired result follows.
\hfill\qed
\end{proof}


\begin{lem}\label{lemm:MaxDepthBigThenWeightUnbounded}
Let $\wps$ be a WPS.
Let $n\in\Z$ and let $\pi = \atuple{c_1, \dots, c_r}$ be a non-decreasing
path with weight at least $n$, with minimal additional stack height among all 
paths from $c_1$ to $c_r$ with weight at least $n$.
If $\ASH(\pi)\geq (|Q|\cdot |\Gamma|)^2$, then for every $m\in\Nat$ 
there exists a non-decreasing path $\pi_m$ from $c_1$ to $c_r$ with $w(\pi_m) \geq m$.
\end{lem}
\begin{proof}
By Lemma~\ref{lem:HeavyPathWithBigDepthImpliesPosPump} there exists a path  $\ov{\pi}$ 
from $c_1$ to $c_r$ that has a pumpable pair $(p_1,p_2)$ such that $w(p_1) + w(p_2) > 0$.
Hence for every $i\in\Nat$ we get that $w(\ov{\pi}_{(p_1,p_2)}^{i+1}) > w(\ov{\pi}_{(p_1,p_2)}^{i})$
(i.e., the weight after pumping $i+1$ times the pair of paths exceeds the 
weight of pumping $i$ times).
Hence for $i = m - w(\ov{\pi})$ we get that 
$w(\ov{\pi}_{(p_1,p_2)}^{i}) \geq m$. 
The desired result follows.
\hfill\qed
\end{proof}

\begin{lem}\label{lem:BigWeightImpliesOmegaOrCycle}
Let $\wps$ be a WPS.
There exists $n_{\wps} \in \Nat$ such that if there 
exists a non-decreasing path $\pi$ from configuration $c_1$ to configuration 
$c_r$ and $w(\pi) \geq n_{\wps}$, then one of the following conditions holds:
\begin{enumerate}
\item The WPS $\wps$ has a reachable good cycle.
\item For every $n '\in \Nat$ there exists a non-decreasing path $\pi'$ from $c_1$ to $c_r$ with $w(\pi') > n'$.
\end{enumerate}
\end{lem}
\begin{proof}
Observe that the number $n_{\wps}$ is of our choice and we will choose it 
sufficiently large for the proof.
Let $d^*= (|Q|\cdot|\Gamma|)^2$, and our choice of 
$n_{\wps}$ is $|Q|\cdot |\Gamma|\cdot n_{\wps,d^*}$ (where $n_{\wps,d^*}$ is as defined in 
Lemma~\ref{lemm:BigWeightMeanPosCycle}). 
Let $\pi = \atuple{c_1, c_2, \dots, c_r}$ be a path such that $c_1$ is a local minimum
and $w(\pi)\geq n_{\wps}$.
Let $m_1,\dots, m_j$ be the local minima along the path.
Note that $m_1 = c_1$ and $c_r = m_j$.
Also note that $j \geq |\alpha_r| - |\alpha_1|$.
Note that if $m_{i_1} = (\alpha_{i_1}\gamma,q)$ and $m_{i_2} = (\alpha_{i_2}\gamma,q)$ (for some $\gamma\in\Gamma$),
then if a good cycle does not exist we get that the weight of the path between 
$m_{i_1}$ and $m_{i_2}$ is not positive.
Hence, since $Q$ and $\Gamma$ are finite, 
either a good cycle exists (by the pigeonhole principle) 
or there exists $m_i, m_{i+1}$ such that 
$\alpha_{i+1} = \alpha_i \gamma$ for some $\gamma\in\Gamma\cup\Set{\epsilon}$ (where $\epsilon$ denotes the empty string) and there exists 
a path from $m_i$ to $m_{i+1}$ such that $m_i$ is a local minimum and the 
weight of the path is at least 
$n_{\wps,d^*}$ 
(since the longest sequence of local minimum configurations that do not contain a cycle is of length at most $|Q|\cdot|\Gamma|$, and 
there is a sequence of acyclic configurations that has a weight of at least $n_{\wps}$).
Let $\pi ^*$ be such a path with minimal additional stack height between 
$m_i$ and $m_{i+1}$. 
We consider two cases to complete the proof.
\begin{enumerate}
\item 
If the additional stack height of $\pi^*$ is smaller than $d^*$, then by 
Lemma~\ref{lemm:BigWeightMeanPosCycle} we have a reachable 
good cycle from $m_i$ and since $m_i$ is reachable from $c_1$ 
we have reachable good cycle from $c_1$ (condition~1 of the lemma holds). 
\item 
If the additional stack of $\pi^*$ is at least $d^*$, then 
by Lemma~\ref{lemm:MaxDepthBigThenWeightUnbounded} for every $n_0$
we can construct a path $\pi_{n_0}$ between $m_i$ and $m_{i+1}$ 
with weight $w(\pi_{n_0})$ at least $n_0$, and $m_1$ is a local minimum of 
$\pi_{n_0}$.
For $n'\in \Nat$, let $n_0=n'+ W\cdot |\pi|$, and let $\pi'$ be the path 
constructed using the segment from $c_1$ to $m_i$, then the path $\pi_{n_0}$, 
and then the segment of $\pi$ from $m_{i+1}$ to $c_r$.
The configuration $c_1$ is a local minimum of $\pi'$ and the weight of 
$\pi'$ is at least $n_0 -W\cdot |\pi| \geq n'$.
Hence it follows that for every $n'$ we can construct a path 
from $c_1$ to $c_r$ with $c_1$ as a local minimum and weight at least $n'$ 
(condition~2 of the lemma holds).
\end{enumerate}
This completes the proof of the lemma.
\hfill\qed
\end{proof}

\begin{lem} \label{lem:GoodCycleImpliesGoodMP}
Let $\wps$ be WPS.
The following statements are equivalent:
(i)~There exists a path $\pi_1$ with $\LimSupAvg(\pi_1) > 0$;
(ii)~there exists a path $\pi_2$ with $\LimInfAvg(\pi_2) > 0$; and
(iii)~there exists a path $\pi$ that contains a good cycle.
\end{lem}
\begin{proof}
The direction from right to left (i.e., (iii)$\Rightarrow$(ii)$\Rightarrow$(i))
is immediate.
Let $\pi = \pi_1 \pi_2$ be a finite path in $\wps$ such that $\pi_2$ is a good cycle.
Let $\pi_1 = c_1 e^1_1 e^1_2 \dots e^1_{n_1}$ and $\pi_2 = c_2 e^2_1 e^2_2 \dots e^2_{n_2}$.
The infinite path $\pi ' = \pi_1 c_2 (e^2_1 e^2_2 \dots e^2_{n_2})^\omega$ 
obtained by repeating the good cycle forever is a valid path which 
witnesses that $\LimSupAvg(\pi ') \geq \LimInfAvg(\pi ') > 0$.

In order to prove the opposite direction, we consider 
an infinite path $\pi$ such that $\LimSupAvg(\pi) > 0$.
Let $q\in Q$ and $\gamma\in\Gamma$ be such that the sequence
$m_1 = (\alpha_{i_1},q) , m_2 = (\alpha_{i_2},q), \dots$ is an infinite 
sequence of local minima of $\pi$ and $\Top(\alpha_{i_j}) = \gamma$ 
(note that such state and symbol are guaranteed to exist due to the existence 
of infinitely many local minima and finiteness of $Q$ and $\Gamma$).
If there exists $j > 1$ such that $w(\pi[i_1,i_j]) > 0$ then by definition 
$\pi[i_1,i_j]$ is a good cycle and the result follows.
Otherwise let us assume that for every $j > 1$ we have $w(\pi[i_1,i_j]) \leq 0$.
As $\LimSupAvg(\pi) > 0$ it follows that for every $n^*\in\Nat$ there exists 
$n \in\Nat$ with $i_n > 1$ such that the path $\pi[i_1,i_n]$ contains a prefix 
with weight at least $n^*$ (otherwise $\LimSupAvg(\pi) \leq 0$).
We now use Lemma~\ref{lem:BigWeightImpliesOmegaOrCycle} to complete the proof.
Let $n^* = n_{\wps}$ (where $n_{\wps}$ is as used in Lemma~\ref{lem:BigWeightImpliesOmegaOrCycle}).
Let $\pi ' = m_1 ,\dots, c^*$ be the prefix of $\pi[i_1,i_n]$ such that 
$w(\pi ') \geq n^*$.
If the first condition of Lemma~\ref{lem:BigWeightImpliesOmegaOrCycle} holds 
(i.e., $\wps$ has a good cycle), then we are done with the proof.
Otherwise, by condition~2 of Lemma~\ref{lem:BigWeightImpliesOmegaOrCycle} it follows 
that for every $n_0\in\Nat$ there exists a path $\pi_{n_0}$ from $m_1$ to $c^*$ such that 
$m_1$ is a local minimum and $w(\pi_{n_0}) \geq n_0$.
Let us choose $n_0 = W \cdot  |\pi[i_1,i_n]| + 1$.
Then consider the path $\ov{\pi}=\pi_{n_0} \pi[i+|\pi '|,i_n]$ that is obtained by 
concatenating the witness 
path $\pi_{n_0}$ for $n_0$ from $m_1$ to $c^*$, and then the part of $\pi$ from $c^*$ to 
$\pi[i_n]$.
For the path $\ov{\pi}$ we have (i)~the sum of weights is at least $n_0 - W\cdot|\pi[i_1,i_n]| \geq 1>0$;
(ii)~$\pi[i_1]$ is a local minimum; and 
(iii)~the state and the top stack symbol of $\pi[i_1]$ and $\pi[i_n]$ are the same.
Thus $\ov{\pi}$ is a witness good cycle.
For conclusion we get that if $\LimSupAvg(\pi) > 0$, then there exists a good cycle, which also 
implies that there exists a path $\pi'$ such that $\LimInfAvg(\pi') > 0$.
This concludes the proof of the lemma.
\hfill\qed
\end{proof}

In the above key lemma we have established the equivalence of the decision 
problems for WPSs with mean-payoff objectives with strict inequality
and the problem of determining the existence of good cycles. 
We will now present a polynomial-time algorithm for detecting good cycles.
To this end we introduce the notion of non-decreasing $\alpha$-paths
and summary functions.

\smallskip\noindent{\em Non-decreasing $\alpha$-paths.}
A path from a configuration $(\alpha\gamma,q_1)$ to a configuration 
$(\alpha\gamma\alpha_2,q_2)$ is a \emph{non-decreasing $\alpha$-path} 
if $(\alpha\gamma,q_1)$ is a local minimum.
Note that if $\pi$ is a non-decreasing $\alpha$-path for some 
$\alpha\in\Gamma^*$, then the same sequence of transitions leads to a 
non-decreasing $\beta$-path for every $\beta\in\Gamma^*$.
Hence we say that $\pi$ is a non-decreasing path if there exists 
$\alpha \in \Gamma^*$ such that $\pi$ is a non-decreasing $\alpha$-path.

\smallskip\noindent{\em Summary function.}
Let $\wps$ be a WPS. 
For $\alpha\in \Gamma^*$ we define $s_{\alpha} : Q \times \Gamma \times Q \to \Set{-\infty} \cup \Z \cup \Set{\omega}$ as follows.
\begin{enumerate}
\item $s_\alpha(q_1,\gamma,q_2) = \omega$ iff for every $n\in\Nat$ there 
exists a non-decreasing path from $(\alpha \gamma, q_1)$ to $(\alpha \gamma, q_2)$ 
with weight at least $n$.
\item $s_\alpha(q_1,\gamma,q_2) = z\in\Z$ iff the weight of the maximal-weight 
non-decreasing path from configuration $(\alpha \gamma, q_1)$ to 
configuration $(\alpha \gamma, q_2)$ is $z$.
\item $s_\alpha(q_1,\gamma,q_2) = -\infty$ iff there is no non-decreasing path 
from $(\alpha \gamma, q_1)$ to $(\alpha \gamma, q_2)$. 
\end{enumerate}

\begin{rem}\label{rem:IndependentOfAlpha}
For every $\alpha_1, \alpha_2\in \Gamma^*$: $s_{\alpha_1} \equiv s_{\alpha_2}$.
\end{rem}
Due to Remark~\ref{rem:IndependentOfAlpha} it is enough to consider only $s \equiv s_{\bot}$.
The computation of the summary function will be achieved by considering the
stack height bounded summary functions defined below.

\smallskip\noindent{\em Stack height bounded summary function.}
For every $d\in\Nat$, the \emph{stack height bounded summary function} 
$s_d : Q \times \Gamma \times Q \to \Set{-\infty} \cup \Z \cup \Set{\omega}$ 
is defined as follows: 
(i)~$s_d(q_1,\gamma,q_2) = \omega$ iff for every $n\in\Nat$ there exists a 
non-decreasing path from $(\bot\gamma, q_1)$ to $(\bot \gamma, q_2)$ with 
weight at least $n$ and additional stack height at most $d$;
(ii)~$s_d(q_1,\gamma,q_2) = z$ iff the weight of the maximal-weight 
non-decreasing path from $(\bot\gamma,q_1)$ to $(\bot\gamma,q_2)$ with 
additional stack height at most $d$ is $z$; and
(iii)~$s_d(q_1,\gamma,q_2) = -\infty$ iff there is no non-decreasing path 
with additional stack height at most $d$
from $(\bot \gamma, q_1)$ to $(\bot \gamma, q_2)$.

\smallskip\noindent{\em Basic facts of summary functions.} 
We have the following basic facts:
(i)~for every $d\in\Nat$, we have $s_{d+1} \geq s_d$ (monotonicity); and
(ii)~$s_{d+1}$ is computable in polynomial time from $s_d$ and $\wps$ (we will show this
fact in Lemma~\ref{lem:ShortCutIsComputableInPoly}).
We first present a lemma that shows that from $s_d$, with $d=(|Q|\cdot |\Gamma|)^2$, 
we obtain the values of function $s$ for all values in $\Z \cup\Set{-\infty}$.

\begin{lem}\label{lemm:AlmostS}
Let $d = (|Q|\cdot |\Gamma|)^2$.
For all $q_1, q_2 \in Q$ and $\gamma \in \Gamma$,
if $s(q_1,\gamma,q_2) \in \Z \cup \Set{-\infty}$, then $s(q_1,\gamma,q_2) = s_d(q_1,\gamma,q_2)$.
\end{lem}
\begin{proof}
By definition we have $s(q_1,\gamma,q_2) \geq s_d(q_1,\gamma,q_2)$.
Towards a contradiction, we assume that 
$s(q_1,\gamma,q_2) > s_d(q_1,\gamma,q_2)$.
By the assumption there exists a non-decreasing 
path $\pi$ with minimal additional stack height from $(\bot \gamma,q_1)$ to 
$(\bot \gamma,q_2)$ with weight $n > s_d(q_1,\gamma,q_2)$ and additional stack 
height $d ' > (|Q|\cdot |\Gamma|)^2 $.
Hence by Lemma~\ref{lemm:MaxDepthBigThenWeightUnbounded} for every 
$m \in\Nat$ there exists a non-decreasing path from $(\bot \gamma,q_1)$ to 
$(\bot \gamma,q_2)$ with weight at least $m$ (note that in 
Lemma~\ref{lemm:MaxDepthBigThenWeightUnbounded} the witness 
path constructed by pumping the positive pumpable pair yields a non-decreasing path).
Hence $s(q_1,\gamma,q_2) = \omega$ in contradiction to the assumption 
that $s(q_1,\gamma,q_2) \in \Z \cup \Set{-\infty}$.
The desired result follows.
\hfill\qed
\end{proof}

Our goal now is the computation of the $\omega$ values of the summary function.
To achieve the computation of $\omega$ values we will define another summary 
function $s^*$ and a new WPS $\wps^*$ such that certain cycles in $\wps^*$ will 
characterize the $\omega$ values of the summary function. 
We now define the summary function $s^*$ and the pushdown system $\wps^*$. 
Let $d= (|Q|\cdot|\Gamma|)^2$. The new summary function $s^*$ is defined as follows:
if the values of $s_d$ and $s_{d+1}$ are the same then it is assigned the value 
of $s_d$, and otherwise the value $\omega$. Formally, for all states 
$q_1,q_2 \in \wps$ and a stack symbol $\gamma$,
\[ 
s^*(q_1,\gamma,q_2) = 
   \left\{ \begin{array}{ll}
	s_d(q_1,\gamma,q_2) & \mbox{   if $s_d(q_1,\gamma,q_2) = s_{d+1}(q_1,\gamma,q_2)$} \\
	\omega & \mbox{   if $s_d(q_1,\gamma,q_2) < s_{d+1}(q_1,\gamma,q_2)$}.
   \end{array} \right. 
\]
The new WPS $\wps^*$ is constructed from $\wps$ by adding the following set of $\omega$-edges:
$\Set{(q_1,\gamma,q_2,\Skip) \mid s^*(q_1,\gamma,q_2) = \omega}$.
Note that $s^*$ is a summary function for $\wps$, but not necessarily for 
$\wps^*$.

\begin{lem}\label{lemm:IfOmegaThenOmega}
For all $q_1, q_2 \in Q$ and $\gamma \in \Gamma$,
the following assertion holds:
the original summary function satisfies $s(q_1,\gamma,q_2) = \omega$ iff 
there exists a non-decreasing path in $\wps^{*}$ from 
$(\bot\gamma,q_1)$ to $(\bot\gamma,q_2)$ that goes through an $\omega$-edge.
\end{lem}
\begin{proof}
The direction from right to left is easy: if there is a non-decreasing
path in $\wps^*$ that goes through an $\omega$-edge, it means that 
there exists $(q_1',\gamma',q_2')$ with either $s_d(q_1',\gamma,q_2') =\omega$ 
or $s_d(q_1',\gamma',q_2')< s_{d+1}(q_1',\gamma',q_2')$.
If $s_d(q_1',\gamma,q_2') =\omega$, then clearly $s(q_1',\gamma,q_2') =\omega$.
Otherwise we have $s_d(q_1',\gamma',q_2')< s_{d+1}(q_1',\gamma',q_2')$, 
and then by Lemma~\ref{lemm:AlmostS} we get that 
$s(q_1',\gamma',q_2')=\omega$.
Since there exists a finite path from $(\bot\gamma,q_1)$ to $(\bot\gamma,q_2)$ 
with an $\omega$-edge it follows that $s(q_1,\gamma,q_2)=\omega$.

For the converse direction, we consider the case that
$s(q_1,\gamma,q_2) = \omega$.
If $s^*(q_1,\gamma,q_2) = \omega$, then the proof follows immediately.
Otherwise it follows that $s_d(q_1,\gamma,q_2) \in\Z$.
Hence there exists a weight $n\in\Z$ such that a non-decreasing path with minimal additional 
stack height with weight $n$ has additional stack height $d' \geq d + 1$.
Let $\pi$ be such a path.
Then there exists a non-decreasing subpath that starts at 
$(\alpha\gamma',q_1')$ and ends at $(\alpha\gamma',q_2')$ with additional stack 
height exactly $d+1$ (for some states $q_1',q_2'$ and stack symbol $\gamma'$).
If $s_{d+1}(q_1',\gamma',q_2') = s_{d}(q_1',\gamma',q_2')$, 
then $\pi$ is not a path with the minimal additional stack height.
Hence, as $s_{d+1}(q_1',\gamma',q_2') > s_{d}(q_1',\gamma',q_2')$, 
by definition $s^*(q_1',\gamma',q_2') = \omega$ and the proof follows. 
\hfill\qed
\end{proof}

We are now ready to show that the summary function $s$ can be computed 
polynomial time.

\begin{remark}
We show that the number of arithmetic operations required is polynomial in the 
size of the WPS, and hence the polynomial time 
bound follows.
In the sequel, instead of polynomial number of operations in the size of the 
WPS we simply write polynomial time.
\end{remark}

\begin{lem}\label{lem:ShortCutIsComputableInPoly}
For a WPS $\wps$,
the summary function $s$ is computable in polynomial time.
\end{lem}
\begin{proof}
There are two key steps of the proof: 
(i)~computation of $s_d$, for $d=(|Q|\cdot |\Gamma|)^2$, and 
we will argue how to compute $s_{i+1}$ from $s_{i}$ in polynomial time; and
(ii)~computation of a non-decreasing path in $\wps^*$ that goes through an 
$\omega$-edge. 
We first argue how the key steps give us the desired result and then present
the details of the key steps.
Given the computation of (i), we construct $s_d$, $s_{d+1}$ in polynomial 
time, and hence also $s^*$.
Given $s^*$ we construct $\wps^*$ in polynomial time. 
By computation (ii) we can assign the $\omega$ values for the summary 
function, and all other have values as defined by $s_d$.
Thus with the computation of key steps (i) and (ii) in polynomial time, 
we can compute the summary function $s$ in polynomial time.
We now describe the key steps:

\begin{enumerate}

\item \emph{Computation of $s_{i+1}$ from $s_{i}$ and $\wps$.}
Let $G_{\wps}$ be the finite weighted graph that is formed by all the 
configurations of $\wps$ with stack height either zero, one or two, that is, 
the vertices are of the form $(\alpha, q)$ where $q\in Q$ and $\alpha \in 
\Set{\bot, \bot \gamma, \bot  \gamma_1  \gamma_2 
\mid \gamma, \gamma_1, \gamma_2 \in \Gamma}$.
The edges (and their weights) are according to the transitions of $\wps$: 
formally, 
(i)~(Skip edges): for vertices $(\bot  \alpha, q)$ we have an edge 
to $(\bot  \alpha,q')$ iff $e=(q, \Top(\alpha), \Skip,q')$ is an edge
in $\wps$ (and the weight of the edge in $G_{\wps}$ is $w(e)$) where 
$\alpha=\gamma$ or $\alpha=\gamma_1 \gamma_2$ for 
$\gamma,\gamma_1,\gamma_2 \in \Gamma$; 
(ii)~(Push edges): for vertices $(\bot  \gamma, q)$ we have an edge to 
$(\bot  \gamma  \gamma', q')$ iff 
$e=(q,\gamma, \Push(\gamma'),q')$ is an edge in $\wps$ 
(and the weight of the edge in $G_{\wps}$ is $w(e)$)  
for $\gamma,\gamma'\in \Gamma$; and
(iii)~(Pop edges): for vertices $(\bot  \gamma  \gamma', q)$
we have an edge to $(\bot  \gamma, q')$ iff 
$e=(q,\gamma', \Pop, q')$ is an edge in $\wps$ 
(and the weight of the edge in $G_{\wps}$ is $w(e)$)  
for $\gamma,\gamma'\in \Gamma$.
Intuitively, $G_{\wps}$ allows skips, push pop pairs, and only one 
additional push.
Note that $G_{\wps}$ has at most $3\cdot |Q|\cdot |\Gamma|^2$ vertices, 
and can be constructed in polynomial time.

For every $i \geq 1$, given the function $s_i$, the graph $G_{\wps}^i$ is 
constructed from $G_{\wps}$ as follows (e.g., see Figure~\ref{fig:example-gi}): adding edges 
$((\bot \gamma_1 \gamma_2, q_1), (\bot \gamma_1 \gamma_2, q_2))$ (if 
the edge does not exist already) and changing its weight to 
$s_i(q_1,\gamma_2,q_2)$ for every $\gamma_1, \gamma_2 \in \Gamma$ and 
$q_1, q_2 \in Q$.
The value of $s_{i+1}(q_1,\gamma,q_2)$ is exactly the weight of a  
maximal-weight path between $(\bot \gamma, q_1)$ and $(\bot \gamma, q_2)$ in 
$G_{\wps}^i$ (with the following convention: $-\infty < z < \omega$, 
$z + \omega = \omega$ and $z + -\infty = \omega + -\infty = -\infty$ 
for every $z\in\Z$). 
If in $G_{\wps}^i$ there is a path from $(\bot \gamma, q_1)$ to 
$(\bot \gamma, q_2)$ that contains a cycle with positive weight,
then we set $s_{i+1}(q_1,\gamma,q_2) = \omega$.
Hence, given $s_i$ and $\wps$, the construction of $G_{\wps}^i$ is 
achieved in polynomial time, and the computation of $s_{i+1}$ 
is achieved using the Bellman-Ford algorithm~\cite{CLRS-Book} in 
polynomial time 
(a maximal-weight path is a shortest-weight path if we define the edge length 
as the negative of the edge weight).
Also note that the Bellman-Ford algorithm reports cycles with positive weight 
(that is, negative length) which is required to set $\omega$ values of 
$s_{i+1}$.
It follows that we can compute $s_{i+1}$ given $s_i$ and $\wps$ in 
polynomial time.
In order to compute $s_0$ we run the Bellman-Ford algorithm over the graph $G_{\wps}^0$ in which all the push and pop transitions are disabled.
We note that the number of vertices in $G_{\wps}^0$ is at most $|\wps|$.
Hence, the computation is polynomial.

\item \emph{Non-decreasing $\omega$-edge path in $\wps^*$.}
We reduce the problem of checking if there exists a non-decreasing path from 
$(\bot\gamma, q_1)$ to $(\bot\gamma, q_2)$ in $\wps^*$ that goes through an 
$\omega$-edge to the problem of pushdown reachability in pushdown systems 
(or pushdown graphs), which is known to be in 
PTIME~\cite{Yan90,ABEGRY05}.
The reduction is as follows: for every state $q \in Q$ we add a fresh (new) 
state $q^\omega$, add a transition (or edge) 
$(q_1^\omega, \gamma, q_2^\omega, \mathit{com})$ for every 
$(q_1, \gamma, q_2, \mathit{com}) \in \Delta$ (i.e., the freshly added 
states follow the transition in the fresh copy as in the original WPS), and 
a transition $(q_1, \gamma, q_2^\omega, \mathit{com})$ for every transition 
$(q_1, \gamma, q_2, \mathit{com})$ that has an $\omega$ weight (i.e., 
there is a transition to the fresh copy only for an $\omega$-edge).
It follows that there exists an $\omega$-edge non-decreasing path in 
$\wps^*$ from $(\bot \gamma,q_1)$ to $(\bot \gamma,q_2)$ iff the configuration 
$(\bot \gamma,q_2^\omega)$ is pushdown reachable from the configuration 
$(\bot \gamma,q_1)$.
Hence it follows that the existence of a non-decreasing $\omega$-edge path in 
$\wps^*$ can be determined in polynomial time.



\end{enumerate}
The desired result follows.
\hfill\qed
\end{proof}


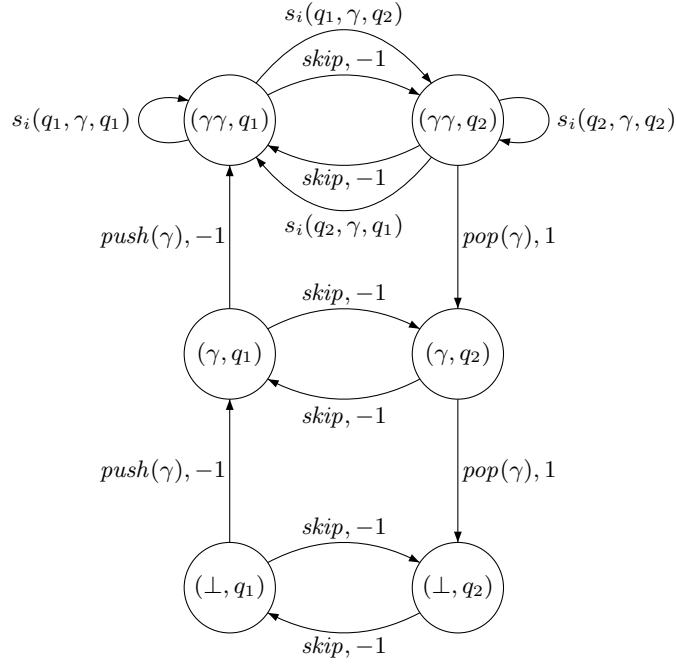
\begin{figure}[!tb]
\begin{center}
\begin{picture}(48,80)(0,-32)

\node[Nmarks=n, iangle=180,Nw= 12,Nh=12,Nmr=6](n0b)(10,-26){$(\bot,q_1$)}
\node[Nmarks=n, Nw= 12,Nh=12,Nmr=6](n1b)(40,-26){$(\bot,q_2)$}

\node[Nmarks=n, iangle=180,Nw= 12,Nh=12,Nmr=6](n0)(10,5){$(\gamma,q_1$)}
\node[Nmarks=n, Nw= 12,Nh=12,Nmr=6](n1)(40,5){$(\gamma,q_2)$}

\node[Nmarks=n, iangle=180, Nw= 12,Nh=12,Nmr=6](n0g)(10,36){$(\gamma\gamma,q_1$)}
\node[Nmarks=n, Nw= 12,Nh=12,Nmr=6](n1g)(40,36){$(\gamma\gamma,q_2)$}


\drawloop[ELside=l,loopCW=y, loopdiam=6, loopangle=180](n0g){$s_i(q_1,\gamma,q_1)$}
\drawloop[ELside=l,loopCW=y, loopdiam=6,loopangle=0](n1g){$s_i(q_2,\gamma,q_2)$}

\drawedge[ELpos=50, ELside=l,ELdist=0.5](n0b,n0){$\Push(\gamma),-1$}

\drawedge[ELpos=50, ELside=l,ELdist=0.5](n0,n0g){$\Push(\gamma),-1$}

\drawedge[ELpos=50, ELside=l,ELdist=0.5](n1g,n1){$\Pop(\gamma),1$}
\drawedge[ELpos=50, ELside=l,ELdist=0.5](n1,n1b){$\Pop(\gamma),1$}

\drawedge[ELpos=50, ELside=l, ELdist=0.5, curvedepth=6](n0,n1){$\Skip,-1$}
\drawedge[ELpos=50, ELside=l, curvedepth=6](n1,n0){$\Skip,-1$}

\drawedge[ELpos=50, ELside=l, ELdist=0.5, curvedepth=6](n0g,n1g){$\Skip,-1$}
\drawedge[ELpos=50, ELside=l, ELdist=0.5, curvedepth=6](n1g,n0g){$\Skip,-1$}

\drawedge[ELpos=50, ELside=l, ELdist=0.5, curvedepth=6](n0b,n1b){$\Skip,-1$}
\drawedge[ELpos=50, ELside=l, ELdist=0.5, curvedepth=6](n1b,n0b){$\Skip,-1$}

\drawedge[ELpos=50, ELside=l, ELdist=0.5, curvedepth=12](n0g,n1g){$s_i(q_1,\gamma,q_2)$}
\drawedge[ELpos=50, ELside=l, ELdist=0.5, curvedepth=12](n1g,n0g){$s_i(q_2,\gamma,q_1)$}

\end{picture}
\caption{$G^i_\wps$  that corresponds to the WPS $\wps$ from Figure~\ref{fig:ex-wps}.
The reader should note that this is a finite graph.
We explicitly label some of the transitions by $\Pop,\Push$ and $\Skip$ only to simplify the illustration.}\label{fig:example-gi}
\end{center}
\end{figure}


\begin{figure}[!tb]
\begin{center}
\begin{picture}(48,55)(0,-5)
\node[Nmarks=n, iangle=180,Nw= 12,Nh=12,Nmr=6](n0)(10,5){$(\bot,q_1$)}
\node[Nmarks=n, Nw= 12,Nh=12,Nmr=6](n1)(40,5){$(\bot,q_2)$}

\node[Nmarks=n, iangle=180, Nw= 12,Nh=12,Nmr=6](n0g)(10,36){$(\gamma,q_1$)}
\node[Nmarks=n, Nw= 12,Nh=12,Nmr=6](n1g)(40,36){$(\gamma,q_2)$}


\drawloop[ELside=l,loopCW=y, loopdiam=6, loopangle=195](n0g){$s(q_1,\gamma,q_1)$}

\drawloop[ELside=l,loopCW=y, loopdiam=6, loopangle=135](n0g){$\Push(\gamma),-1$}

\drawloop[ELside=l,loopCW=y, loopdiam=6,loopangle=0](n1g){$s(q_2,\gamma,q_2)$}

\drawedge[ELpos=50, ELside=l,ELdist=0.5](n0,n0g){$\Push(\gamma),-1$}

\drawedge[ELpos=50, ELside=l, ELdist=0.5, curvedepth=6](n0,n1){$s(q_1,\bot,q_2)$}
\drawedge[ELpos=50, ELside=l, curvedepth=6](n1,n0){$s(q_2,\bot,q_1)$}

\drawedge[ELpos=50, ELside=l, ELdist=0.5, curvedepth=6](n0g,n1g){$s(q_1,\gamma,q_2)$}
\drawedge[ELpos=50, ELside=l, ELdist=0.5, curvedepth=6](n1g,n0g){$s(q_2,\gamma,q_1)$}

\end{picture}
\caption{$\Gr(\wps)$ that corresponds to the WPS $\wps$ from Figure~\ref{fig:ex-wps}.
The reader should note that this is a finite graph.
We explicitly label some of the transitions by $\Pop,\Push$ and $\Skip$ only to simplify the illustration.}\label{fig:example-gr}
\end{center}
\end{figure}
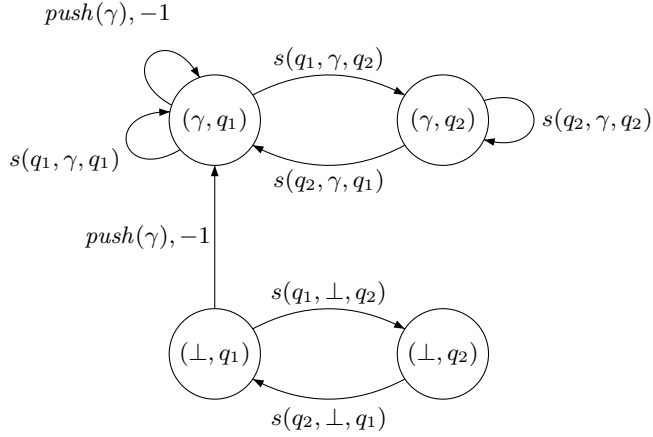


Given the computation of the summary function, we will construct a 
summary graph, and show the equivalence of the existence of good
cycles in a WPS with the existence of positive cycles in the summary graph.

\smallskip\noindent{\em Summary graph and positive simple cycles.} Given a WPS 
$\wps=\atuple{Q, \Gamma, q_0\in Q, E \subseteq  (Q\times\Gamma) \times (Q\times \Com(\Gamma)), w:E \to \Z}$
and the summary function $s$, we construct the \emph{summary graph} $\Gr(\wps)=(\ov{V},\ov{E})$ of $\wps$ with 
a weight function $\ov{w}: \ov{E} \to \Z \cup \Set{\omega}$ as follows (e.g., see Figure~\ref{fig:example-gr}): (i)~$\ov{V} = Q\times \Gamma$; and 
(ii)~$\ov{E} = E_{\Skip} \cup E_{\Push}$ where
$E_{\Skip} = \Set{((q_1,\gamma),(q_2,\gamma)) \mid s(q_1,\gamma,q_2) > -\infty}$, and
$E_{\Push} = \Set{((q_1,\gamma_1),(q_2,\gamma_2)) \mid (q_1,\gamma_1,q_2,\Push(\gamma_2)) \in E}$;
and (iii)~for all $e = ((q_1,\gamma),(q_2,\gamma))\in E_{\Skip}$ we have $\ov{w}(e) =  s(q_1,\gamma,q_2)$,
and for all $e\in E_{\Push}$ 
we have $\ov{w}(e)=w(e)$ (i.e., according to the weight function of $\wps$).
A simple cycle $C$ in $\Gr(\wps)$ is a \emph{positive simple cycle} iff one of
the following conditions holds: (i)~either $C$ contains an $\omega$-edge (i.e.,
edge labeled $\omega$ by $\ov{w}$); or 
(ii)~the sum of the weights of the edges of the cycle $C$ according to 
$\ov{w}$ is positive.

\begin{lem}\label{lemm:GoodCycleIffPositiveCycle}
A WPS $\wps$ has a good cycle iff the summary graph $\Gr(\wps)$ has a positive simple cycle.
\end{lem}
\begin{proof}
If $\wps$ has a good cycle, then let $\pi$ be a good cycle.
The good cycle $\pi$ is a non-decreasing path $\atuple{c_1,\dots,c_n}$ 
such that $c_1 = (\alpha_1\gamma,q)$ and either $c_n = (\alpha_1\gamma\alpha_2\gamma, q)$ 
or $c_n = (\alpha_1\gamma, q)$ and $w(\pi) > 0$.
Let $m_1, \dots, m_r$ be the local minima along the path.
Note that for every $i < r$, either $m_i$ and $m_{i+1}$ have the same stack 
height or $m_{i+1}$ is reachable from $m_i$ via one push transition.
For configuration $c = (\alpha\gamma,q)$, let us denote $\Top(c) = (\gamma,q)$.
Hence the path $\Top(m_1), \dots, \Top(m_r)$ is a cycle in $\Gr(\wps)$.
If the cycle contains an $\omega$-edge, then it is a positive cycle (by 
the definition of positive cycles in $\Gr(\wps)$). 
Otherwise, the weight of the cycle in $\Gr(\wps)$ is at least $w(\pi)$, and
therefore $\Gr(\wps)$ has a positive cycle (and therefore 
a positive simple cycle).

The other direction is as follows.
Consider a positive cycle in $\Gr(\wps)$.
If the cycle does not contain an $\omega$-edge, then there exists a 
non-decreasing path in $\wps$ with the same weight that forms a good cycle.
Otherwise, let $(\gamma,q)$ be a vertex in the cycle, and 
$((\gamma_1,q_1),(\gamma_1,q_2))$ be an $\omega$-edge in the 
cycle of $\Gr(\wps)$. 
From the construction of $\Gr(\wps)$, it follows that there exist 
$\alpha_1,\alpha_2,\alpha_3$ in $\wps$ such that the following 
non-decreasing paths exist:
\begin{itemize}
\item A non-decreasing path $\pi_1$ from $(\alpha_1\gamma,q)$ to 
$(\alpha_1\gamma\alpha_2\gamma_1,q_1)$ (due to the path of the cycle).
\item For every $m\in\Nat$: a non-decreasing path $\pi^m$ from $(\alpha_1\gamma\alpha_2\gamma_1,q_1)$ to $(\alpha_1\gamma\alpha_2\gamma_1,q_2)$ with weight at least $m$ (due to the $\omega$-edge).
\item A non-decreasing path $\pi_2$ from $(\alpha_1\gamma\alpha_2\gamma_1,q_2)$ to $(\alpha_1\gamma\alpha_2\gamma_1\alpha_3\gamma,q)$ (due to the path of the
cycle).
\end{itemize}
Hence, for $m = W\cdot (|\pi_1| + |\pi_2|) + 1$, we get that the path 
$\pi_1 \pi^m \pi_2$ is a good cycle.
This completes both directions of the proof and gives us the result.
\hfill\qed
\end{proof}

Since the summary function and the summary graph can be constructed in polynomial 
time, and the existence of a positive cycle in a graph can be checked in 
polynomial time (for example, first checking the existence of a cycle 
with an $\omega$-edge, and then applying Karp's mean-cycle 
algorithm~\cite{Karp78} after removing all $\omega$ edges),
we have the following lemma.

\begin{lem}\label{lem:GoodCycleWithShortCutIsEasy}
Given a WPS $\wps$, whether $\wps$ has a good cycle 
can be decided in polynomial time.
\end{lem}

Lemma~\ref{lem:GoodCycleImpliesGoodMP} and Lemma~\ref{lem:GoodCycleWithShortCutIsEasy}
give us the following theorem.

\begin{thm}\label{thm:MPSupMPInfInPTIME}
Given a WPS $\wps$, whether there exists an infinite path 
$\pi$ such that $\LimInfAvg(\pi) > 0$ (or $\LimSupAvg(\pi) > 0$)
can be decided in polynomial time.
If there exists an infinite path $\pi$ such that $\LimSupAvg(\pi) > 0$, 
then there exists an ultimately periodic 
infinite path $\pi'$ such that both 
$\LimSupAvg(\pi') > 0$ and $\LimInfAvg(\pi') > 0$.
\end{thm}

\subsection{Objectives $\LimInfAvg\geq 0$ and $\LimSupAvg\geq 0$}\label{subsec:wps2}
In this section we consider mean-payoff objectives with non-strict 
inequality.
Since in this section we will also consider rational weights due to certain 
transformation, we first discuss the issue of integer versus rational weights.

\smallskip\noindent{\em Integer versus rational weights.}
We will assume that the input WPS $\wps$ has integer weights, but we will 
consider certain transformations that produce rational weight functions.
We note that we can easily transform rational weights back to integer weights by multiplying 
all the weights by the least common multiple (LCM) of all the denominators of the weights. 
As a result the mean-payoff value of every path is multiplied by the least common multiple, 
but since we only ask if the mean-payoff value is positive (or non-negative), 
the result does not change.
We also note that the least common multiple is bounded by $D^{|E|}$, where $D$ is the greatest denominator 
that occur in the weight function (in absolute value) and $|E|$ is the number of transitions.
Hence, the least common multiple requires only $|E|\cdot\log(D)$ bits to encode and 
the blowup is polynomial.

\smallskip\noindent{\em Transformed weight functions and weighted graphs.}
Let $w:E\to\Q$ be a weight function, and $r \in \Q$ be a rational value, 
then the weight function $w+r: E\to\Q$ is defined as follows: 
for all $e \in E$ we have $(w + r)(e)=w(e) + r$.
Let $G=(V,E)$ be a (possibly infinite)\footnote{In this subsection we often look at a WPS as an infinite graph of the configurations.} 
graph with a weight function $w:E\to\Q$.
In order to emphasize that $w$ is the weight function for $G$, we use $w_G$.
We denote by $G^r$ the same infinite graph with weight function $w_G + r$.
We first show that if the lim-inf-average objective can be satisfied for all $\epsilon>0$,
then the non-strict lim-inf-average objective can also be satisfied.

\NewChange{
\smallskip\noindent{\em Solution overview.}
We prove that there is a computable $\epsilon > 0$ with polynomial number of bits such that there is a witness path $\pi$ with $\LimInfAvg(\pi) \geq 0$ iff there is a witness path $\pi$ with $\LimInfAvg(\pi) > -\epsilon$.
Hence, a reduction to the strict mean-payoff problem follows.
In Lemmas~\ref{rem:FromGToGe} and~\ref{lem:EpsilonGreatherThenImpliesGe} we prove the existence of such $\epsilon$, and in Theorem~\ref{thm:LeqIsInP} we prove the correctness of the reduction.
}

\begin{lem}\label{rem:FromGToGe}
Let $\wps$ be a WPS.
There exists a path $\pi$ with $\LimInfAvg(\pi) \geq 0$ iff for every $\epsilon > 0$ there exists a path $\pi_\epsilon$ with $\LimInfAvg(\pi_\epsilon) > -\epsilon$.
\end{lem}
\begin{proof}
The direction from left to right is trivial.
In order to prove the converse direction let us assume that for every $n\in\Nat$ there exists a path $\pi_n$ with $\LimInfAvg(\pi_n) > -\frac{1}{n}$.
Hence for every $n\in\Nat$ there exists a path $\pi^*_n$ which leads to a path $C_n$ that is a good cycle with respect to the 
weight function $w + \frac{1}{n}$.
Since there are infinitely many values of $n \in \Nat$, and since $Q$ and $\Gamma$ are 
finite, w.l.o.g all the good cycles (with respect to $w + \frac{1}{n}$) start at 
the same top configuration $(\gamma,q)$.
We define an infinite path $\pi = \pi^*_1 C_1^{W\cdot |C_2|} C_2^{2\cdot W\cdot |C_3|}\dots C_i^{i\cdot W\cdot |C_{i+1}|}\dots$ 
such that $|C_{i+1}| > |C_i|$ (since we can always extend the length of a cyclic path by taking several copies of it).
We claim that $\LimInfAvg(\pi) \geq 0$.
To prove the claim it is enough to show that for all $\nu > 0$ we have $\LimInfAvg(\pi) \geq -\nu$, and 
for this purpose it is enough to prove that for the suffix $\pi'$ that begins at position $|\pi_1^*|$
we have $\LimInfAvg(\pi') \geq -\nu$.
For every $\ell\in\Nat$ we denote by $\pi'(\ell)$ the prefix of $\pi '$ that ends at position 
$\sum_{i=1}^\ell i\cdot |C_i|\cdot |C_{i+1}| \cdot W$.
Since $C_i$ is a good cycle for $w+\frac{1}{i}$ we get that $w(C_i) \geq -\frac{|C_i|}{i}$, 
hence the average weight of $\pi'(\ell)$ is at least
\[
\frac{\sum_{i=1}^\ell -\frac{i\cdot W \cdot |C_i| \cdot |C_{i+1}|}{i}}{\sum_{i=1}^\ell i\cdot |C_i|\cdot |C_{i+1}|\cdot W} = 
-\frac{\sum_{i=1}^\ell |C_i|\cdot |C_{i+1}|}{\sum_{i=1}^\ell i\cdot |C_i|\cdot |C_{i+1}|}.
\]
Since $|C_{i+1}| > |C_i|$,
we get that 
\[
\sum_{i=1}^\ell |C_i|\cdot |C_{i+1}| \leq 2\cdot \sum_{i=\frac{\ell}{2}}^\ell |C_i|\cdot |C_{i+1}|;
\quad \text{and} \quad 
\sum_{i=1}^\ell i\cdot |C_i|\cdot |C_{i+1}| \geq \frac{\ell}{2}\cdot \sum_{i=\frac{\ell}{2}}^\ell |C_i|\cdot |C_{i+1}|,
\] 
and thus we get that the average weight of $\pi'(\ell)$ is at least $-\frac{4}{\ell}$.
Since each copy of $C_{\ell + 1}$ has an average weight of at least $-\frac{1}{1+\ell}$ 
it is obvious that the average value of $\pi'(\ell) C_{\ell+1}^n$ remains at least $-\frac{4}{\ell}$.
Moreover, since $|\pi'(\ell)| \geq \ell \cdot W\cdot |C_{\ell + 1}|$ we get that the average value of $\pi'(\ell) C_{\ell+1}^n D$ 
(where $D$ is a prefix of the finite path $C_{\ell + 1}$) is at least 
\[
\frac{-\frac{4 \cdot \ell \cdot  W \cdot |C_{\ell + 1}|}{\ell} - W\cdot |D|}{\ell \cdot W\cdot |C_{\ell + 1}| + |D|} 
\geq -\frac{5}{\ell}.
\]
Hence, we get that for all $\nu > 0$, every prefix of $\pi'$ that is longer than $|\pi'(\lceil \frac{5}{\nu} \rceil)|$ 
has an average weight of at least $-\nu$ and thus $\LimInfAvg(\pi) = \LimInfAvg(\pi ') \geq -\nu$ and the proof is completed.
%
\hfill\qed
\end{proof}

\begin{lem}\label{lem:EpsilonGreatherThenImpliesGe}
Let $\Automat{A}$ be a WPS with integer weights (weight function $w$).
Let $\ell=|\Gamma|\cdot |Q|$, and fix $\epsilon = \frac{1}{\ell^{(\ell+ 1)^2} \cdot 2\cdot \ell}$.
Then the WPS $\wps^{\epsilon}$ (with weight function $w+\epsilon$) 
has a good cycle iff for every $\delta > 0$ the WPS $\wps^{\delta}$ 
(with weight function $w+\delta$) has a good cycle.
Moreover, every good cycle in $\Gr(\wps^\epsilon)$ is a good cycle in $\Gr(\wps^\delta)$.
\end{lem}
\begin{proof}
The direction from right to left is trivial. 
For the converse direction we first prove the following lemma.

\begin{lem}\label{lemm:DiffNotTooBig}
Let $s^{\epsilon}$ be the summary function for $\wps^{\epsilon}$.
\begin{enumerate}
\item If $s^{\epsilon}(q_1,\gamma,q_2) \neq \omega$, then $s^{\epsilon}(q_1,\gamma,q_2) \leq s(q_1,\gamma,q_2) + \frac{1}{2\cdot \ell}$.
\item If $s^{\epsilon}(q_1,\gamma,q_2) = \omega$, then for every $\delta > 0$ we have
$s^{\delta}(q_1,\gamma,q_2) = \omega$, where $s^{\delta}$ is the summary function for $\wps^{\delta}$.
\end{enumerate}
\end{lem}
\begin{proof} We prove both the items below.
\begin{enumerate}
\item If $s^{\epsilon}(q_1,\gamma,q_2) \neq \omega$, then consider a
maximal-weight non-decreasing path with minimal additional stack height from 
$(\bot \gamma,q_1)$ to $(\bot \gamma,q_2)$ that has an additional stack height 
of at most $(|Q|\cdot|\Gamma|)^2=\ell^2$.
Note that this path does not contain positive cycles (since $s^{\epsilon}(q_1,\gamma,q_2) \neq \omega$). 
Hence there exists a path $\pi$ with the same weight and with stack height at 
most $\ell^2$ which does not contain any cycles.
Hence $|\pi| \leq \ell^{\ell^2}$, and therefore 
\[
w_{\wps^{\epsilon}}(\pi) = w_{\wps}(\pi) + \epsilon \cdot |\pi| 
\leq
w_{\wps}(\pi) + \epsilon \cdot \ell^{\ell^2} 
\leq w_{\wps}(\pi) + \frac{1}{2 \cdot \ell}.
\]
Since $s(q_1,\gamma,q_2) \geq w_{\wps}(\pi)$ (as $\pi$ is a 
non-decreasing path we have $s(q_1,\gamma,q_2) \geq w_{\wps}(\pi)$), 
we obtain the result of the first item.

\item
In order to prove the second item of the lemma, 
it is enough to prove that if an edge weight is $\omega$ in  
$(\wps^\epsilon)^*$ 
(where  $(\wps^\epsilon)^*$ is the WPS constructed with the 
function $(s^\epsilon)^*$),
then for every $\delta > 0$ the weight of the edge is also $\omega$ in 
the summary graph $\Gr(\wps^\delta)$  of $\wps^{\delta}$. 
We consider two cases to complete the proof.

\begin{itemize}
\item \emph{Case~1.}
If $s_\ell^{\epsilon}(q_1,\gamma,q_2) = \omega$, then the 
infinite graph $\wps^{\epsilon}$ has a positive cycle $C$ with stack height at most $\ell^2$,
and hence there exists a positive cycle $C'$ such that $|C'| \leq \ell^{\ell^2}$.
Towards a contradiction, let us assume that $w_{\wps}(C') < 0$.
As all the weights in $\wps$ are integers we get that $w_{\wps}(C') \leq -1$.
As $w_{\wps}(C') + \epsilon \cdot |C'| = w_{\wps^\epsilon}(C') \geq 0$ 
we get that $|C'| \geq \frac{1}{\epsilon}$ which is a contradiction.
Thus $w_{\wps}(C') \geq 0$, and hence for every $\delta > 0$ we have 
$w_{\wps^\delta}(C') > 0$. Thus $s_\ell^{\delta}(q_1,\gamma,q_2) = \omega$.

\item \emph{Case~2.} Otherwise, we have 
$s_{\ell+1}^\epsilon (q_1,\gamma,q_2) > s_{\ell}^\epsilon (q_1,\gamma,q_2)$.
Let $\pi$ be a path from $(\bot\gamma,q_1)$ to $(\bot\gamma,q_2)$ with additional 
stack height $\ell+1$ and weight $s_{\ell+1}^\epsilon (q_1,\gamma,q_2)$.
As $s_{\ell+1}^\epsilon (q_1,\gamma,q_2) > s_{\ell}^\epsilon (q_1,\gamma,q_2)$, 
by Lemma~\ref{lem:HeavyPathWithBigDepthImpliesPosPump} it follows that 
$\pi$ has a pumpable pair $(p_1,p_2)$ with $w_{\wps^\epsilon}(p_1) + w_{\wps^\epsilon}(p_2) > 0$.
If $p_1$ (resp. $p_2$) contains a positive cycle, then by the same arguments 
presented in the proof of the first item of the lemma this cycle will 
be positive also in $\wps^\delta$, for every $\delta > 0$, and hence 
$s^\delta(q_1,\gamma,q_2) = \omega$.
If $p_1$ (resp. $p_2$) contains a non-negative cycle, then we can remove the cycle 
and still obtain a pumpable pair with sum of weights positive.
Therefore w.l.o.g both $p_1$ and $p_2$ do not contain any cycles and thus
$|p_1|,|p_2| \leq \ell^{\ell+1}$.
Again by the same arguments presented in the proof of the first item
we obtain that $w_{\wps}(p_1) + w_{\wps}(p_2) \geq 0$ and 
hence for every $\delta > 0$ we have $w_{\wps^\delta}(p_1) + w_{\wps^\delta}(p_2) > 0$.
As $(p_1,p_2)$ is a positive pumpable pair in $\wps^\delta$
it follows that $s^\delta(q_1,\gamma,q_2) = \omega$. 
\end{itemize}
This completes the proof of the second item.
\end{enumerate}
We obtain the desired result of the lemma.
\hfill\qed
\end{proof}

We are now ready to prove Lemma~\ref{lem:EpsilonGreatherThenImpliesGe}.
Let us assume that there exists a good cycle in $\wps^\epsilon$.
Then by Lemma~\ref{lemm:GoodCycleIffPositiveCycle} there exists 
a positive simple cycle $C$ in the summary graph $\Gr(\wps^\epsilon)$.
We consider two cases:

\begin{itemize}
\item If $C$ contains an $\omega$-edge $e$, 
then by Lemma~\ref{lemm:DiffNotTooBig} for every $\delta > 0$
the same cycle in $\Gr(\wps^\delta)$ will also contain an $\omega$-edge. 
Therefore $C$ is a positive cycle also in $\Gr(\wps^\delta)$ and 
hence $\wps^\delta$ has a good cycle.

\item
Otherwise $C$ does not contain an $\omega$-edge.
Towards a contradiction assume that the weight of $C$ in $\Gr(\wps)$ is negative.
As the weights of $\wps$ are integers it follows that the weight of $C$ is at most $-1$.
By Lemma~\ref{lemm:DiffNotTooBig}, for every $e\in C$ we have
$w_{\Gr(\wps^\epsilon)}(e) \leq w_{\Gr(\wps)}(e) + \frac{1}{2 \cdot \ell}$; and 
thus 
$w_{\Gr(\wps^\epsilon)}(C) \leq w_{\Gr(\wps)}(C) + \frac{|C|}{2\cdot \ell}$.
As $C$ is a simple cycle (in $\Gr(\wps^\epsilon)$) we get that $|C| \leq \ell$,
and hence we have 
$w_{\Gr(\wps^\epsilon)}(C) \leq w_{\Gr(\wps)}(C) + \frac{1}{2} \leq -\frac{1}{2}$, 
which contradicts the assumption that $C$ is a positive cycle.
Therefore we have $w_{\Gr(\wps)}(C) \geq 0$, and therefore for every $\delta >0$ 
we get that $w_{\Gr(\wps^\delta)}(C) > 0$ and hence $\wps^\delta$ 
has a good cycle. 
\end{itemize}
The moreover part of the lemma follows from the fact that $C$ is an arbitrary positive cycle in $\Gr(\wps^\epsilon)$.
This completes the proof of the lemma. 
\hfill\qed
\end{proof}

\begin{thm}\label{thm:LeqIsInP}
Given a WPS $\wps$, whether there exists an infinite path 
$\pi$ such that $\LimInfAvg(\pi) \geq 0$ (or $\LimSupAvg(\pi) \geq 0$)
can be decided in polynomial time.
There exists a WPS $\wps$ such that there exists a path $\pi$ with 
$\LimInfAvg(\pi) = 0$ but for every ultimately periodic path $\pi$ we have 
both $\LimInfAvg(\pi) < 0$ and  $\LimSupAvg(\pi) < 0$.
\end{thm}
\begin{proof}
From Lemma~\ref{rem:FromGToGe} it follows that if there is a 
path $\pi$ such that $\LimInfAvg(\pi) \geq 0$, then for every $\epsilon_1>0$
there is a path $\pi'$ such that $\LimInfAvg(\pi') > -\epsilon_1$.
By Lemma~\ref{lem:EpsilonGreatherThenImpliesGe} it follows that it suffices
to check for $\epsilon$ (for the $\epsilon$ described by 
Lemma~\ref{lem:EpsilonGreatherThenImpliesGe}). 
Given a WPS $\wps$, the WPS $\wps^\epsilon$ 
can be constructed in polynomial time (as $\epsilon$ only has polynomial number of bits).
Then applying the polynomial-time algorithm to find good cycles (as given in 
the previous subsection) we answer the decision problems 
in polynomial time.
We observe that Lemma~\ref{rem:FromGToGe} and Lemma~\ref{lem:EpsilonGreatherThenImpliesGe}
also hold for $\LimSupAvg$ objectives, and thus the result also follows
for $\LimSupAvg$ objectives.

Example~\ref{ex:illustration1} shows \NewChange{that ultimately periodic witness paths might not exist in some cases.}
\hfill\qed
\end{proof}

We get the next corollary from Lemma~\ref{lem:EpsilonGreatherThenImpliesGe}.

\begin{cor}\label{cor:BoundOnGe}
Given a WPS $\wps$ with integer weights, 
let $\ell=|\Gamma|\cdot |Q|$, and $\epsilon = \frac{1}{\ell^{(\ell+ 1)^2} \cdot 2\cdot \ell}$.
Then there is a path $\pi'$ with $\LimInfAvg(\pi') > 0$ if and only there is a path 
$\pi$ with $\LimInfAvg(\pi) \geq \epsilon$.
\end{cor}
\begin{proof}
Towards a contradiction we assume that there is a path $\pi'$ such that $\LimInfAvg(\pi') > 0$ 
but for all paths $\pi$ we have $\LimInfAvg(\pi) < \epsilon$.
Consider the WPS $\wps^{-\epsilon}$ obtained by subtracting $\epsilon$ from all the weights of $\wps$.
Then in  $\wps^{-\epsilon}$ there are no good cycles, that is, for all the cycles $C$ in 
$\Gr(\wps^{-\epsilon})$ the sum of the weights of the cycle $C$ is negative 
(note that if $\Gr(\wps^{-\epsilon})$ has a cycle with non-negative sum of weights then we could 
obtain a path $\pi$ with $\LimInfAvg(\pi) \geq \epsilon$).
We construct a WPS $(-\wps)^\epsilon$ by multiplying all the weights of $\wps^{-\epsilon}$ by $-1$.
Since all the cycles in $\Gr(\wps^{-\epsilon})$ have negative sum of weights it follows that all the cycles in $\Gr((-\wps)^\epsilon)$ are good.
Hence, by Lemma~\ref{lem:EpsilonGreatherThenImpliesGe} it follows that for every $\delta > 0$, all the cycles in $\Gr((-\wps)^\delta)$ are good.
Therefore, for every $\delta > 0$ we get that all the cycles in $\Gr(\wps^{-\delta})$ have negative sum of weights 
and it follows that for all paths $\pi$ in $\wps^{-\delta}$ we have $\LimInfAvg(\pi) <0$.
Thus, we conclude that for every $\delta > 0$ and for every infinite path $\pi$ we have $\LimInfAvg(\pi) \leq \delta$, 
which contradicts the assumption that there exists a path $\pi'$ with $\LimInfAvg(\pi') > 0$ 
(since surely, for some $\delta > 0$ it must hold that $\LimInfAvg(\pi') > \delta$).
\hfill\qed 
\end{proof}

\begin{remark}
In Corollary~\ref{cor:BoundOnGe} we show a bound $\epsilon$ such that if there 
is a path $\pi'$ with $\LimInfAvg(\pi') > 0$ then there is a path $\pi$ such 
that $\LimInfAvg(\pi) \geq \epsilon >0$, where the denominator of $\epsilon$ 
is exponential in the input WPS (thus can be expressed in polynomially many 
bits). 
A matching exponential lower bound is also easy to obtain and we describe the main 
ideas: 
consider the language that consists of a single string $a^{2^n}$, and it is 
well-known that a pushdown automata with $O(n)$ states can accept the 
language. 
Consider a WPS obtained from the pushdown automata that after reaching the 
accepting state receives weight~1, and all other weights are~0, and from the 
accepting state returns back to the start state.
In such a WPS with $O(n)$ states there is a path $\pi'$ such that 
$\LimInfAvg(\pi')>0$, but for all paths $\pi$ we have 
$\LimInfAvg(\pi) \leq \frac{1}{2^n}$. 
\end{remark}

\begin{cor}\label{cor:Ult-Per-Check}
Given a WPS $\wps$, whether there exists an ultimately periodic 
infinite path $\pi$ such that $\LimInfAvg(\pi) \geq 0$ (or $\LimSupAvg(\pi) \geq 0$)
can be decided in polynomial time.
\end{cor}
\begin{proof}
We first observe that for any ultimately periodic path $\pi$ we have 
$\LimInfAvg(\pi) = \LimSupAvg(\pi)$.
Hence, it is enough to prove the assertion for the $\LimInfAvg(\pi) \geq 0$ objective.
The proof will immediately follow from the next claim:
There exists an ultimately periodic path $\pi = \pi_0 (\pi_1)^\omega$ with $\LimInfAvg(\pi) \geq 0$ 
if and only if the summary graph $\Gr(\wps)$ contains a (reachable) cycle with non-negative sum of 
weights.
We first prove the claim.
The proof for the direction from right to left is straightforward and is as follows.
If $\Gr(\wps)$ has a cycle with non-negative sum of weights, then there exists a non-decreasing cyclic path $\pi_1$ 
that begins in configuration $(\alpha \gamma, q)$ with $w(\pi_1) \geq 0$.
Since the cycle is reachable, then there is a path $\pi_0$ that begins in $(\bot,q_0)$ and ends in $(\alpha \gamma, q)$. 
Since $\pi_1$ is non-decreasing, the path $\pi_0 (\pi_1)^\omega$ is a valid infinite path, and since $w(\pi_1) \geq 0$ 
we get that $\LimInfAvg(\pi_0 (\pi_1)^\omega) = \frac{w(\pi_1)}{|\pi_1|} \geq 0$.
To prove the converse direction, we assume that all the cycles in $\Gr(\wps)$ have negative sum of weights and 
show that for all ultimately periodic paths $\pi$ we have $\LimInfAvg(\pi) <0$.
Let $\pi = \pi_0 (\pi_1)^\omega$ be a valid path in $\wps$, and 
let $c_1,c_2,\dots$ be the configurations of $\pi$ and 
$m_1,m_2,\dots$ be the sequence of infinitely many local minima in $\pi$.
Let $(\gamma,q)$ be a stack symbol and a state such that $|\{i \mid \Top(m_i) = (\gamma,q)\}| = \infty$.
Since $\pi$ is ultimately periodic, then for some index $i$ we get that for every $j\in\Nat$ it holds that $\Top(c_{i+|\pi_1|\cdot j}) = (\gamma,q)$.
Hence, if we denote by $W_{(\gamma,q)}$ the weight of the maximal-weight non-decreasing path from $(\gamma,q)$ to $(\gamma,q)$ 
we get that $\LimInfAvg(\pi = \pi_0 (\pi_1)^\omega) \leq \frac{W_{(\gamma,q)}}{|\pi_1|}$.
Since all the cycles in $\Gr(\wps)$ have negative sum of weights we get that $W_{(\gamma,q)} < 0$, 
and therefore $\LimInfAvg(\pi) < 0$.
Hence, the polynomial-time algorithm is to construct $\Gr(\wps)$ and to detect the existence of a non-negative cycle.
\hfill\qed
\end{proof}

\subsection{Mean-payoff objectives with stack boundedness}
In this section we consider WPSs with mean-payoff objectives
along with the \emph{stack boundedness} condition that 
requires the height of the stack to be bounded.
An infinite path $\pi =\atuple{c_1, c_2, \dots c_i \dots}$ 
is a \emph{stack bounded path} if there exists $n\in\Nat$ 
such that $|\alpha_i| \leq n$ for every $i\in\Nat$ 
(recall that $\alpha_i$ is the stack string of configuration $c_i$).

\begin{thm}\label{thm:StackBoundnessInPTIME}
Given a WPS $\wps$, the following problems can be solved in 
polynomial time.
\begin{enumerate}
\item Does there exist a stack bounded infinite path $\pi$ such that $\LimInfAvg(\pi) \bowtie 0$ 
(resp. $\LimSupAvg(\pi) \bowtie 0$), for $\bowtie \in \Set{\geq,>}$?
\item Is $\sup\Set{\LimInfAvg(\pi) \mid \mbox{$\pi$ is a stack bounded path}} \geq 0$
 (resp. $\sup\Set{\LimSupAvg(\pi) | \mbox{$\pi$ is a stack bounded path}} \geq 0$)?
\end{enumerate}
\end{thm}
\begin{proof}
The results for each item are proved with a lemma below.

\begin{lem}\label{lemm:OnlyCheckSelfLoop}
There exists a stack bounded infinite path $\pi$ in  $\wps$ 
such that $\LimSupAvg(\pi) > 0$ (resp. $\LimSupAvg(\pi) \geq 0$) iff the 
summary graph $\Gr(\wps)$ has a vertex with self-loop that has a positive 
(resp. non-negative) weight.
\end{lem}
\begin{proof}
If there exists a stack bounded infinite path $\pi$ in $\wps$ 
such that $\LimSupAvg(\pi) > 0$ (resp. $\LimSupAvg(\pi) \geq 0$), 
then it contains a cycle that begins and ends at configuration 
$(\alpha\gamma, q)$ with positive (resp. non-negative) weight.
Hence in the summary graph $\Gr(\wps)$ the vertex $(\gamma,q)$ will have a self-loop 
with positive (resp. non-negative) weight.
Conversely, if there is a (reachable) vertex $(\gamma,q)$ in $\Gr(\wps)$ with a positive (resp. non-negative) weight self-loop, then for some stack string $\alpha$ there is a non-decreasing path $\pi_1$ from $(\alpha \gamma,q)$ to $(\alpha \gamma,q)$ with $w(\pi_1) > 0$ (resp. $w(\pi_1) \geq 0$) and there is a path $\pi_0$ from $(\bot,q_0)$ to $(\alpha \gamma,q)$.
Hence the stack height of the path $\pi=\pi_0 (\pi_1)^\omega$ is bounded by $\ASH(\pi_0) + \ASH(\pi_1)$, and 
$\LimSupAvg(\pi) >0$ (resp. $\LimSupAvg(\pi) \geq 0$). 
\hfill\qed
\end{proof}

It is straightforward to verify that Lemma~\ref{lemm:OnlyCheckSelfLoop} 
also holds for $\LimInfAvg(\pi)$ objectives 
(we simply replace every occurrence of $\LimSupAvg$ by $\LimInfAvg$ and 
the same argument holds).
This gives us the first item of the theorem.
The next lemma proves the last item of the theorem.

\begin{lem}\label{lemm:SupMP}
Let $\epsilon$ be the constant from Lemma \ref{lem:EpsilonGreatherThenImpliesGe}.
Then there exists a stack bounded infinite path $\pi$ such that 
$\LimInfAvg(\pi) > -\epsilon$ iff $\sup\Set{\LimInfAvg(\pi) \mid \mbox{$\pi$ is a stack bounded path}} \geq 0$.
\end{lem}
\begin{proof}
The direction from right to left is immediate (clearly, if all paths $\pi$ have $\LimInfAvg(\pi) \leq -\epsilon$, 
then $\sup\Set{\LimInfAvg(\pi) \mid \mbox{$\pi$ is a stack bounded path}} < 0$).
In order to prove the other direction let us assume that there exists a stack bounded infinite path $\pi$ such that $\LimInfAvg(\pi) > -\epsilon$.
Hence by Lemma~\ref{lemm:OnlyCheckSelfLoop} the summary graph $\Gr(\wps^\epsilon)$ contains 
a self-loop for vertex $(\gamma, q)$ with positive weight.
By the same argument used in the proof of Lemma~\ref{lem:EpsilonGreatherThenImpliesGe} 
it follows that for every $\delta > 0$ the self-loop of vertex $(\gamma,q)$ will have a positive weight 
in graph $\Gr(\wps^\delta)$.
Hence for every $\delta > 0$ there exists a stack bounded path $\pi_\delta$ such that 
$\LimInfAvg(\pi_\delta) > -\delta$, which implies that $\sup\Set{\LimInfAvg(\pi) \mid 
\mbox{$\pi$ is a stack bounded path}} \geq 0$. \pfbox
\end{proof}
The proof of Lemma~\ref{lemm:SupMP} straightforwardly extends to $\LimSupAvg(\pi)$ objectives 
(we simply replace every occurrence of $\LimInfAvg$ by $\LimSupAvg$ and the same argument holds),
and hence we have the desired result of the theorem.
\hfill\qed
\end{proof}

Thus we have the following result summarizing the computational 
complexity.

\begin{thm}\label{thrm:wps-final-all-one}
Given a WPS $\wps$, the following questions can be solved in polynomial time:
(1)~Whether there exists a path $\pi$ in $\Phi \bowtie 0$, 
where $\Phi \in \Set{\LimSupAvg,\LimInfAvg}$ and $\bowtie \in \Set{\geq, >}$; 
and
(2)~whether there exists a path $\pi$ 
in $\Phi \bowtie 0$ such that $\pi$ is stack bounded, where 
$\Phi \in \Set{\LimSupAvg,\LimInfAvg}$ and $\bowtie \in \Set{\geq, >}$.
\end{thm}

\def\QI{{Q_{I}}}
\def\QII{Q_{II}}
\def\VI{{V_{I}}}
\def\VII{V_{II}}
\newcommand{\wpg}{\mathcal{G}}
\newcommand{\wpa}{\mathcal{A}}
\newcommand{\rev}{\mathsf{rev}}
\newcommand{\ch}{\mathit{ch}}
\newcommand{\Gad}{\mathsf{Gad}}

\section{Mean-Payoff Pushdown Games}\label{sec:games}
In this section we consider pushdown games with mean-payoff objectives.
We will show that the problem of deciding the existence of a strategy 
(or even a finite-memory strategy) to ensure mean-payoff objectives in 
pushdown games is undecidable.
The undecidability results will be obtained by a reduction from the 
\emph{universality problem} of weighted automata, which is known 
to be undecidable~\cite{Krob,ABK11}.
We start with the definition of weighted pushdown games.

\smallskip\noindent{\bf Weighted pushdown games (WPGs).}
A \emph{weighted pushdown game (WPG)} $\wpg=\atuple{\wps,(Q_1,Q_2)}$ 
consists of a WPS $\wps$ and a partition $(Q_1,Q_2)$ of the state
space $Q$ of $\wps$ into player-1 states $Q_1$ and player-2 states
$Q_2$.
A WPG defines an infinite-state game graph  
$(\ov{V},\ov{E})$ with partition $(\ov{V}_1,\ov{V}_2)$ of the vertex set 
$\ov{V}$, where $\ov{V}$ is the set of configurations of $\wps$, 
and $\ov{V}_1 = \Set{(\alpha, q) \in \ov{V} \mid q\in Q_1}$, 
$\ov{V}_2 = \Set{(\alpha, q) \in \ov{V} \mid q\in Q_2}$ 
and $\ov{E}$ is obtained from the transitions of $\wps$.
The \emph{initial vertex} is the configuration $(\bot,q_0)$.

\smallskip\noindent{\em Plays and strategies.}
A \emph{play} on $\wpg$ (or equivalently on the infinite-state game graph) is 
played in the following way: a pebble (or token) is placed on the initial vertex; 
and in every round, if the pebble is currently on player-$1$ vertex (a vertex
in $\ov{V}_1$), then he chooses an edge to follow, and moves the pebble accordingly;
and if the current vertex is a player-2 vertex, he does likewise.
The process goes on forever and generates an infinite \emph{play}
(an infinite path $\pi$ in the infinite graph of the game).
A \emph{strategy} for player~1 is a recipe to extend plays; 
formally, a strategy for player~1 is a function 
$\tau : \ov{V}^* \times \ov{V}_1 \to \ov{V}$ such that for all 
$w \in \ov{V}^*$ and $v \in \ov{V}_1$ we have 
$(v,\tau(w \cdot v)) \in \ov{E}$.
Equivalently a strategy for player~1 given a history of configurations (i.e., the sequence 
of configurations of the finite prefix of a play) ending in a player-1 state, 
chooses the successor configuration according to the transitions of $\wps$.
A play $\pi = v_1 v_2 \dots$ is \emph{consistent} with a strategy $\tau$ 
if for every $v_i \in \ov{V}_1$ we have $v_{i+1} = \tau(v_1 v_2 \dots v_i)$, i.e., 
the play is possible according to the strategy $\tau$.
The definition of player-2 strategies is analogous.
Informally a strategy can be viewed as a transducer that takes as 
input the sequence of transitions, and outputs the transitions to be taken.
A strategy is called a \emph{finite-memory strategy} if there is a finite-state transducer to 
implement the strategy.
Formally, a finite path in a WPG $\wpg$ with WPS $\wps$ starts in the initial 
configuration and is a finite sequence of transitions in $\wps$, and a  
finite-memory strategy is a finite-state transducer that has the set of 
transitions as both the input and the output alphabet.
Thus given a finite sequence of transitions, the strategy as the transducer 
outputs the next transition to be played.

\smallskip\noindent{\em Winning strategies.}
We will consider mean-payoff objectives, as already defined in the 
previous section.
A player-1 strategy $\tau$ is a \emph{winning strategy} if for every play $\pi$ consistent with $\tau$ 
we have $\LimInfAvg(\pi) \geq 0$ (resp. $\LimInfAvg(\pi) > 0, \LimSupAvg(\pi) \geq 0, \LimSupAvg(\pi) > 0$).
In other words, a winning strategy for player~1 ensures the mean-payoff objective against
all strategies of player~2.
We are interested in the question of existence of a winning strategy, 
and the existence of a finite-memory winning strategy for player~1 
in WPGs with mean-payoff objectives.
\begin{comment} 
Player-$i$ the \emph{winner of the game} if it has a winning strategy.

A strategy $\tau : V^* \times \VI \to V$ is a \emph{positional strategy} if $\tau(\rho v) = \tau(v)$ for every $v\in \VI$ and $\rho \in V^*$.

\Heading{Finite Memory Strategies}
\textbf{YARON TO KRISH - I am not sure that it is worth refereeing to finite memory strategies. It takes a lot of space, and I am not sure these definitions are standard}

An \emph{alphabeted mean-payoff pushdown game} is a tuple
$\Automat{A} = (Q = \QI \cup \QII, q_0 \in Q, \Gamma,\bot \in \Gamma,
\Sigma,
E : (Q\times \Gamma \times \Sigma) \to (Q \times \Com(\Gamma)), w:E\to\Z)$.

Where $\Sigma$ is a finite alphabet.

The edges of $G_{\Automat{A}} = (V = \VI \cup \VII, E)$ are labeled with letters from $\Sigma$ according to the transitions of $\Automat{A}$.
W.l.o.g the outdegree of every vertex is $|\Sigma|$ and the graph is bipartite.

A \emph{play} is played as follows.
At every round, player-1 selects a letter $\sigma_1 \in \Sigma$ and player-2 responds with a letter $\sigma_2 \in \Sigma$.

An infinite play forms an infinite word $\sigma_1 \sigma_2 \dots$
An infinite word forms a path (only one) $\pi$.
Player-1 wins the play if $\LimInfAvg(\pi) \geq 0$.

A \emph{player-1} strategy is a function $\tau : (\Sigma \times \Sigma)^* \to \Sigma$.

A strategy $\tau$ is a \emph{finite memory strategy} if it can be implemented by a transducer.
\subsection{Weighted Finite Automata}
\end{comment}
Our undecidability results for WPGs with mean-payoff objectives will
be obtained by a reduction from the non-universality problem of weighted 
finite automata.
We define the problem below.

\smallskip\noindent{\bf Weighted finite automata (WFA).} 
A \emph{weighted finite automaton (WFA)} is a tuple 
$\wpa = \atuple{\Sigma,Q,q_0,\Delta,w:\Delta\to\Z}$,
where $\Sigma$ is a finite input alphabet,
$Q$ is a finite set of states,
$\Delta \subseteq Q \times \Sigma \times Q$ is a transition relation,
$w:\Delta\to\Z$ is a weight function 
and $q_0 \in Q$ is the initial state.
For a word $\rho = \sigma_1 \sigma_2 \dots \sigma_n$, 
a \emph{run} of $\wps$ on $\rho$ is a sequence $r = r_0 r_1 \dots r_n \in Q^+$, 
where $r_0 = q_0$, and for all $1 \leq i \leq n$ we have 
$d_i = (r_{i-1},\sigma_i,r_i) \in \Delta$.
The weight of the run $r$ is $w(r) = \sum_{i=1}^n w(d_i)$.
Since the automaton is non-deterministic there maybe several 
runs for a word, and the weight of a finite word  
$\rho \in \Sigma^*$ over $\wps$ is the minimal weight over all runs 
on $\rho$, i.e.,
$L_{\wpa}(\rho) = \min\Set{w(r) \mid \mbox{$r$ is a run of $\wps$ on $\rho$}}$.
The \emph{non-universality} problem asks, given $\nu\in \Z$, whether there exists a word 
$\rho\in\Sigma^*$ for which $L_{\wpa}(\rho) \geq \nu$?; 
equivalently, is it not the case that for every $\rho\in\Sigma^*$ we have 
$L_{\wpa}(\rho) \leq \nu - 1$?

\begin{thm}[\cite{ABK11}]\label{thm:WeightedAutoUndec}
The non-universality problem is undecidable for WFA, \NewChange{even in the special case where the weight function is
$w :\Delta\to\Set{-1,0,1}$, the automaton has unique initial state and the threshold is $0$}. 
\end{thm}

Informally, given a WFA $\wpa$ we will construct a WPG
in such way that in the first rounds player-1 fills the stack with 
letters that construct a word $\rho$ of $\wpa$, and then player-2 
simulates the WFA's minimal run on $\rho$ and then the game returns to the 
initial state.
If for all $\rho\in\Sigma^*$ we have $L_{\wpa}(\rho) \leq 0$, 
then the mean-payoff of the play will be at most $0$,
otherwise, there exists a word $\rho\in\Sigma^*$ such that 
$L_{\wpa}(\rho) > 0$, and then by playing according to $\rho$, player-1 can 
ensure a positive mean-payoff.

\smallskip\noindent{\bf Reduction: WFA to WPGs.}
We first prove that WPGs are undecidable for $\LimInfAvg(\pi) > 0$ 
and $\LimSupAvg(\pi) > 0$ objectives.
This proof will immediately show the undecidability also for 
$\LimInfAvg(\pi) \geq 0$ and $\LimSupAvg(\pi) \geq 0$ objectives, 
as $\LimInfAvg(\pi) \geq 0$ (resp.  $\LimSupAvg(\pi) \geq 0) $ is dual to the
objective of player~2 when the objective of player~1 is $\LimSupAvg(\pi) > 0$ 
(resp. $\LimInfAvg(\pi) > 0$).

\smallskip\noindent{\em Reduction.}
The reduction from the non-universality problem of a weighted automaton is as follows.
Given a WFA $\wpa = \atuple{\Sigma,Q,q_0,\Delta,w:\Delta\to\Set{-1,0,1}}$ 
we construct a WPG $\wpg$ with the aid of \emph{five} gadgets, and we describe
the gadgets below.
WLOG we assume that there is a special symbol $\$$ that does not belong to 
$\Sigma$.

\begin{enumerate}

\item \emph{Gadget~1.} 
The first gadget contains only one state, namely $q_{{\$}}$, which is a 
player-1 state.
The state has two possible transitions.
In the first transition it pushes $\$$ into the stack and remains in the same 
state.
In the second transition it pushes $\$$ and goes to the second gadget.
All the weights in this gadget are $-10$.

\item \emph{Gadget~2.}
The second gadget 
also contains one state, namely $q_{{\Sigma}}$, which is also a player-1 state.
For every $\sigma \in \Sigma$ the state has a transition that pushes 
$\sigma$ into the stack and remains in the same state.
In addition there is one more transition, which leads to  the third gadget
keeping the stack unchanged with weight~0.
All the weights in this gadget (other than the skip transition) are $-1$.
Informally, in this gadget player~1 needs to construct a word $\rho$ such that 
the reverse of $\rho$ has value at least~1 in $\wpa$.
For a word $\rho$, let $\rev(\rho)$ denote the reverse of the word.

\begin{itemize}
\item In this gadget player~1 should construct a word $\rho \in\Sigma^*$ for which $L_{\wpa}(\rev(\rho)) \geq 1$.
\item The WPG $\wpg$ will be constructed in such way that player~1 must play 
in a way so that the number of $\$$ in the stack will be greater than the 
number of letters from $\Sigma$ to ensure the mean-payoff objectives.
\end{itemize}

\item \emph{Gadget~3.}
The third gadget is the \emph{choice} gadget with only 
one player-2 state $q_{\ch}$, which either leads to 
the fourth gadget or the fifth gadget.
The weights of the transitions are $0$ and the stack is not changed.
Informally, player-2 should go to the fifth gadget if the word that player~1 
pushed into the stack has non-positive weight, and 
should go to the fourth gadget if the number of $\$$ symbols in the stack is 
less than the number of symbols from $\Sigma$.

\item \emph{Gadget~4.}
The fourth gadget consists of only one player-2 state
$q_{<\$}$ (to denote that there is not enough $\$$ symbols). 
It has a transition $\Pop(\sigma)$ with $0$ weight, for all 
$\sigma \in \Sigma$; and a transition $\Pop(\$)$ with $+11$ weight.
If the stack is empty, then there is a transition to the initial state.

\begin{figure}[!tb]
\begin{center}
\begin{picture}(65,65)(-30,-30)
\node[Nmarks=i, iangle=180](n0)(-30,12){$q_{\$}$}
\node[Nmarks=n](n1)(10,12){$q_{\Sigma}$}
\node[Nmarks=n, Nmr=0](n2)(50,12){$q_{\ch}$}
\node[Nmarks=n, Nmr=0](n3)(50,36){$q_{<\$}$}
\node[Nmarks=n, Nw=20, Nh=20, Nmr=1](n4)(50,-12){$\Gad_5$}

\drawloop[ELside=l,loopCW=y, loopdiam=6, loopangle=90](n0){$\_,\Push(\$),-10$}
\drawloop[ELside=l,loopCW=y, loopdiam=6, loopangle=270](n1){$\_,\Push(\sigma),-1; \sigma\in \Sigma$}
\drawedge[ELpos=50, ELside=l, ELdist=0.5](n0,n1){$\_,\Push(\$),-10$}
\drawedge[ELpos=50, ELside=l, ELdist=0.5](n1,n2){$\_,\Skip,0$}
\drawedge[ELpos=50, ELside=r, ELdist=0.5](n2,n3){$\_,\Skip,0$}
\drawedge[ELpos=30, ELside=l, ELdist=0.5](n2,n4){$\_,\Skip,0$}

\drawloop[ELside=l,loopCW=y, loopdiam=6, loopangle=90](n3){$\sigma,\Pop(\sigma),0; \sigma\in \Sigma$}
\drawloop[ELside=l,loopCW=y, loopdiam=6, loopangle=0](n3){$\$,\Pop(\$),11$}
\drawedge[ELpos=50, ELside=r, curvedepth=-10, ELdist=0.5](n3,n0){$\bot,\Skip,0$}
\drawedge[ELpos=50, ELside=l, curvedepth=10, ELdist=0.5](n4,n0){$\bot,\Skip,0$}
\end{picture}
\caption{The WPG $\wpg$ from WFA $\wpa$.}\label{fig:undec_1}
\end{center}
\end{figure}
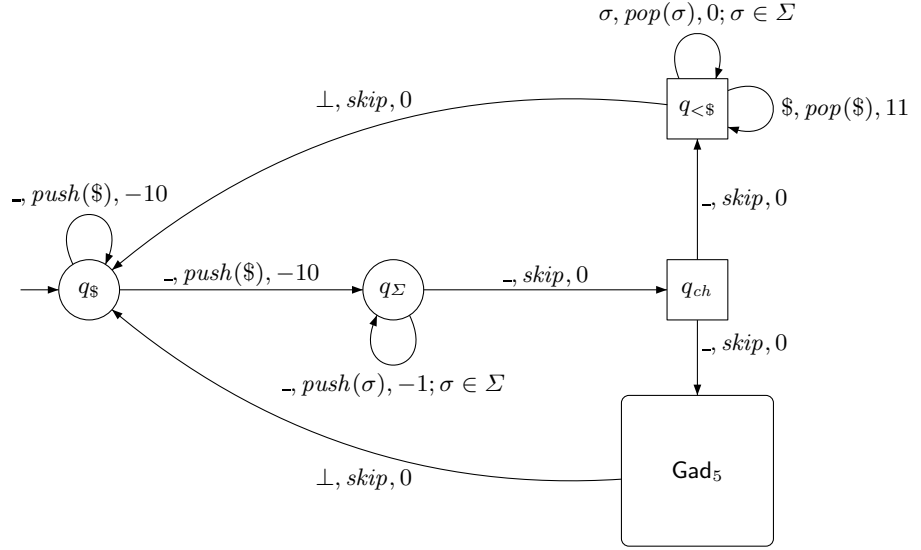

\noindent A pictorial description of the first four gadgets is shown in 
Figure~\ref{fig:undec_1}, where $\bigcirc$ denotes player-1 states, 
$\Box$ denotes player-2 states, and edges are labeled by stack top;
followed by the stack command; and then the weight; and if the stack top is 
irrelevant (i.e., the transition is valid for all stack tops), then it is
denoted as $\_$.
We now describe the fifth gadget.

\item \emph{Gadget~5.}
The fifth gadget is the \emph{simulate run} gadget.
The states in this gadgets are essentially the set $Q$ of states of the 
automaton $\wpa$; and all the states are player-2 states.
The transitions and edge weights are as follows:
(i)~for every $(q,\sigma,q')\in\Delta$ we have a transition
$(q,\sigma,q',\Pop(\sigma))$, with weight  
$w_{\wpa}(q,\sigma,q') + 1$ (1 plus the weight in $\wpa$); and
(ii)~in addition there exists a transition $(q,\$,q,\Pop(\$))$ 
with weight $+10$ and a transition $(q,\bot,q_{{\$}},\Skip)$ 
to the initial state for empty stack with weight $0$.

\end{enumerate}

\begin{remark}{\bf (Comment on naive solution).}
Our reduction consists of five gadgets and we comment on a simpler reduction.
Consider a naive solution with only two gadgets (Gadget~2 and Gadget~5): 
in one player~1 fills the stack and in the other player~2 simulates the run.
This reduction only works for $\LimInfAvg$ objectives with strict inequality.
We comment on other cases below.
\begin{enumerate}
\item  
For $\LimInfAvg(\pi)\geq 0$ objective, it is possible that the WFA assigns $-1$ 
value for all finite words, but player~1 can still achieve $\LimInfAvg(\pi)=0$ 
by selecting words with increasing lengths.  
\item 
For the $\LimSupAvg(\pi) >0$ objective it is possible that player~1 will select 
a finite word $\rho$ such that there is a run of the WFA $\wpa$ that assigns a 
zero weight to $\rho$, but the run has a prefix with positive weight.
Hence, for infinitely many positions of the path of the two-player game the total 
weight is positive and hence the path satisfies $\LimSupAvg(\pi) >0$ objective. 
\end{enumerate}
Hence a reduction with only Gadget~2 and Gadget~5 fails to uniformly 
capture all cases.
To handle the above issues, we introduce Gadgets~1, 3, and 4.
Intuitively, these gadgets ensure that before player~1 fills the stack with a word $\rho$, 
he will make at least $|\rho|+1$ push $\$$ transitions (with negative weights). 
Consequently, when player~2 simulates a run of the WFA the total weight of the play remains non-positive, 
at least until the pop $\$$ transitions are done, and the weight becomes positive only if 
$L_{\wpa}(\rho) > 0$.
\end{remark}

\Heading{Correctness of reduction.} 
We will now prove the correctness of the reduction by showing 
that there is a winning strategy (also a finite-memory winning strategy) 
in the WPG $\wpg$ for mean-payoff objectives with strict inequality 
iff there is a finite word $\rho \in \Sigma^*$ such that 
$L_{\wpa}(\rho) \geq 1$. Let $\pi$ be a play on the above WPG $\wpg$.
The \emph{$i$-th iteration} of the play are the positions between the $i$-th 
visit and the  $(i+1)$-th visit to the initial state.


\begin{lem}\label{lemm:game1}
If there is a word $\rho \in \Sigma^*$ such that $L_{\wpa}(\rho) \geq 1$, 
then there exists a finite-memory strategy $\tau_1^*$ for player~1  to ensure
that for all plays $\pi$ consistent with $\tau_1^*$ we have  
$\LimSupAvg(\pi)>0$ and 
$\LimInfAvg(\pi)>0$.
\end{lem}
\begin{proof}
The finite-memory strategy $\tau_1^*$ for player~1 is to play in every 
iteration $\$^{n+1}$ in $q_{\$}$ followed by $\rev(\rho)$ in $q_{\Sigma}$, 
where $n=|\rho|$ is the length of the word $\rho$.
In every iteration, the sum of the weights is at least $1$ as 
$L_{\wpa}(\rho) \geq 1$, and the length of a play in every iteration is at 
most $4\cdot n + 2$.
It follows that for all plays $\pi$ consistent with $\tau_1^*$ we have both
$\LimSupAvg(\pi)>0$ and $\LimInfAvg(\pi)>0$.
\hfill\qed
\end{proof}

\begin{lem}\label{lemm:game2}
If for all words $\rho \in \Sigma^*$ we have  $L_{\wpa}(\rho) \leq 0$, 
then there exists a counter strategy $\tau_2^*$ for player~2 to ensure that 
for all strategies $\tau_1$ of player~1, for the play given $\tau_2^*$ and 
$\tau_1$: 
for all iterations $i$, for every position between the $i$-th iteration and the 
$(i+1)$-th iteration, the sum of the weights from the beginning of the iteration
to the current position of the iteration is at most~0.
\end{lem}
\begin{proof} 
The counter strategy $\tau_2^*$ is as follows: 
consider an iteration $i$, and let the strategy of player~1 in this 
iteration produce the sequence $\$^n \rho$, for $\rho \in \Sigma^*$.
Note that if the state $q_{\ch}$ is never reached, then all the 
weights are negative (in $q_{\$}$ and $q_{\Sigma}$ all weights are negative).
The strategy $\tau_2^*$ is described considering the following two cases.
\begin{enumerate}
\item 
If $n \leq |\rho|$, then the strategy $\tau_2^*$ chooses the state 
$q_{<\$}$ at the state $q_{\ch}$ (since there are not enough $\$$ in the 
stack).
For any position of the play till state $q_{\ch}$ is reached,
the sum of the weights is negative.
Once $q_{<\$}$ is reached, for any position, the payoff is at most 
$-10\cdot n -|\rho| + 11\cdot n = -|\rho| + n \leq 0$, as 
$n \leq |\rho|$.

\item Otherwise, we have $n > |\rho|$ and $L_{\wps}(\rev(\rho)) \leq 0$.
There exists a run $r$ on $\rev(\rho)$ such that for every prefix $\beta$ of 
$\rev(\rho)$ the sum of the weights is at most $2\cdot |\beta| 
\leq 2 \cdot |\rho| < 2\cdot n$ (since the absolute value of the weights of 
$\wpa$ are bounded by $1$ and therefore the weights in Gadget~5 are bounded by $2$) 
and in the end of the run, the sum of the weights is at most $0$.
The counter strategy $\tau_2^*$ follows the run $r$.  
Hence the sum of the weights for any prefix $\beta$ is at most 
$-10 \cdot n + 2\cdot|\beta| \leq -10 \cdot n + 2\cdot n < 0$, 
until the letter $\$$ is the top symbol of the stack.
Once $\$$ is the top symbol, the sum of the weights is at most $-10\cdot n$, 
since the sum of the weights of the run is at most~0.
Since with each pop of $\$$ the weight is $10$, and there are $n$ pops,
it follows that in every position of the iteration the sum of the weights
is at most~0.
Finally, once the iteration is completed the sum of the weights is 
also at most $0$.
\end{enumerate}
The desired result follows.
\hfill\qed
\end{proof}

\begin{lem}\label{lemm:game-key}
Given WFA $\wpa$ and the WPG $\wpg$ constructed by the reduction, the 
following the assertions hold:
\begin{enumerate}
\item If there is a word $\rho \in \Sigma^*$ such that $L_{\wpa}(\rho) \geq 1$, 
then there is a finite-memory winning strategy  $\tau_1^*$ for player~1
for the objectives $\LimSupAvg(\pi)>0$ and $\LimInfAvg(\pi)>0$.

\item If for all words $\rho \in \Sigma^*$ we have  $L_{\wpa}(\rho) \leq 0$, 
then there is no winning strategy for player~1 for the objectives 
$\LimSupAvg(\pi)>0$ and $\LimInfAvg(\pi)>0$.

\item There exists a winning strategy (resp. a finite-memory winning strategy)
for player~1 for the objectives $\LimSupAvg(\pi)>0$ and $\LimInfAvg(\pi)>0$
iff there is a word $\rho \in \Sigma^*$ such that $L_{\wpa}(\rho) \geq 1$. 
\end{enumerate}
\end{lem}
\begin{proof}
Note that the third item is a consequence of the first two items. 
The first item follows from Lemma~\ref{lemm:game1}.
We now use Lemma~\ref{lemm:game2} to prove the second item.
Given the condition of the second item, 
let us consider the strategy $\tau_2^*$ for player~2 as described in
Lemma~\ref{lemm:game2}. 
Let $\pi$ be a play consistent with  $\tau_2^*$. We consider two cases
to complete the proof.
\begin{itemize}
\item 
If $\pi$ does not have an infinite number of iterations, then from some point on
only states $q_{\$}$ or $q_{\Sigma}$ are visited, and they both have only 
negative weights. 
Hence all the weights that occur in $\pi$ from some point on are non-positive and 
hence $\LimInfAvg(\pi) \leq \LimSupAvg(\pi) \leq 0$.
\item
Otherwise, $\pi$ has an infinite number of iterations.
Given $\tau_2^*$ it follows from Lemma~\ref{lemm:game2} that for all iterations, 
in every position of an iteration, the sum of the weights from the beginning of 
the iteration to the current position is non-positive.
Hence $\LimInfAvg(\pi) \leq \LimSupAvg(\pi) \leq 0$.
\end{itemize}
The desired result follows.
\hfill\qed
\end{proof}

\smallskip\noindent{\bf Undecidability for related decision problems.}
It follows from Lemma~\ref{lemm:game-key} that the existence of 
winning strategies (resp. finite-memory winning strategies) for mean-payoff 
objectives with strict inequality is undecidable for WPGs.
For general strategies the result also follows for non-strict inequality by 
duality.
We now show the undecidability for finite-memory strategies for the non-strict
inequality as well as undecidability for stack boundedness.
This is done by showing a reduction from the non-universality problem for WFA 
with threshold $\nu = 0$.
The reduction is identical to the original reduction presented in this section.
If there exists a word $\rho \in \Sigma^*$ such that $L_{\wpa}(\rho) \geq 0$, 
then playing $\$^{|\rho|+1}\rev(\rho)$ in every iteration is a finite-memory winning 
strategy for player-1 (also for the stack boundedness condition).
Otherwise, for every $\rho \in \Sigma^*$ we have $L_{\wpa}(\rho) \leq -1$.
In this case, against every player-1 finite-memory strategy, with memory size $M$, 
player-2 has a strategy that ensures that the mean-payoff is at most $-\frac{1}{2M}$.
Indeed, every iteration is either of length at most $2M$ (i.e., $M$ steps for 
filling the stack and $M$ steps for simulating the WFA) or 
infinite (player~1 repeatedly pushes symbols to the stack, and player~2 does
not get a chance to make a move).
In the first case the mean-payoff is at most $-\frac{1}{2M}$ and in the second 
case it is at most $-1$.
Similarly, against every player-1 strategy that ensures stack height at most 
$M$, player~2 has a strategy that ensures that the mean-payoff at most 
$-\frac{1}{2M}$ 
(since every iteration is of length at most $2M$ and the total weight of every 
iteration is at most $-1$).

\begin{thm}\label{thrm:game}
Given a WPG $\wpg$, the following questions are undecidable:
(1)~Whether there exists a winning strategy (resp. finite-memory winning strategy) 
to ensure $\Phi \bowtie 0$, where $\Phi \in \Set{\LimSupAvg,\LimInfAvg}$ and 
$\bowtie \in \Set{\geq, >}$; and
(2)~whether there exists a winning strategy (resp. finite-memory winning strategy) 
to ensure $\Phi \bowtie 0$ along with stack boundedness, where 
$\Phi \in \Set{\LimSupAvg,\LimInfAvg}$ and $\bowtie \in \Set{\geq, >}$.
\end{thm}

\smallskip\noindent{\bf Distinguishing facts.}
We now show some interesting facts about WPGs with mean-payoff objectives
that distinguish them from finite game graphs with mean-payoff objectives.
\begin{enumerate}
\item \emph{Fact~1.} It follows from the proof of Theorem~\ref{thm:LeqIsInP} 
(Example~\ref{ex:illustration1} for the fact that ultimately periodic paths 
are not sufficient) that in general, positional (or memoryless) strategies are 
not sufficient, and infinite-memory strategies are required in general 
(in contrast, in finite game graphs, memoryless winning strategies are 
guaranteed to exist).

\item \emph{Fact~2.} The objectives $\LimSupAvg$ and $\LimInfAvg$ do not 
coincide in general for WPGs.
We show this in Example~\ref{Ex:SupAndInfAreDifferent}.

\item \emph{Fact~3.} We also note that pushdown mean-payoff games are very 
different as compared to parity games. 
For finite-state games both parity and mean-payoff objectives have the same
complexity (both lie in NP $\cap$ coNP and also in UP $\cap$ coUP~\cite{Jur98}),
in contrast for pushdown games the mean-payoff problem is undecidable, whereas 
the parity problem is EXPTIME-complete~\cite{Wal01}.
Moreover, for countably infinite games with finitely many priorities, 
for parity objectives memoryless winning strategies exist~\cite{Thomas97},
whereas as we show (Fact~1) for mean-payoff pushdown games infinite-memory 
strategies are required.

\end{enumerate}

\begin{figure}[!tb]
\begin{center}
\begin{picture}(48,48)(0,0)
\node[Nmarks=i, iangle=180](n0)(0,12){$q_{I}^1$}
\node[Nmarks=n](n1)(50,12){$q_{I}^2$}
\node[Nmarks=n, Nmr=0](n2)(0,36){$q_{II}^2$}
\node[Nmarks=n, Nmr=0](n3)(50,36){$q_{II}^1$}
\drawloop[ELside=l,loopCW=y, loopdiam=6, loopangle=270](n0){$\_,\Push(\gamma),-2$}
\drawloop[ELside=l,loopCW=y, loopdiam=6, loopangle=270](n1){$\_,\Pop,4$}
\drawloop[ELside=l,loopCW=y, loopdiam=6, loopangle=90](n3){$\_,\Push(\gamma),2$}
\drawloop[ELside=l,loopCW=y, loopdiam=6, loopangle=90](n2){$\_,\Pop,-4$}
\drawedge[ELpos=50, ELside=r, ELdist=0.5](n0,n1){$\_,\Skip,0$}
\drawedge[ELpos=50, ELside=r, ELdist=0.5](n1,n3){$\_,\Skip,0$}
\drawedge[ELpos=50, ELside=r, ELdist=0.5](n3,n2){$\_,\Skip,0$}
\drawedge[ELpos=50, ELside=r, ELdist=0.5](n2,n0){$\_,\Skip,0$}
\end{picture}
\caption{Example WPG $\wpg$.}\label{fig:ex}
\end{center}
\end{figure}
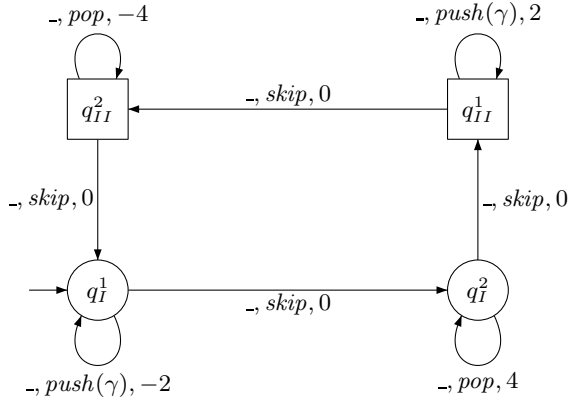

\begin{examp}\label{Ex:SupAndInfAreDifferent}
We show that there exists a WPG such that player-1 can ensure that $\LimSupAvg(\pi)\geq 2$ 
and player-2 can ensure that $\LimInfAvg(\pi) \leq -2$.
The WPG is described as follows:
Let $Q_1 = \Set{q^1_I, q^2_I}$ and $Q_{2} = \Set{q^1_{II}, q^2_{II}}$, and 
let $E$, the set of transitions, be as follows
\begin{itemize}
\item $(q^1_I, \bot, q^1_I, \Push(\gamma))$ with weight $-2$;
\item $(q^1_I, \gamma, q^1_I, \Push(\gamma))$ with weight $-2$;
\item $(q^1_I, \gamma, q^2_I, \Skip)$ with weight $0$;
\item $(q^2_I, \gamma, q^2_I, \Pop)$ with weight $+4$;
\item $(q^2_I, \bot, q^1_{II}, \Skip)$ with weight $0$;
\item $(q^1_{II}, \bot, q^1_{II}, \Push(\gamma))$ with weight $+2$;
\item $(q^1_{II}, \gamma, q^1_{II}, \Push(\gamma))$ with weight $+2$;
\item $(q^1_{II}, \gamma, q^2_{II}, \Skip)$ with weight $0$;
\item $(q^2_{II}, \gamma, q^2_{II}, \Pop)$ with weight $-4$;
\item $(q^2_{II}, \bot, q^1_{I}, \Skip)$ with weight $0$.
\end{itemize}
The WPG is shown in Figure~\ref{fig:ex}.
It is straightforward to verify that player-1 can ensure 
$\LimSupAvg(\pi) \geq 2$, and player-2 can ensure 
$\LimInfAvg(\pi) \leq -2$. 
\hfill\qed
\end{examp}

\newcommand{\wrg}{\mathcal{A}}
\newcommand{\En}{\mathit{En}}
\newcommand{\Ex}{\mathit{Ex}}

\section{Recursive Games and Modular Strategies}\label{sec:modular}
In this section we will consider a special class of strategies, namely \emph{modular} 
strategies, in pushdown games. 
Modular strategies are more intuitive in the model of recursive game graphs,
and recursive game graphs are equivalent to pushdown game graphs.
\NewChange{ We first present the definitions of recursive game graphs from~\cite{AlurReach}, then give an overview of the solution and present some basic properties for recursive game graphs.}

\smallskip\noindent{\bf Weighted recursive game graphs (WRGs).}
A \emph{recursive game graph} $\wrg$ consists of a tuple 
$\atuple{A_1,\dots,A_n}$ of \emph{game modules}, where each 
game module $A_i = (N_i,B_i,V^1_i,V_i^2,\En_i,\Ex_i,\delta_i)$ consists of the following components:
\begin{itemize}
\item A finite nonempty set of \emph{nodes} $N_i$.
\item A nonempty set of \emph{entry} nodes $\En_i \subseteq N_i$ and a 
nonempty set of \emph{exit} nodes $\Ex_i \subseteq N_i$.
\item A set of \emph{boxes} $B_i$.
\item Two disjoint sets $V_i^1$ and $V_i^2$ that partition the set of 
nodes and boxes into two sets, i.e., $V_i^1 \cup V_i^2 = N_i \cup B_i$ 
and $V_i^1 \cap V_i^2 = \emptyset$.
The set $V_i^1$ (resp. $V_i^2$) denotes the places where it is the 
turn of player~1 (resp. player~2) to play (i.e., choose transitions).
We denote the union of $V_i^1$ and $V_i^2$ by $V_i$.
\item A labeling $Y_i : B_i \to \RangeSet{1}{n}$ that assigns to 
every box an index of the game modules $A_1 \dots A_n$.
\item Let $\Calls_i = \Set{(b,e) \mid b\in B_i, e \in \En_j, j = Y_i(b)}$ 
denote the set of \emph{calls} of module $A_i$ and let 
$\Returns_i = \Set{(b,x) \mid b\in B_i, x\in \Ex_j, j = Y_i(b)}$ denote the set of \emph{returns} in $A_i$.
Then, $\delta_i \subseteq (N_i \cup \Returns_i)\times(N_i \cup \Calls_i)$ is the 
\emph{transition relation} for module $A_i$.
\end{itemize}
A \emph{weighted recursive game graph} (for short WRG) is a recursive game graph, equipped with a weight function 
$w$ on the transitions.
We also refer the readers to~\cite{AlurReach} for detailed description and illustration with figures 
of recursive game graphs.
WLOG we shall assume that the boxes and nodes of all modules are disjoint.
Let $B = \bigcup_i B_i$ denote the set of all boxes, 
$N=\bigcup_i N_i$ denote the set of all nodes, 
$\En = \bigcup_i \En_i$ denote the set of all entry nodes, 
$\Ex = \bigcup_i \Ex_i$ denote the set of all exit nodes, 
$V^1 = \bigcup_i V^1_i$ (resp. $V^2 = \bigcup_i V^2_i$) denote 
the set of all places under player~1's control (resp. player~2's control), 
and $V=V^1 \cup V^2$ denote the set of all vertices.
We will also consider the special case of one-player WRGs, where 
either $V^2$ is empty (player-1 WRGs) or $V^1$ is empty (player-2 WRGs).

\smallskip\noindent{\bf Configurations, paths, and local history.}
A \emph{configuration} $c$ consists of a sequence $(b_1,\dots,b_r,u)$, 
where $b_1,\dots,b_r \in B$ and $u\in N$.
Intuitively,  $b_1,\dots,b_r$ denote the current stack (of modules), and $u$ is the current node.
A sequence of configurations is \emph{valid} if it does not violate the transition relation.
The \emph{configuration stack height} of $c$ is $r$.
Let us denote by $\mathbb{C}$ the set of all configurations, and let 
$\mathbb{C}_1$ (resp. $\mathbb{C}_2$) denote the set of all configurations
under player~1's control (resp. player~2's control), that is, all 
configurations in which $u$ is under the control of player~1 (resp., player~2).
A \emph{path} $\pi = \atuple{c_1, c_2, c_3, \dots}$ is a valid sequence of configurations.
Let $\rho = \atuple{c_1, c_2, \dots, c_k}$ be a valid finite sequence of configurations, such that
$c_i = (b^i_1,\dots,b^i_{d_i}, u_i)$, and the stack height of $c_i$ is $d_i$.
The \emph{stack height} of $\pi$, denoted by $\SH(\pi)$, is $\max\{d_1,\dots,d_k\}$ and $\ASH(\pi) = \SH(\pi) - \max\{d_1,d_k\}$.
Let $c_i$ be the first configuration with stack height $d_i = d_k$, 
such that for every $i\leq j\leq k$, if $c_j$ has stack height $d_i$, 
then $u_j \notin \Ex$ ($u_j$ is not an exit node).
The \emph{local history} of $\rho$, denoted by $\LocalHistory (\rho)$, 
is the longest sequence $(u_{j_1},\dots,u_{j_m})$ such that $c_{j_1} = c_i$, 
$c_{j_m} = c_k$, $j_1 < j_2 < \dots < j_m$, and the stack height of 
$c_{j_1},\dots,c_{j_m}$ is exactly $d_i$.
In other words, the local history contains the nodes of configurations $c_i$ 
and $c_r$ and all the nodes of the intermediate configurations of the 
same stack height.
Intuitively, the local history is the sequence of nodes in a module. 
Note that by definition, for every $\rho \in \mathbb{C}^*$, there exists 
$i\in \Set{1,\dots,n}$ such that all the nodes that occur in 
$\LocalHistory(\rho)$ belongs to $V_i$. 
We say that $\LocalHistory(\rho) \in A_i$ if all the nodes in 
$\LocalHistory(\rho)$ belong to $V_i$.

\smallskip\noindent{\bf Global game graph and isomorphism to pushdown game graphs.}
The \emph{global game graph} corresponding to a WRG $\wrg=\atuple{A_1,\dots,A_n}$ 
is the graph of all valid configurations, and an edge $(c_1,c_2)$ between 
configurations $c_1$ and $c_2$ exists if there is a transition from $c_1$ to $c_2$.
It follows from the results of~\cite{AlurReach} that every recursive game graph has 
an isomorphic pushdown game graph that is computable in polynomial time.
We note that the simulation requires some extra transitions that may influence the value of mean-payoff in the original and the simulated runs. But since we only ask whether the mean-payoff value is positive (or non-negative), this does not influence our results, since we assign zero weights for the auxiliary transitions. 

\smallskip\noindent{\bf Plays, strategies, and modular strategies.}
A play begins at the entry node of module $A_0$ and it is played in the usual sense over the global game graph 
(which is possibly an infinite graph).
A (finite) play is a (finite) valid sequence of configurations 
$\atuple{c_1, c_2, c_3, \dots}$ (i.e., a path in the global game graph).
A finite path $\pi$ is a \emph{sub-play} if there exist a finite path $\pi_0$ such that $\pi_0 \cdot \pi$ is a prefix of a valid play.
A \emph{strategy} for player 1 is a function $\tau : \mathbb{C}^* \times \mathbb{C}_1 \to \mathbb{C}$
respecting the edge relationship of the global game graph, i.e., 
for all $w \in \mathbb{C}^*$ and $c_1 \in \mathbb{C}_1$ we have that 
$(c_1,\tau(w\cdot c_1))$ is an edge in the global game graph.
A \emph{modular strategy} $\tau$ for player~1 is a set of functions 
$\Set{\tau_i}_{i=1}^n$, one for each module, where for every $i$,
we have $\tau_i : (N_i \cup \Returns_i) ^* \to \delta_i$.
The function $\tau$ is defined as follows:
For every play prefix $\rho$ we have 
$\tau(\rho) = \tau_i(\LocalHistory(\rho))$, where $\LocalHistory(\rho) \in A_i$.
The function $\tau_i$ is the \emph{local strategy} of module $A_i$.
Intuitively, a modular strategy only depends on the local history, and 
not on the context of invocation of the module.
A modular strategy $\tau = \Set{\tau_i}_{i=1}^n$ is a \emph{finite-memory} 
modular strategy if $\tau_i$ is a finite-memory strategy for every $i\in\RangeSet{1}{n}$.
A \emph{memoryless} modular strategy is defined in similar way, 
where every component local strategy is memoryless.

\smallskip\noindent{\bf Mean-payoff  objectives and winning modular strategies.}
The weight of a finite path $\pi$, denoted by $w(\pi)$ is the sum of all the weights 
along the path.
For an infinite path $\pi$ (as in the previous sections) 
we denote $\LimInfAvg(\pi) = \liminf_{n\to\infty} \frac{w(\pi[1,n])}{n}$
(resp. $\LimSupAvg(\pi) = \limsup_{n\to\infty} \frac{w(\pi[1,n])}{n}$), where 
$\pi[1,n]$ is the initial prefix of length $n$. 
The \emph{modular winning strategy problem} asks if player~1 has a 
modular strategy $\tau$ such that for every play $\rho$ consistent with $\tau$
we have $\LimInfAvg(\rho) \geq 0$ (note that the counter strategy of player~2 
is a general strategy), and similarly for other mean-payoff objectives.

\smallskip\noindent{\bf Overview of the solution.}
We first show that player~1 has a modular winning strategy that is 
\emph{cycle free}, namely, it is does not depend on the simple cycles 
that occur in the history of the play.
We then show that a cycle-free strategy can be \emph{simulated} by 
a memoryless strategy and by the results of 
Section~\ref{sect:PushDownProcesses} we get that mean-payoff modular games 
are in NP (as we have a polynomial-time 
verifier for a player-1 winning strategy, namely, a verifier for a memoryless 
winning strategy).
The NP-hardness is obtained via a reduction from the 3-SAT problem.

\smallskip\noindent{\bf Basic properties.}
We now present some basic properties of recursive game graphs.

\smallskip\noindent{\em Non-decreasing cycles and proper cycles.}
A \emph{non-decreasing cycle} in a recursive game graph 
$\wrg=\atuple{A_1,\dots, A_n}$ 
is a path segment from a module $A_i$ and vertex $v_i \in A_i$ to the same module 
and the same vertex (possibly at different stack level), such that the 
first occurrence of module $A_i$ in the path segment does not return 
(i.e., does not reach an exit node) during the path segment.
A non-decreasing cycle $C$ is a \emph{proper cycle} if the stack heights 
at the beginning and the end of the path segment are the same.

\begin{lem}\label{lemm:NotEmptyIfCycle}
Consider a one-player WRG $\wrg=\atuple{A_1,\dots, A_n}$ (i.e., 
consists of only one-player). 
The following assertions hold: 
\begin{itemize}
\item The WRG $\wrg$ has a path $\pi$ with $\LimInfAvg(\pi) > 0$ 
(resp.  $\LimSupAvg(\pi) > 0$) iff there exists a positive non-decreasing cycle.

\item The WRG $\wrg$ has a path $\pi$ with $\LimInfAvg(\pi) < 0$ (resp. $\LimSupAvg(\pi) < 0$) iff 
there exists a negative non-decreasing cycle.
\end{itemize}
\end{lem}
\begin{proof}
The first item follows from (i)~the isomorphism of one-player WRGs and 
weighted pushdown systems (WPSs), (ii)~the correspondence of positive 
non-decreasing cycles and good cycles for WPSs, and 
(iii)~the results established in Section~\ref{sect:PushDownProcesses}
showing the equivalence of the existence of a path  
$\pi$ with $\LimInfAvg(\pi) > 0$ (resp.  $\LimSupAvg(\pi) > 0$)
and the existence of good cycles in a WPS. 
The second item follows from the duality of $\LimInfAvg(\pi) > 0$ 
and $\LimSupAvg(\pi) < 0$. 
\hfill\qed
\end{proof}

\smallskip\noindent{\em WRG given finite-memory strategies.}
Given a WRG $\wrg$, let $\tau = \Set{\tau_i}_{i=1}^n$ be a finite-memory 
modular strategy.
Let $M_i$ be the set of memory states of strategy $\tau_i$, i.e., 
$\tau_i$ is described as a deterministic transducer with state space $M_i$.
The \emph{one-player WRG (player-2 WRG)} given $\tau$ is the tuple 
$\wrg^{\tau}=\atuple{A_1^{\tau_1} = A_1 \times M_1, \dots , A_n^{\tau_n} = A_n \times M_n}$,
where each $A_i^{\tau_i}= A_i \times M_i$ is obtained as the synchronous 
product of $A_i$ and the deterministic transducer describing the local strategy 
$\tau_i$. 
Formally, in the product, if the second component is a state $x_i \in M_i$, 
then the transition for the second component is as defined by the transition 
function of the deterministic transducer over $x_i$, and in the first 
component transition we only retain the transition prescribed by $x_i$.
The weights of the transitions are specified according to the weight 
function of $\wrg$.
Note that if $\tau$ is a memoryless modular strategy, then $\wrg^{\tau}$ is a sub-game graph of $\wrg$.

\begin{lem}\label{lemm:EveryPathIsAlsoInGraph}
Given a WRG $\wrg$ and a modular strategy $\tau$,
every (finite or infinite) path in the one-player WRG 
$\wrg^{\tau}$ is a (finite or infinite) play in $\wrg$ consistent with $\tau$, 
and vice versa.
\end{lem}

\Comment
{
\begin{lem}\label{lemm:FiniteWinImpliesGoodCycles}
Let $A$ be an RSM game graph.
Let $\tau$ be a finite memory modular strategy.
Then $\tau$ is a winning strategy, for the objective $\LimInfAvg \geq 0$, iff for every non decreasing cycle in $A^{\tau}$ has a non-negative weight.
\end{lem}
\begin{proof}
We begin with the proof for the direction from left to right.
Let us assume that $\tau$ is a winning strategy.
Towards contradiction we assume that there exist a finite path $\pi = \pi_0 \pi_1$ in $A^{\tau}$, such that $\pi_1$ is a cycle with negative weight.
Hence, by Lemma~\ref{lemm:EveryPathIsAlsoInGraph}, we get that $\pi_0 (\pi_1)^\omega$ is a play according to $\tau$ with $\LimInfAvg(\pi_0 (\pi_1)^\omega) < -\frac{1}{|\pi_1|} \leq 0$, which contradicts the assumption that $\tau$ is a winning strategy.
Hence, every non decreasing cycle in $A^{\tau}$ has a non-negative weight.

For the opposite direction let us assume that $\tau$ is not a winning strategy.
Hence in there is a play $\pi$ such that $\LimInfAvg(\pi) < 0$.
Hence in the graph $A^\tau$ there exists a path $\pi^*$ such that $\LimInfAvg(\pi^*) < 0$.
Thus, by Lemma~\ref{lemm:NotEmptyIfCycle} there exists a cycle in $A^\tau$ with negative weight.
\pfbox
\end{proof}
}

\subsection{Decidability of the modular winning strategy problem}
In this section we will establish the decidability of the 
existence of a modular winning strategy problem.
In the following section we will establish the NP upper bound,
and finally show NP-hardness.
We start with the objective $\LimInfAvg\geq 0$, and then show the 
result for the objective $\LimSupAvg \geq 0$. 
The results for mean-payoff objectives with strict inequality will
also easily follow from our results.

\begin{remark}
We note that the reduction we presented in Section~\ref{sec:games} from the non-universality
problem of WFA to two-player mean-payoff pushdown games does not hold when we restrict player~1 to modular strategies.
Indeed, a modular strategy cannot fill the stack with an arbitrary long stack alphabet string without eventually 
visiting the same module twice (and then it must be the case that the operations are repeated forever).
Hence, player~1 cannot fill the stack with an arbitrary word to witness the non-universality of the WFA.
\end{remark}

\smallskip\noindent{\bf Objective $\LimInfAvg\geq 0$.}
For the decidability result, we will show the existence of 
\emph{cycle independent} modular winning strategies, and 
the result will also be useful to establish the complexity results.
\NewChange{
\smallskip\noindent{Informal description of the solution.}
We show that there is a winning strategy that is oblivious to cycles that are formed in the play.
Informally, this is true since there must exists a winning strategy in which all the cycles that are formed during the play has non-negative weight.
Intuitively, if a strategy cannot prevent a negative weight cycle to occur even once, then it cannot prevent it from occurring infinitely often, and thus it cannot prevent a play with negative mean-payoff.
Since the strategy may assume that all formed cycles have positive weight, it may ignore any formed cycle, as it should be able to win even in the scenario where this cycle never occurred, and the fact that it did occur only increase the total weighted sum of the play.
We formalize this intuition in the next lemmas:
In Lemma~\ref{lemm:ManIsWellDefinedAndConsistent} we show how to construct a cycle-free strategy from any arbitrary (not necessarily modular) winning strategy.
In Lemmas~\ref{lem:ManStrategySimPayoff} and~\ref{lem:LimInfImpliesManipLimSup} we show that if the original strategy is winning for the lim-inf objective, then the formed cycle-free strategy is winning for the lim-sup objective.
In Lemmas~\ref{lemm:StrategyIsModularAndCycleFree} and~\ref{lem:ModularStrategyInfImpCycleIndStrategiesSup} we show that if the original strategy is a modular one, then the formed strategy is also a modular winning strategy for the lim-inf objective, and thus a modular winning strategy exists if and only if there exists a cycle-free modular winning strategy.
}

We start with the notion of a cycle free path in a graph.

\smallskip\noindent{\bf Cycle free path.}
Let $G=(V,E)$ be a simple (no parallel edges) directed graph.  
We define the operator $\CycleFree : V^* \to V^*$ in the following way:
let $\pi = \atuple{v_1, v_2, \dots, v_n}$ be a finite path in $G$.
\begin{itemize}
\item $\CycleFree(\pi) = \pi$ if $\pi$ is a simple path (i.e., with no cycles).
\item Otherwise we define $\CycleFree$ inductively as follows.
Let $\CycleFree(v_1 \dots v_{n-1}) = u_1 u_2 \dots u_m$.
Let $i$ be the first index such that $v_n = u_i$.
If such an index does not exist, then $\CycleFree(\pi) = u_1 u_2 \dots u_m v_n$.
Otherwise $\CycleFree(\pi) = u_1 u_2 \dots u_i$.
Intuitively, the $\CycleFree$ operator takes a finite path and returns 
a simple path by removing the simple cycles according to the order of 
appearance.
\end{itemize}



\smallskip\noindent{\bf Cycle independent modular strategy.} 
Given a recursive game graph, a local strategy $\tau_i$ for module $A_i$ 
is  a \emph{cycle independent local strategy}, 
if for every $\rho \in V_i^*$ we have $\tau_i(\rho) = \tau_i(\CycleFree(\rho))$.
A modular strategy $\tau = \Set{\tau_i}_{i=1}^n$ is a 
\emph{cycle independent modular strategy} if $\tau_i$ is a cycle independent 
local strategy for every $i\in\RangeSet{1}{n}$.

\begin{Obs}\label{obs:NotManyCycleIndStrategies}
For a recursive game graph $\wrg = \atuple{A_1, \dots, A_n}$, 
there exist at most $|V|^{|V|^{|V|}}$ different cycle independent modular 
strategies, where $|V|$ is the number of vertices in $\wrg$. 
\end{Obs}

The main result of this section is that if there is a modular winning 
strategy, then there is a cycle independent modular winning strategy.
To establish the result we introduce the notion of \emph{manipulated} 
paths, using \emph{rewind}, \emph{fast forward}, and \emph{simulation} 
operations.

\smallskip\noindent{\bf Manipulated paths, rewind, fast forward, and simulation operations.}
Let $\tau = \Set{\tau_i}_{i=1}^n$ be a modular winning strategy for the objective 
$\LimInfAvg\geq 0$, and let $\epsilon > 0$ be an arbitrary constant.
Let $\pi_m = \pi_{m-1} \cdot n_i$ be a play prefix at position $m$, 
that ends at node $n_i \in A_i$ .
The \emph{manipulated play prefix} of $\pi_m$ according to $\tau$ and 
$\epsilon$, denoted by $\Manipulated(\pi_m)$, 
is defined inductively as follows:
Let $\Manipulated(\pi_{m-1})$ be the manipulated play prefix at position $m-1$.
Then $\Manipulated(\pi_{m})$ is obtained from $\Manipulated(\pi_{m-1})$ and $n_i$ by one of 
the following operations.
\begin{enumerate}
\item \emph{Rewind operation}:
The condition for the rewind operation is that 
$\CycleFree(\Manipulated(\pi_{m-1})) \cdot n_i$ closes a proper cycle in the 
top module $A_i$.
If the rewind condition holds, then $\Manipulated(\pi_{m})$ is formed from 
$\Manipulated(\pi_{m-1}) \cdot n_i$ by removing the proper cycle suffix from 
$\Manipulated(\pi_{m-1}) \cdot n_i$.
Intuitively the rewind operation \emph{rewinds} the path by chopping off 
the cycle in the end 
(we note that the cycle may not be simple\NewChange{, but it is unique}).

\item \emph{Fast forward operation:}
Let $h_0 = \Manipulated(\pi_{m-1}) \cdot n_i$.
The fast forward condition for a history $h$ that ends at node $n_i$ is 
as follows: 
there exists a play prefix $h \cdot \pi '(h)$ consistent with $\tau$ 
such that $n_i \cdot \pi '(h)$ is a proper cycle with average weight 
less than $-\epsilon$. 
In order to be precise, we define $\pi '(h)$ as the first such prefix according to 
the lexicographic ordering of the prefixes.
If the rewind condition does not hold, and the fast forward condition holds for $h_0$, 
then construct $h_1 = h_0 \cdot \pi'(h_0)$.
Continue the process and build $h_i= h_{i-1} \cdot \pi '(h_{i-1})$, 
as long as $h_{i-1}$ satisfies the fast forward condition.
If there exists a minimal index $i\in\Nat$ such that $h_i$ does not satisfy 
the fast forward condition, then we define $\Manipulated(\pi_{m}) = h_i$.
Otherwise, $\Manipulated(\pi_{m})$ is undefined (not well defined), 
and we say that the process is stuck in the fast forward operation.

\item \emph{Simulation operation:}
Else, if the rewind and the fast forward conditions do not hold, then we 
have $\Manipulated(\pi_{m}) = \Manipulated(\pi_{m-1}) \cdot n_i$.
\end{enumerate}

\Heading{Intuitive overview of the $\Manipulated$ operator.}
The $\Manipulated$ operator generates an \emph{alternative history} for the play.
The generated history does not contain cycles with average weight more than $-\epsilon$ (due to the rewind operations) 
and if the strategy $\tau$ allows a \emph{possible future} in which the play returns to the same position and 
the formed cycle has an average weight less than $-\epsilon$, 
then this possible future is added to the history (fast forward operation).
The alternative history has the following three key properties:
(i)~It is consistent with $\tau$, i.e., it could really have happened.
(ii)~If the average weight of the alternative history is is at least $-\epsilon$, then the average weight of the actual history is also at least $-\epsilon$. 
Hence, if player~1 wins in the play that is induced by the alternative history, then he also wins in the real play.
(iii)~When player~1 applies $\tau$ on the alternative history, all the cycles in the real history have average weight 
at least $-\epsilon$. 
Hence, player~1 does not need to \emph{remember} the actual cycles in the history (because in the worst case scenario 
they will simply occur again, and their weights are good for him) and he can play independently of the formed cycles.
The next example demonstrates a strategy according to manipulated history.

\begin{examp}\label{examp:Manip}
Consider the RSM shown in Figure~\ref{fig:ManExample}, and consider a 
player-1 modular strategy $\tau$ that follows the edge $v_1 \to v_3$ if $v_1$ 
is visited odd number of times (in the current invocation of $A_0$) and 
otherwise it follows the edge $v_1\to v_2$. 
In this strategy player~1 will play $v_1\to v_3$ in the first time $v_1$ is 
visited, $v_1 \to v_2$ in the second time, $v_1\to v_3$ in the third time, 
and so forth.
With this strategy player~1 can ensure a mean-payoff value of at least $0$.
We now illustrate a play according to the manipulated history for 
$\epsilon = \frac{1}{2}$.
The play begins by following $\En \to v_1$ and the $\Manipulated$ operator 
performs a simulation step. So the current manipulated history is $\En \to v_1$.
According to $\tau$, the next move for player~1 is $v_1\to v_3$ and if 
player~2 will then select $v_3 \to v_1$, then a cycle with average weight $-1 < -\epsilon$ will be formed.
Hence, a fast-forward step is made and the manipulated history is now 
$\En \to v_1\to v_3\to v_1$ (the \emph{real} history is $\En \to v_1$).
According to the manipulated history, $v_1$ has been visited twice, hence the next 
move for player~1 is $v_1 \to v_2$, and a corresponding simulation step is 
done for the manipulated history (which is currently 
$\En \to v_1 \to v_3 \to v_1\to v_2$).
We now consider that the next move for player~2 is $v_2 \to v_1$.
Hence, the new manipulated history is $\En\to v_1\to v_3\to v_1\to v_2\to v_1$, 
and since the suffix contains a cycle with average weight $1 > -\epsilon$, 
then a rewind operation is done and the manipulated history is (again) 
$\En\to v_1\to v_3\to v_1$, and therefore the next move for player~1 is 
(again) $v_1\to v_2$.
We now consider that the next move for player~2 is to invoke the module $A_0$.
This move is simulated in the manipulated history, and the play continues.

We observe that when playing according to the manipulated history the only 
\emph{real} move that player~1 will ever do is $v_1\to v_2$, and therefore the 
obtained strategy is cycle independent (although $\tau$ is not), and it is easy 
to verify that this strategy ensures a mean-payoff value of at least $-\epsilon$ 
(and in this example even a positive mean-payoff).

\begin{figure}[H]
\begin{center}
\begin{picture}(70,35)(16,-50)
\put(20,-22){\Large{$A_0$}}
\node[Nw=16.0,Nh=8.0,Nmr=1.0](n1)(68,-40.0){$A_0$}
\node[Nw=2.0,Nh=2.0,Nmr=1.0](n2)(60.0,-24.0){}

\node[Nw=16.0,Nh=8.0,Nmr=1.0](n3)(68.0,-24.0){$A_0$}
\node[Nw=2.0,Nh=2.0,Nmr=1.0](n4)(60.0,-40.0){}

\node[Nw=4.0,Nh=4.0,Nmr=2.0](Player1)(32.0,-32.0){$v_1$}

\node[Nw=4.0,Nh=4.0,Nmr=0.0](Player2Up)(46.0,-24.0){$v_2$}
\node[Nw=4.0,Nh=4.0,Nmr=0.0](Player2Down)(46.0,-40.0){$v_3$}

\node[Nw=2.0,Nh=2.0,Nmr=1.0](Entry)(17.455,-32.0){}

\drawedge(Entry,Player1){$0$}

\drawedge[curvedepth=4](Player1,Player2Up){$7$}
\drawedge[curvedepth=0](Player2Up,Player1){$-5$}

\drawedge(Player2Up,n2){$0$}

\drawedge[curvedepth=0](Player1,Player2Down){$9$}
\drawedge[curvedepth=4](Player2Down,Player1){$-11$}

\drawedge(Player2Down,n4){$-9$}

\node[Nw=70.0,Nh=35,Nmr=7.475](Frame)(52.455,-32.0){}
\end{picture}
\caption{RSM with only one module ($A_0$) and no exit nodes.
Player~1 controls the circle vertex and the rest of the vertices are controlled by player~2.}\label{fig:ManExample}
\end{center}
\end{figure}
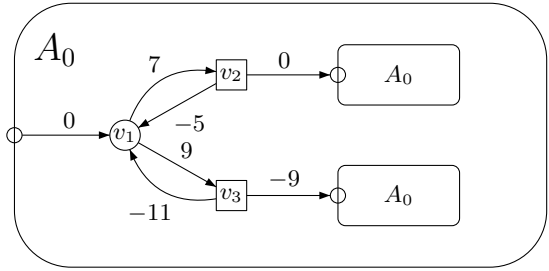

\end{examp}

In the following lemma we establish consistency of the manipulated operation for a winning strategy and the fact that it is well-defined.

\begin{lem}\label{lemm:ManIsWellDefinedAndConsistent}
Let $\tau$ be a winning strategy (a general winning strategy, not necessarily modular)
for the objective $\LimInfAvg \geq 0$.
Let $\epsilon>0$ be an arbitrary constant.
We define a strategy $\sigma$ in the following way: for a history $\pi'$ 
we have $\sigma(\pi') = \tau(\Manipulated(\pi'))$.
Let $\pi_m$ be a play prefix of length $m$ that is consistent with $\sigma$.
Then the following assertions hold:
\begin{itemize}
\item $\Manipulated(\pi_{m})$ is well defined, i.e., the process does not 
get stuck in the fast forward operation.
\item $\Manipulated(\pi_{m})$ is consistent with $\tau$.
\end{itemize}
\end{lem}
\begin{proof}
We prove both the items 
by induction on $m$. In the base case when $m = 0$ (i.e., empty play prefix), 
all the claims are trivially satisfied.
We now consider the inductive case with  $m > 0$.
Let $\pi_m = \pi_{m-1} \cdot n_i$, for some $n_i \in A_i$, be a play prefix
consistent with $\sigma$.
By the inductive hypothesis $\Manipulated(\pi_{m-1})$ is well defined and 
consistent with $\tau$. 
Then $\Manipulated(\pi_{m})$ is computed by performing one of the following operations.
\begin{itemize}
\item \emph{Rewind operation:}
In this case clearly $\Manipulated(\pi_{m})$ is well defined.
In addition $\Manipulated(\pi_{m})$ is a prefix of $\Manipulated(\pi_{m-1})$, 
which is by the inductive hypothesis consistent with $\tau$, 
hence also $\Manipulated(\pi_{m})$ is consistent with $\tau$.

\item \emph{Fast forward operation:} 
Towards contradiction, let us assume that the fast forward process enters an
infinite loop.
We consider the prefix $h_0 = \ov{h}_0 \cdot n_i$, where 
$\ov{h}_0=\Manipulated(\pi_{m-1})$, and let $\ov{v}_0$ be the last 
vertex in $\ov{h}_0$. 
The prefix $h_0$ is consistent with $\tau$ for the following reason: 
$\Manipulated(\pi_{m-1})$ is consistent with $\tau$ by the inductive 
hypothesis, and if $\ov{v}_0$ is under player~1's control, 
then $\tau(\Manipulated(\pi_{m-1})) = n_i$ (as $\pi_m$ consistent with $\sigma$), and 
otherwise $\ov{v}_0$ is under player~2's control, and since there is a 
transition from $\ov{v}_0$ to $n_i$ (as $\pi_m$ is a play prefix) 
the consistency of $h_0$ follows.
The prefix $h_0$ has an infinite sequence of extensions 
$\pi^1, \pi^2, \dots$ such that the infinite play 
$h = h_0 \pi^1 \pi^2 \dots \pi^j \pi^{j+1} \dots$ is consistent with $\tau$ and 
$\Avg(\pi^j) < -\epsilon$ for every $j\in\Nat$ (by the infinite loop
of the fast forward operation).
Hence, by definition, $\LimInfAvg(h) \leq -\epsilon < 0$.\footnote{
only this inequality need not hold for $\LimSupAvg(h)$} 
Thus we get that there exists a play consistent with $\tau$ that is not winning
for the objective $\LimInfAvg \geq 0$, 
which contradicts that $\tau$ is a winning strategy.
Hence, the fast forward process always terminates.
It follows that $\Manipulated(\pi_{m})$ is well defined and also 
by definition of the fast forward operation it is consistent with $\tau$. 

\item \emph{Simulation operation:}
By definition, $\Manipulated(\pi_m) = \Manipulated(\pi_{m-1}) \cdot n_i$.
Let $\ov{h}_0=\Manipulated(\pi_{m-1})$, and let $\ov{v}_0$ be the last 
vertex in $\ov{h}_0$. 
The prefix $\Manipulated(\pi_m)$ is consistent with $\tau$ for the following 
reason: $\Manipulated(\pi_{m-1})$ is consistent with $\tau$ by the inductive 
hypothesis, and if $\ov{v}_0$ is under player~1's control, 
then $\tau(\Manipulated(\pi_{m-1})) = n_i$, and 
otherwise $\ov{v}_0$ is under player~2's control, and there is a 
transition from $\ov{v}_0$ to $n_i$.
Thus we have the consistency of $\Manipulated(\pi_m)$, and the fact that it well-defined is trivial.
\end{itemize}
Hence we have that $\Manipulated(\pi_m)$ is both consistent with $\tau$ and 
well defined.
\hfill\qed
\end{proof}

In the following lemma we obtain a bound on the average of the 
play prefixes obtained from the manipulated operation of a winning 
strategy.

\begin{lem}\label{lem:ManStrategySimPayoff}
Let $\tau$ be a winning strategy (a general winning strategy, not necessarily modular)
for the objective $\LimInfAvg \geq 0$. 
Let $\epsilon>0$ be an arbitrary constant.
We define a strategy $\sigma$ in the following way: for a history $\pi'$ 
we have $\sigma(\pi') = \tau(\Manipulated(\pi'))$.
Let $\pi_m$ be a play prefix of length $m$ that is consistent with $\sigma$.
Then we have $w(\pi_m) \geq w(\Manipulated(\pi_m)) - \epsilon \cdot |\pi_m|$.
\end{lem}
\begin{proof}
The following claim is the key for the proof.

\smallskip\noindent{\em Claim.}
Every time a rewind operation is done, 
the cycle $C$, for which $\Manipulated(\pi_m) \cdot C = \Manipulated(\pi_{m-1}) \cdot n_i$, 
satisfies that $\Avg(C) \geq -\epsilon$. 

We first prove the claim.
Towards a contradiction, assume that $\Avg(C) < -\epsilon$.
Let $j < m$ be the first index for which $\Manipulated(\pi_j) = \Manipulated(\pi_m)$.
Recall that by definition a rewind operation was done if $\CycleFree(\Manipulated(\pi_{m-1}) \cdot n_i)$ forms a proper cycle that begins and ends in the node $n_i$. Hence, the beginning of the proper cycle was not generated by a fast-forward operation (otherwise, it would have been omitted by the $\CycleFree$ operator).
Thus, such an index $j$ must exist.
We first argue that $\Manipulated(\pi_m) \cdot C = \Manipulated(\pi_{m-1}) \cdot n_i$ 
is consistent with $\tau$: 
(i)~$\Manipulated(\pi_{m-1})$ is consistent with $\tau$ (by Lemma~\ref{lemm:ManIsWellDefinedAndConsistent});
and
(ii)~let $\ov{h}_0=\Manipulated(\pi_{m-1})$, and let $\ov{v}_0$ be the last 
vertex in $\ov{h}_0$; 
if $\ov{v}_0$ is under player~1's control, 
then $\tau(\Manipulated(\pi_{m-1})) = n_i$, and 
otherwise $\ov{v}_0$ is under player~2's control, and there is a 
transition from $\ov{v}_0$ to $n_i$.
Thus we have the consistency of 
$\Manipulated(\pi_m) \cdot C = \Manipulated(\pi_{m-1}) \cdot n_i$, and it 
follows that $\Manipulated(\pi_j) \cdot C$ is also consistent with $\tau$.
Hence, since $\Avg(C) < -\epsilon$, a fast forward operation had occurred in 
position $j$ 
(note that it is not possible that $\Manipulated(\pi_j)$ was obtained after a rewind operation, 
since $j$ is the first index for which $\Manipulated(\pi_j) = \Manipulated(\pi_m)$). 
So by definition, a fast forward operation and not a simulation operation must occur.
Hence it is not possible that $\Manipulated(\pi_j) = \Manipulated(\pi_m)$, since at the very least, 
$\Manipulated(\pi_m) \cdot C$ is a prefix of $\Manipulated(\pi_j)$.
Thus, for every $j < m$ we have $\Manipulated(\pi_j) \neq \Manipulated(\pi_m)$,
and the contradiction (to the existence of $j$) is obtained. 

We now complete the proof of the lemma using the claim.
We note that the difference between $\pi_m$ and $\Manipulated(\pi_m)$ 
contains only (i)~cycles with negative weight that were added to 
$\Manipulated(\pi_m)$ (by fast forward operation)
or (ii)~cycles with average weight at most $-\epsilon$
and length at most $|\pi_m|$ that were chopped from $\Manipulated(\pi_m)$
(by rewind operation).
The desired result follows. 
\hfill\qed
\end{proof}

We now show that from a winning strategy for the objective $\LimInfAvg \geq 0$,
the strategy obtained using manipulated operation is winning for the 
objective $\LimSupAvg \geq -2\cdot \epsilon$. 

\begin{lem}\label{lem:LimInfImpliesManipLimSup}
Let $\tau$ be a winning strategy (a general winning strategy, not necessarily modular)
for the objective $\LimInfAvg \geq 0$. 
Let $\epsilon>0$ be an arbitrary constant.
We define a strategy $\sigma$ in the following way: for a history $\pi'$ 
we have $\sigma(\pi') = \tau(\Manipulated(\pi'))$.
Then $\sigma$ is a winning strategy for the objective $\LimSupAvg \geq -2\cdot\epsilon$.
\end{lem}
\begin{proof}
Let $\pi$ be a play consistent with $\sigma$, and let $\pi_m$ be the play prefix until position $m$.
We consider two cases to complete the proof.
\begin{enumerate}
\item In the first case, 
there exists a constant $n_0\in\Nat$ such that for infinitely many indices 
$m_1, m_2, \dots$, we have $| \Manipulated(\pi_{m_i}) |\leq n_0$.
In this case, due to Lemma~\ref{lem:ManStrategySimPayoff}, in positions 
$m_1, m_2, \dots$ we get that  
$w(\pi_{m_i}) \geq -n_0\cdot W -\epsilon\cdot |\pi_{m_i}|$.
Hence, by definition, $\LimSupAvg(\pi) \geq -\epsilon > -2\cdot\epsilon$.

\item In the second case, for every $i > 0$ there exists $\ell_i \in \Nat$, 
such that for every $m > \ell_i$ we have $|\Manipulated(\pi_m) | \geq i$.
By the definition of the manipulation operations, we get that 
$\Manipulated(\pi_{\ell_i})[0,i] = \Manipulated(\pi_{\ell_{i}+1})[0,i]$, i.e., 
the prefix up to length $i$ coincides.
Denote $\rho_i = \Manipulated(\pi_{\ell_i})[i]$ the $i$-th position of 
$\Manipulated(\pi_{\ell_i})$. 
Due to Lemma~\ref{lemm:ManIsWellDefinedAndConsistent} the infinite play 
$\rho = \rho_1 \rho_2 \dots$ is consistent with $\tau$.
Since $\tau$ is a winning strategy we get that $\LimInfAvg(\rho) \geq 0$.
Hence there exists infinitely many indices $m_1, m_2, \dots$ for which 
$\Avg(\Manipulated(\pi_{m_i})) \geq -\epsilon$, and therefore, due to 
Lemma~\ref{lem:ManStrategySimPayoff}, we get that $\Avg(\pi_{m_i}) \geq -2\cdot\epsilon$.
Hence  by definition of $\LimSupAvg$, we obtain $\LimSupAvg(\pi) \geq -2\cdot \epsilon$.
\end{enumerate}

This concludes the proof of the lemma. 
\hfill\qed
\end{proof}

\begin{lem}\label{lemm:StrategyIsModularAndCycleFree}
Given a WRG $\wrg$, let $\tau$ be a modular winning strategy for the objective 
$\LimInfAvg \geq 0$.
Let $\epsilon>0$ be an arbitrary constant.
We define a strategy $\sigma$ in the following way: for a history $\pi'$ 
we have $\sigma(\pi') = \tau(\Manipulated(\pi'))$.
Then $\sigma$ is a cycle independent modular strategy.
\end{lem}
\begin{proof}
In order to verify that $\sigma$ is a modular strategy, 
we observe that $\LocalHistory(\Manipulated(\pi))$ only depends
on $\LocalHistory(\pi)$ and that $\tau$ is a modular strategy.

To verify that $\sigma$ is a cycle independent strategy, 
we observe that if $\pi$ and $\pi \cdot \pi_c$ are consistent 
with $\sigma$ and $\pi_c$ is a cycle in $\CycleFree(\pi) \cdot \pi_c$, 
then $\Manipulated(\pi \cdot \pi_c) = \Manipulated(\pi)$ as $\pi_c$ will be chopped
by the rewind operation. 
\hfill\qed
\end{proof}

\begin{lem}\label{lem:ModularStrategyInfImpCycleIndStrategiesSup}
Given a WRG $\wrg$, if there exists a modular winning strategy for the 
objective $\LimInfAvg \geq 0$ for player~1, then there exists a 
cycle independent modular winning strategy for the objective for player~1.
\end{lem}
\begin{proof}
Let $\tau$ be a modular winning strategy for the objective $\LimInfAvg \geq 0$.
For every $\epsilon > 0$, define the strategy $\sigma_\epsilon$ in the following way: 
for a history $\pi'$ we have
$\sigma_\epsilon(\pi') = \tau(\Manipulated(\pi'))$.
By Lemma~\ref{lem:LimInfImpliesManipLimSup} and 
Lemma~\ref{lemm:StrategyIsModularAndCycleFree}, 
for every $\epsilon > 0$ we have that $\sigma_\epsilon$ is a  
cycle independent modular winning strategy for the objective 
$\LimSupAvg > -2\cdot\epsilon$.
There are only a bounded number of cycle independent modular strategies
(by Observation~\ref{obs:NotManyCycleIndStrategies}), thus there is an optimal cycle independent modular strategy, and
it must be the case that it is a winning strategy for 
the $\LimSupAvg \geq 0$ objective (otherwise it does not win for the objective $\LimSupAvg > -\epsilon$ for some $\epsilon > 0$). 
Let $\sigma$ be that strategy.
Let $\wrg^\sigma$ be the player-2 WRG obtained given the strategy $\sigma$.
As $\sigma$ is a winning strategy for the objective $\LimSupAvg \geq 0$,
then due to Lemma~\ref{lemm:EveryPathIsAlsoInGraph} and Lemma~\ref{lemm:NotEmptyIfCycle}, 
the graph $\wrg^\sigma$ does not have a negative non-decreasing cycle.
Hence due to Lemma~\ref{lemm:EveryPathIsAlsoInGraph} and Lemma~\ref{lemm:NotEmptyIfCycle}, 
the strategy $\sigma$ is winning also for the objective $\LimInfAvg \geq 0$.
This completes the proof of the result.
\hfill\qed
\end{proof}

The next theorem is an immediate consequence of Lemma~\ref{lem:ModularStrategyInfImpCycleIndStrategiesSup}.

\begin{thm}\label{thm:ModularStartegyIsDecidable}
Given a WRG $\wrg$, the problem of deciding if player~1 has a modular winning strategy
for the objective $\LimInfAvg\geq 0$ is decidable.
\end{thm} 
\begin{proof}
By Lemma~\ref{lem:ModularStrategyInfImpCycleIndStrategiesSup} it is enough to check if player~1 
has a cycle independent modular strategy.
As the number of such strategies is bounded (by Observation~\ref{obs:NotManyCycleIndStrategies}), it is 
enough to construct the graph $\wrg^\tau$ for every cycle 
independent modular strategy $\tau$ and check if in $\wrg^\tau$ there 
exists a path $\pi$ with $\LimInfAvg(\pi) < 0$ (this check is achieved using 
the algorithms of Section~\ref{sect:PushDownProcesses}). 
\hfill\qed
\end{proof}

\newcommand{\PosCycleFree}{\mathsf{NonNegCFLocalHistory}^\tau}

\smallskip\noindent{\bf Objective $\LimSupAvg\geq0$.}
The proof of Lemma~\ref{lem:ModularStrategyInfImpCycleIndStrategiesSup} shows the following:
if player~1 has a modular winning strategy for $\LimInfAvg \geq 0$ objective, 
then he has a cycle independent modular winning strategy for the $\LimSupAvg\geq0$ objective,
which is also a cycle independent modular winning strategy for the $\LimInfAvg\geq0$ objective.
However, it does not imply that an arbitrary modular winning strategy for the 
$\LimSupAvg\geq0$ objective can be transformed into a winning strategy for the $\LimInfAvg\geq0$ objective, 
or that it can be transformed into a cycle independent modular strategy.  
The key ingredient that is missing is that it does not establish that if 
there is a modular winning strategy for $\LimSupAvg \geq 0$ objective,
then there is a cycle independent modular winning strategy for $\LimSupAvg \geq 0$ objective.
We establish this fact now.
The proof for the objective $\LimSupAvg\geq0$ will reuse many parts of the proof
for  the objective $\LimInfAvg \geq 0$, however, some parts of the proof are 
different and we present them below.
In fact we show that for modular winning strategies 
the objective $\LimSupAvg\geq 0$ coincides with 
the objective $\LimInfAvg\geq 0$. 
For the proof we need the notion of non-negative cycle free local 
history, which we define below.

\smallskip\noindent{\bf Non-negative cycle free local history operator.}
Consider a WRG $\wrg$ and let $\tau$ be a modular strategy, and 
$\pi$ be a path in $\wrg$, consistent with $\tau$, that begins at the entry 
of module $A_i$ and ends at module $A_i$ (in the same stack height).
The \emph{non-negative cycle free local history operator} is defined as follows:
\begin{enumerate}
\item For $|\LocalHistory(\pi)| = 1$ we have $\PosCycleFree(\pi) = \LocalHistory(\pi)$.
\item For $|\LocalHistory(\pi)| > 1$, let $\pi = \pi_0 v_i$ such that
$\PosCycleFree(\pi_0) = u_0 u_1 \dots u_m$.
Let $j\in\RangeSet{0}{m}$ be the first index such that $u_j = v_i$, and 
every sub-play $\pi^*$ consistent with $\tau$ with local history $u_j u_{j+1} \dots u_m v_i$ 
is a cycle with non-negative total weight.
If such index $j$ exists, then $\PosCycleFree(\pi) = u_0 \dots u_j$,
otherwise $\PosCycleFree(\pi) = u_0 u_1 \dots u_m v_i$.
\end{enumerate}
Informally, the $\PosCycleFree$ operator removes cycles that are ensured to have non-negative total weight from the local history.

\smallskip\noindent{\bf Non-negative cycle independent modular strategy.}
Given a modular strategy $\tau$, \emph{the non-negative cycle independent modular strategy} $\sigma$ of $\tau$, 
is defined as follows: For a local history $\rho$ we have $\sigma(\rho) = \tau(\PosCycleFree(\rho))$.
We now present some notations required for the proofs.

\smallskip\noindent{\em Sure non-negative cycle and proper simple cycle.}
Given a modular strategy $\tau$,  a path $\pi = v_0 v_1 \dots v_m$ is a \emph{sure non-negative cycle} 
if $\pi$ is a proper cycle consistent with $\tau$, and every path $\pi'$ consistent with $\tau$ such that 
$\LocalHistory(\pi) = \LocalHistory(\pi')$ is a proper cycle with non-negative weight.
A path $\pi$ is a \emph{proper simple cycle} if $\pi$ is a proper cycle, and $\LocalHistory(\pi)$ is a simple cycle.

\begin{lem}\label{lem:PositiveCycleIndIsEnoughForSup}
If $\tau$ is a modular winning strategy for the objective $\LimSupAvg \geq 0$,
then the non-negative cycle independent modular strategy $\sigma$ of $\tau$ is also 
a winning strategy for the objective $\LimSupAvg \geq 0$.
\end{lem}
\begin{proof}
The proof is essentially similar to the proof of the result for the 
objective $\LimInfAvg\geq 0$, 
and we present a succinct argument (as it is very similar to the previous proofs
for $\LimInfAvg\geq 0$).
As previously, we define the manipulated operations on histories such that every 
sure non-negative cycle is chopped (by the rewind operation) from the history.
By definition, the non-negative cycle independent strategy $\sigma$ of $\tau$ makes
choices according to the choices of $\tau$ on the manipulated history. 
The weight of the manipulated history, in every position, is at most the weight of 
the original history, since only non-negative cycles are chopped.
By similar arguments to those presented in Lemma~\ref{lem:LimInfImpliesManipLimSup}, 
we get that the $\LimSupAvg$ of the original history is at most $0$, and 
thus the non-negative cycle independent strategy $\sigma$ of $\tau$ is a winning strategy.
\hfill\qed
\end{proof}

In the following lemmas we establish that for modular strategies 
the objectives $\LimSupAvg\geq 0$ and $\LimInfAvg\geq 0$ coincide.

\begin{lem}\label{lem:NegativeCycleForPosFree}
Given a WRG $\wrg$, let $\sigma$ be a non-negative cycle independent modular strategy.
Let $A_i$ be a module in $\wrg$.
Then there exist $n_\sigma \in \Nat$ and $\delta_\sigma > 0$, such that for every 
possible non-negative cycle free local histories $h_0$ and $h_1$ of the module $A_i$, 
and for every sub-play $\rho$, consistent with $\sigma$ such that
\begin{itemize}
\item $\rho$ is a proper cycle with negative weight; and
\item the non-negative cycle free local history of $A_i$ before (resp. after) $\rho$ was played is $h_0$ (resp. $h_1$);
\end{itemize}
there exists a sub-play $\rho_{h_0,h_1}$, consistent with $\sigma$ and that satisfies the items above, 
such that every play prefix of $\rho_{h_0,h_1}$ with length at least $n_\sigma$ has an average weight 
of at most $-\delta_\sigma$.
\end{lem}
\begin{proof}
Let $(b_1,n_1), \dots ,(b_m,n_m)$ be the pairs of all boxes and their return nodes that appear in module $A_i$.
For every pair $(b_j,n_j)$, let $\pi_{b_j,n_j}$ be one of the shortest plays, consistent with $\sigma$, with minimal weight from $b_j$ to $n_j$.
If such a path exists we denote $w_{(b_j,n_j)} = w(\pi_{b_j,n_j})$.
Let $X$ be the maximal such weight.
W.l.o.g all the weights of the edges that occur in $A_i$ are at most $X$, and $X\geq 0$.

For every $b_j,n_j$ such that there exists a play from $b_j$ to $n_j$ consistent with $\sigma$, but a minimal-weight 
play does not exist, we denote by $\pi_{b_j,n_j}$ one of the shortest plays consistent with $\sigma$ that leads from 
$b_j$ to $n_j$ with weight at most $-20\cdot |V_i| \cdot X$, (where $|V_i|$ is the size of the vertex set in $A_i$).
Let $\Pi$ be one of the longest plays among all $\pi_{b_j,n_j}$.
Let $\rho$ be a sub-play consistent with $\sigma$ such that $\rho$ is a proper cycle in module $A_i$ 
with negative weight.
Let $h_0$ (resp. $h_1$) be the non-negative cycle free history of $A_i$ before (resp. after) playing $\rho$.

First, we form $\rho '$ from $\rho$ by removing all the sure non-negative cycles from $\rho$.
Clearly, (i)~the sum of weights of $\rho '$ is negative; and (ii)~$\rho '$ is consistent with $\sigma$, because 
$\rho$ is consistent with $\sigma$ and $\sigma$ is a non-negative cycle independent strategy.
Moreover, the non-negative cycle free local history after playing $\rho '$ is the same as after 
playing $\rho$.
Next we form $\rho ''$ from $\rho '$ by replacing every sub-play from $b_j$ to $n_j$ (such that 
$n_j$ is the first time the sub-play enters $A_i$) in $\rho '$ with $\pi_{b_j,n_j}$.
Again, $\rho ''$ is consistent with $\sigma$, since $\sigma$ is a modular strategy.
In addition, the local history of $A_i$ was not changed at all.

We claim that $\rho ''$ does not contain simple proper non-negative cycles in module $A_i$.
Indeed, towards a contradiction let $v_1 v_2 \dots v_m$ be the local history of the first such cycle 
(note that $m \leq |V_i|$).
Note that if the cycle contains a sub-play $\pi_{b_j,n_j}$ such that $w_{(b_j,n_j)} = -20\cdot |V_i|\cdot X$, 
then the sum of the weights of the cycle cannot be non-negative.
Hence it follows that this cycle is the cycle with minimal weight among all simple cycles with 
local history $v_1 v_2 \dots v_m$.
Hence this cycle is sure non-negative, which contradicts the fact that $v_1 v_2 \dots v_m$ 
is a local history of sub-play of $\rho ''$.
Thus for every simple proper cycle in $\rho ''$ the sum of the weights of the cycle is negative.
Moreover, the length of every simple proper cycle in $\rho ''$ is at most $|V_i|\cdot |\Pi|$.
Hence every sub-play of $\rho ''$ with length at least $n_\sigma = (|V_i|\cdot |\Pi|\cdot X)^2$ 
will have an average weight of at most $-\delta_\sigma = -\frac{1}{|V_i|\cdot |\Pi|\cdot X}$.
Note that $n_\sigma$ and $\delta_\sigma$ do not depend on $\rho$, and hence the desired result follows. 
\hfill\qed
\end{proof}

\begin{lem}\label{lem:NoInfiniteNegativeAvgCycles}
Let $\sigma$ be a non-negative cycle independent modular strategy.
If there exists a play $\rho$ consistent with $\sigma$ such that 
the suffix of $\rho$ is an infinite sequence of proper cycles $C_1, C_2, \dots$ 
with negative weights, then $\sigma$ is not a winning strategy
for the objective $\LimSupAvg \geq 0$.
\end{lem}
\begin{proof}
Let $\rho = \rho_0 C_1 C_2 \dots C_i \dots$,
and let $A_i$ and $n_i$ be the module and the vertex, respectively, 
that all the cycles begin and end in.
Let $n_\sigma$, $\delta_\sigma$ be the constants from Lemma~\ref{lem:NegativeCycleForPosFree}.
Let $h^i_0$ be the non-negative cycle free local history of $A_i$ before cycle $C_i$ is played, 
and $h^i_1$ be the non-negative cycle free local history of $A_i$ after cycle $C_i$ is played.
By Lemma \ref{lem:NegativeCycleForPosFree}, for every cycle $C_i$ there exists a cycle $C_{h^i_0,h^i_1}$ 
consistent with $\sigma$, with negative weight, and the average weight of every sub-play of 
$C_{h^i_0,h^i_1}$ longer than $n_\sigma$ is at most $-\delta_\sigma$.
The play $\rho' = \rho_0 C_1 C_2 \dots C_{i-1}  C_{h^i_0,h^i_1} C_{i+1} \dots$ is consistent with 
$\sigma$, since $C_{h^i_0,h^i_1}$ is consistent with the initial non-negative cycle free 
local history $h^i_0$, and the non-negative cycle free local history after playing $C_{h^i_0,h^i_1}$ is $h^i_1$.
Since $\sigma$ is a non-negative cycle independent strategy,
the play $\rho^* = \rho_0 C_{h^0_0,h^0_1} C_{h^1_0,h^1_1} \dots C_{h^i_0,h^i_1} \dots$ is 
consistent with $\sigma$.
On the other hand, it is straightforward to verify that we have 
$\LimSupAvg(\rho^*) \leq \max\Set{-\frac{1}{n_\sigma},-\delta_\sigma} < 0$.
Hence $\sigma$ is not a winning strategy for the objective $\LimSupAvg \geq 0$,
and the desired result follows.
\hfill\qed
\end{proof}

\begin{lem}\label{lem:CycleIndAreEnoughAlsoForLimSup}
If player~1 has a modular winning strategy for the objective $\LimSupAvg \geq 0$, 
then for every $\epsilon > 0$, player~1 has a cycle independent modular winning strategy 
$\sigma$ for the objective $\LimSupAvg \geq -\epsilon$.
\end{lem}
\begin{proof}
Let $\tau^*$ be a modular winning strategy for the objective $\LimSupAvg \geq 0$, and 
let $\tau$ be the non-negative cycle independent strategy of $\tau^*$.
For $\epsilon>0$, consider the cycle independent modular strategy $\sigma$ such that for 
histories $\pi$ we have \NewChange{$\sigma(\pi) = \tau(\Manipulated(\pi))$.}
By Lemma~\ref{lem:NegativeCycleForPosFree} the strategy $\tau$ is also a winning strategy.
We note that all the arguments for the objective $\LimInfAvg\geq 0$ also hold for 
the objective $\LimSupAvg\geq 0$, other than the inequality (mentioned as footnote) in
Lemma~\ref{lemm:ManIsWellDefinedAndConsistent}.
We  replace the inequality (mentioned as footnote) of Lemma~\ref{lemm:ManIsWellDefinedAndConsistent} 
with Lemma \ref{lem:NoInfiniteNegativeAvgCycles}, and repeat exactly the same arguments for the 
objective $\LimInfAvg\geq 0$ (up to Lemma~\ref{lem:ModularStrategyInfImpCycleIndStrategiesSup}) 
to obtain the desired result. 
\hfill\qed
\end{proof}

We obtain the following result as a corollary, and all the desired results follow for the 
objective $\LimSupAvg\geq 0$.

\begin{cor}\label{cor:ModularSupIffInf}
Given a WRG $\wrg$, there exists a modular winning strategy for the objective 
$\LimSupAvg \geq 0$ iff there exists a modular winning strategy for the objective 
$\LimInfAvg \geq 0$.
\end{cor}

\smallskip\noindent{\em Comment on $\LimSupAvg$ and $\LimInfAvg$ coincide.}
We have showed  in Example~\ref{Ex:SupAndInfAreDifferent} that in pushdown 
games with arbitrary strategies $\LimSupAvg$ and $\LimInfAvg$ do not coincide, 
whereas in contrast in Corollary~\ref{cor:ModularSupIffInf} we show that 
they coincide for modular strategies. 
The key reason that they do not coincide for arbitrary strategies is due to 
infinite-memory strategies, which makes it possible to construct paths 
(that are not ultimately periodic) where $\LimSupAvg$ and $\LimInfAvg$ do not 
coincide.
Our results for WPS (weighted pushdown systems, or only one-player pushdown
games) show that if there is only one-player (the opponent) then $\LimSupAvg$
and $\LimInfAvg$ coincide, even though infinite-memory strategies can be used.
Our results for modular strategies establish that even if the opponent is allowed
infinite-memory strategies, memoryless modular winning strategies exist. 
Since once a memoryless modular winning strategy is fixed, we obtain a WPS
where $\LimInfAvg$ and $\LimSupAvg$ coincides, we obtain the result
of Corollary~\ref{cor:ModularSupIffInf}.

\subsection{Modular winning strategy problem in NP}
In this section we will show that the modular winning 
strategy problem is in NP. 
\NewChange{
\smallskip\noindent{\bf Informal description of the solution.}
Our solution relays on the notion of \emph{strategy signature}.
The signature of a strategy is defined for every couple of nodes (in the same module) and stands for the weight of the worst-case path (w.r.t all possible histories) that exists between the two nodes, which is consistent with the strategy.
The intuition is that if two strategies has the same signature, then they obtain the same mean-payoff value.
This intuition need not hold in the general case, but we prove that it does hold if one strategy is modular and cycle-independent and the second strategy is memoryless.
Finally, the NP solution is to guess a signature and a memoryless strategy that realizes the signature, and to verify (in polynomial time) that the memoryless strategy is winning.
We formalize the intuition with the next lemmas:
In Lemma~\ref{lemm:WinInLocalIffWinInReduction} we show how to verify that a signature can be realized,
and in Lemma~\ref{lem:WinIffMemorylessStrategy} we show that a realizable signatures can be realized by memoryless strategies.
In Lemma~\ref{lemm:NoMinusInfInSig} we show that a cycle-independent modular winning strategy has a signature of polynomial size (and thus, we can guess one), and in Lemmas~\ref{lem:SignatureStrategy}-\ref{lem:MemorylessWinningStrategy} we prove that a memoryless strategy that has the same signature as a cycle-independent modular winning strategy is also a winning strategy.
}

\smallskip\noindent{\bf The signature game.}
Let $G=((V,E),(V_1,V_2))$ be a finite two-player game graph (on 
finite directed graph $(V,E)$) 
with vertex set $V$, edge set $E$,
and partition $(V_1,V_2)$ of the vertex set into player-1 (resp. player-2) 
vertex set  $V_1$ (resp. $V_2$).
Let the game graph be  equipped with a weight function 
$w : E \to \Z \cup \Set{-\omega}$. 
Let the initial vertex be $v_0 \in V$.
Let $\vec{\nu} = (\nu_0 , \dots, \nu_{|V|-1})$ be a threshold vector 
such that all $\nu_i \in \Z\cup \Set{+\infty, -\omega}$ (equivalently, $\vec{\nu}$ is a mapping from $V$ to $\Z \cup \Set{+\infty,-\omega}$\NewChange{ with $\vec{\nu}(v_i) = \nu_i$}).
The weight of a finite path in $G$ is the sum of the edge weights
of the path,
according to the following convention: $-\omega + z = -\omega$, for any $z\in\Z\cup \Set{-\omega}$.
A \emph{signature game} is defined with respect to the mapping vector 
$\vec{\nu}$ and consists of a tuple $(G,\vec{\nu})$, where $G$ is a two-player 
game graph, and $\vec{\nu}$ is the threshold vector.
For a play $\rho = \rho_0 \rho_{1} \rho_{2} \dots \rho_{j} \rho_{j+1} \dots$, 
player~1 is the winner if both the following two conditions hold:
\begin{itemize}
\item The play $\rho$ does not contain a negative cycle, or a cycle that has an edge with weight $-\omega$.
\item For every $j\in\Nat$, 
$w(\rho_0  \rho_{1} \rho_{2} \dots \rho_{j}) \geq \NewChange{\vec{\nu}(\rho_j)}$, 
according to the convention (i)~$-\omega < z$ for every $z\in\Z$, and
(ii)~$ z, -\omega < +\infty$ for every $z\in\Z$. 
In other words, the sum of the weights up to any index $j$ must be at least \NewChange{$\vec{\nu}(\rho_j)$}, and no $-\omega$-edge must be visited 
unless \NewChange{$\vec{\nu}(\rho_j)=-\omega$}.
\end{itemize}
We will consider signature games such that if $\nu_i \in \Z$, then $\nu_i 
\geq -2\cdot W\cdot |V|$, where $W$ is the maximum absolute values of the 
integer weights in graph $G$.
If for a vertex $v$, the threshold value is $+\infty$, then to ensure winning 
player~1 must ensure that $v$ is never visited. 
In other words, vertices with $+\infty$ threshold must be avoided (it can be 
interpreted as a safety objective with $+\infty$ vertices as non-safe 
vertices to be avoided).  
Our first goal is to show that memoryless winning strategies exist
in signature games, and the result will be obtained by a reduction 
to finite-state mean-payoff games. 
Given a signature game $(G,\vec{\nu})$ we define an auxiliary 
finite-state mean-payoff game as follows.

\smallskip\noindent{\bf From signature game $(G,\vec{\nu})$ to 
auxiliary mean-payoff game $\ov{G}_{\vec{\nu}}$.}
Given a signature game $(G,\vec{\nu})$ we construct a finite-state  
auxiliary mean-payoff game $\ov{G}_{\vec{\nu}}=((\ov{V},\ov{E}),(\ov{V}_1,\ov{V}_2))$, 
with a weight function $\ov{w}$ as follows: let the signature game 
graph be $G=((V,E),(V_1,V_2))$ and the weight function in $G$ be $w_G$. 
Then we have the following components in the auxiliary game:
\begin{itemize}
\item \emph{(Vertex set and partition).} 
$\ov{V} = V \times \Set{1,2}$; and $\ov{V}_1 = V_1 \times \Set{1}$. 
\item \emph{(Edges).}
$\ov{E} = \Set{((u,1),(v,2)) \mid (u,v) \in E} 
\cup \Set{((v,2),(v,1)) \mid v\in V} \cup \Set{((v_i,2),(v_0,1)) \mid v_i\in V, \nu_i \neq -\omega}$.
\item \emph{(Weight function).} 
If $w_G(u,v) \neq -\omega$, then $\ov{w}((u,1),(v,2)) = w_G(u,v)$; 
otherwise (we have $w_G(u,v) = -\omega$) we set  
$\ov{w}((u,1),(v,2)) = -10\cdot W\cdot |V|$.
Moreover, $\ov{w}((v_i,2),(v_0,1)) = -\nu_i$ and all the other edges 
are assigned with zero weight.
Note that if $\nu_i=+\infty$, then $\ov{w}$ assigns weight $-\infty$, and 
to win player~1 must avoid such edges (can be interpreted as a safety 
condition).
\end{itemize}
Informally, the auxiliary mean-payoff game is constructed from the signature 
game by adding for every vertex $v_i$ a fresh copy (vertex $(v_i,2)$ and the 
original vertex is represented as $(v_i,1)$), and an option for player~2 to 
return to the initial vertex  $(v_0,1)$ 
``paying" cost $-\nu_i$ (whenever $\nu_i \neq -\omega$).
Thus, if at any position of the play the current vertex is $(v_i,2)$, and 
the sum of the weights since the last visit to $(v_0,1)$ is less than 
$\nu_i$, then player~2 can ensure that a negative cycle is completed.
The mean-payoff objective of player~1 is to ensure non-negative 
average payoff.
Also note that player~1 must avoid the $-\infty$ edge weights, and equivalently
it can be treated as a mean-payoff safety game.
In the following lemma we establish the relation of the 
signature game and the auxiliary game.

\begin{lem}\label{lemm:WinInLocalIffWinInReduction}
Let $(G,\vec{\nu})$ be a signature game such that 
$\nu_i \geq -2\cdot W\cdot |V|$ for every $\nu_i \in \Z$, and
let $\ov{G}_{\vec{\nu}}$ be the corresponding auxiliary mean-payoff 
game.
Then the following statements are equivalent:
\begin{enumerate}
\item Player~1 is the winner in the auxiliary mean-payoff game $\ov{G}_{\vec{\nu}}$
(i.e., player~1 can ensure non-negative mean-payoff).
\item Player~1 has a memoryless winning strategy in the signature game  $(G,\vec{\nu})$.
\item Player~1 is the winner in the signature game $(G,\vec{\nu})$.
\end{enumerate}
\end{lem}
\begin{proof}
We first prove that item~1 implies item~2.
\begin{enumerate}
\item 
In order to prove that item~1 implies item~2, 
let us assume that player~1 is the winner in the auxiliary mean-payoff game.
Note that the auxiliary mean-payoff game is equivalently a finite-state mean-payoff safety game, and
therefore player~1 has a memoryless winning strategy $\tau$ in the 
auxiliary mean-payoff game~\cite{EM79}
(a mean-payoff safety game is easily transformed to a mean-payoff game by making the
non-safe vertices absorbing with negative weights).
Hence the memoryless strategy $\tau$ ensures that edges with weight $-\infty$ are
never visited.
Towards a contradiction, let us assume that $\tau$ is not a winning strategy 
for the signature game (note that $\tau$ is also a well-defined player-1 strategy 
in the signature game choosing edges in copy~1 according to $\tau$).
Therefore one of the following two cases occur.

\begin{itemize}
\item Case~1: There exists a finite play prefix $\rho$ 
that is consistent with $\tau$, which starts from the initial vertex $v_0$ to some 
vertex $v_i \in V$ with sum of weights less than $\nu_i$.
In this case, either $\rho$ goes through an $-\omega$ edge, or 
$w_G(\rho) < \nu_i$.
If $\rho$ goes through an $-\omega$ edge in $G$, 
then the weight of $\rho$ in $\ov{G}_{\vec{\nu}}$ is at most 
$-9\cdot W\cdot |V|$, since w.l.o.g we can assume that $\rho$ does not have 
positive cycles (as $\tau$ is memoryless).
As $-9\cdot W\cdot |V| < \nu_i$, 
it follows that the path $(\rho \cdot ((v_i,2),(v_0,1)))^\omega$ is 
consistent with $\tau$ and has a negative mean-payoff in the auxiliary game.
This contradicts the assumption that $\tau$ is a winning strategy.
If $\rho$ does not go through an $-\omega$ edge, then $w_G(\rho) < \nu_i$ and again 
$(\rho \cdot ((v_i,2),(v_0,1)))^\omega$ is consistent with $\tau$ and has a negative mean-payoff
in the auxiliary game. 
This is again a contradiction that $\tau$ is a winning strategy, 
and concludes the proof of the first case.

\item
Case~2: There exists a finite play prefix $\rho = \rho_1 \cdot \rho_2$ 
that is consistent with $\tau$, such that $\rho_2$ is a negative cycle 
(or a cycle with $-\omega$ edge) in the signature game graph.
If $\rho_2$ does not contain an $-\omega$ edge $e$, 
then by definition, $\rho_2$ is a negative cycle also in the auxiliary game.
Otherwise, $\rho_2$ contains an $-\omega$ edge $e$, and then 
again $\rho_2$ is a negative cycle in the auxiliary game, 
as w.l.o.g we can assume that $\rho_2$ does not contain positive cycles, 
and since $\ov{w}(e) \leq -10\cdot W \cdot |V|$.
Thus the play $\rho_1 \cdot (\rho_2)^\omega$ is consistent with $\tau$ and 
has a negative mean-payoff in the auxiliary game.
This contradicts the assumption that $\tau$ is a winning strategy,
and completes the proof.
\end{itemize}

\item Item~2 trivially implies item~3.

\item We now show that item~3 immediately implies item~1.
It is straightforward to verify that if player~1 plays according 
to the signature game winning strategy in every position, then a negative 
cycle will not be formed in the auxiliary game (as a negative cycle 
is not formed in the signature game, and the threshold vector is
always satisfied in the signature game) and a vertex with threshold 
$+\infty$ will never be reached.
Hence the mean-payoff of the play in the auxiliary mean-payoff game will 
be non-negative and $-\infty$ edges will never be visited.
This shows that item~3 implies item~1.
\end{enumerate}
This completes the proof.
\hfill\qed
\end{proof}

\begin{lem}\label{lem:WinIffMemorylessStrategy} 
Let $(G,\vec{\nu})$ be a signature game such that 
$\nu_i \geq -2\cdot W\cdot |V|$ for every $\nu_i \in \Z$.
There is a winning strategy for player~1 in the
signature game $(G,\vec{\nu})$ iff player~1 has a memoryless winning strategy.
\end{lem}
\begin{proof}
Follows from Lemma~\ref{lemm:WinInLocalIffWinInReduction}.
\hfill\qed
\end{proof}

\begin{remark}
The result of Lemma~\ref{lem:WinIffMemorylessStrategy} holds for all thresholds 
$\nu_i$, but we consider $\nu_i \geq -2 \cdot W\cdot |V|$ for simplicity of the 
proof.
\end{remark}

\smallskip\noindent{\bf Signature games to memoryless modular strategies.}
We will now use the existence of cycle independent modular winning
strategies, and memoryless winning strategies in signature games to 
show the existence of memoryless modular strategies.
For simplicity we will consider recursive game graphs where every 
module has a single entry, and a simple polynomial reduction from 
multi-entry recursive game graphs to single entry recursive game
graphs is established in~\cite{AlurReach}.
To prove the result of memoryless modular strategies we define 
the signature games for modular strategies.

\smallskip\noindent{\bf Signature games for modular strategies.}
Consider a WRG $\wrg=\atuple{A_1, A_2,\dots,A_n}$ and 
let $\tau=\Set{\tau_i}_{i=1}^n$ be a modular strategy.
Consider a module $A_i$ in $\wrg$.
Let $b \in B_i$ be a box in module $A_i$, 
which invokes the module $A_j$, and 
let $n_i \in N_i$ be a node in module $A_i$ that is connected to 
one exit of $A_j$, which is reachable according to the strategy $\tau$.
We denote by $w^{\tau}_{b,n_i}$ the minimal weight of all plays according 
to $\tau$ that begins at the call to box $b$ and ends at $n_i$ 
(in the same stack height), and do not visit any other vertices in 
$A_i$ (in the same stack height).
If such a minimal-weight play does not exist, then let $w^{\tau}_{b,n_i} = -\omega$.
For every module $A_i$, we form a finite-state two-player game graph $G_{A_i}$, 
with a weight function as follows: 
(i)~in the module $A_i$ we add an edge from every box $b$ to every return node 
$n_i$ with weight $w^{\tau}_{b,n_i}$, and add a self loop, with weight $0$, 
to every exit node; and (ii)~every box is now interpreted as a player-2 
vertex.
Note that the local strategy $\tau_i$ is a well-defined player-1 
strategy in the game graph $G_{A_i}$.
For a vertex $v \in V_i$,
(i)~if $v$ is visited along a play consistent with $\tau_i$, then
let $\eta_v$ denote the maximal value such that in every position of a 
play according to $\tau_i$ on $G_{A_i}$, that begins in the entry node of $A_i$, 
and is currently at vertex $v \in V_i$, the sum of weights from the beginning 
of the play is at least $\eta_v$; and 
(ii)~otherwise, $v$ is never visited along all plays 
consistent with $\tau_i$, then $\eta_v=+\infty$ (note that this is like a 
safety condition to ensure that $v$ is not visited).
The \emph{signature game for $\tau$ on module $A_i$} 
consists of the game graph $G_{A_i}$ and the threshold vector $\vec{\nu}^i \in 
(\Set{-\omega, +\infty}\cup\Z)^{|V_i|}$, 
such that $\nu^i_v=\eta_v$ for all vertices $v \in V_i$.
We denote by $(G_{A_i},\vec{\nu}^i,\tau)$ the signature game obtained
given the modular strategy $\tau$ on module $A_i$.
\NewChange{We first give an example that illustrates the connection between signature games and cycle independent modular winning strategies.}
We establish some basic properties in Lemma~\ref{lemm:NoMinusInfInSig},
and then in Lemma~\ref{lem:SignatureStrategy} we establish properties of 
the winning strategies in the signature games from modular strategies.

\NewChange{
\begin{examp}\label{examp:4_2}
Consider the RSM $\langle A_0, A_1 \rangle$ (shown in 
Figure~\ref{fig:StraExample}) and a player-1 modular strategy 
$\tau = \{\tau_0, \tau_1\}$ such that in module $A_0$ the strategy $\tau_0$ 
always selects $v_4\to v_5$ if the play visited $v_3$ and otherwise it always 
invokes $A_1$, and in module $A_1$, the strategy $\tau_1$ selects the upper 
exit (denoted by $\Ex_1$) if $u_3$ is visited (in the current invocation of 
$A_1$) and otherwise it selects the lower exit (denoted by $\Ex_2$).
The strategy $\tau$ is a cycle independent strategy, but it is not a 
memoryless strategy.
In a play according to $\tau$ if the upper exit of $A_1$ is reached, then the 
path $\En_1 \to u_1\to u_3\to \Ex_2$ with weight $-1$ is played, and if the 
lower exit is reached, then the weight of the sub-play is $1$.
Hence, if in module $A_0$, vertex $v_4$ invokes $A_1$, then if the upper exit 
of $A_1$ is reached, then the play continues to $v_5$ and from there to $v_4$ 
and a cycle with weight $0$ is formed.
If the lower exit of $A_1$ is reached, then the play continues to $v_4$ and 
a cycle with weight $1$ is formed.
If in $v_4$ the player-1 move is $v_4 \to v_5$, then the play continues to 
$v_4$, and a cycle with weight $0$ is formed.
Therefore, the strategy $\tau$ ensure that mean-payoff is at least $0$.

The corresponding signature game is illustrated in 
Figure~\ref{fig:StraToSigExample}. Note that the box that invokes $A_1$ is 
replaced by $b_{A_1}$.
We note that player~1 has to decide on the next move only in $u_4$ in $A_1$ 
and in $v_4$ in $A_0$.
Hence, the strategies $\tau_0$ and $\tau_1$ are well defined over 
$G_{A_0}$ and $G_{A_1}$, respectively.
The strategy $\tau_1$ ensures the following signature over $G_{A_1}$:
$\nu^0  = 0,
\nu^1 = 0,
\nu^2 = 7,
\nu^3 = 6,
\nu^4 = 6,
\nu^5 = -1,
\nu^6 = 1$ and the strategy $\tau_0$ ensures the following signature over 
$G_{A_0}$:
$\nu^0 = 0,
\nu^1 = 0,
\nu^2 = 1,
\nu^3 = -1,
\nu^4 = 2,
\nu^5 = 7,
\nu^{b_{A_i}} = 2,
\nu^6 = 1,
\nu^7 = 3$.
The same signatures are satisfied by a memoryless strategy that always selects 
$u_4\to u_5$ in $G_{A_1}$ and $v_4 \to b_{A_1}$ in $G_{A_0}$.
The corresponding memoryless strategy in $A_1$ is to select the upper exit in 
$u_4$ and in $A_0$ is to invoke $A_1$ when in $v_4$.
It is easy to verify that this memoryless strategy is a winning strategy in 
the recursive game for the mean-payoff objective.

\begin{figure}[H]
\begin{center}

\begin{picture}(70,35)(16,-50)



\node[Nw=4.0,Nh=4.0,Nmr=0.0](Player2)(32.0,-32.0){$u_1$}
\node[Nw=4.0,Nh=4.0,Nmr=0.0](Player2Up)(46.0,-24.0){$u_2$}
\node[Nw=4.0,Nh=4.0,Nmr=0.0](Player2Down)(46.0,-40.0){$u_3$}

\node[Nw=4.0,Nh=4.0,Nmr=2.0](Player1)(60,-32.0){$u_4$}

\put(20,-26.5){\Large{$A_1$}}

\node[Nw=2.0,Nh=2.0,Nmr=1.0](Entry)(18,-32.0){}

\node[Nw=2.0,Nh=2.0,Nmr=1.0](ExUp)(74,-28.0){}
\node[Nw=2.0,Nh=2.0,Nmr=1.0](ExDown)(74,-36.0){}

\node[Nw=56.0,Nh=25,Nmr=7.475](Frame)(46,-32.0){}

\drawedge(Entry,Player2){$0$}
\drawedge(Player2,Player2Up){$7$}
\drawedge(Player2,Player2Down){$6$}
\drawedge(Player2Down,Player1){$0$}
\drawedge(Player2Up,Player1){$0$}
\drawedge(Player1,ExUp){$-7$}
\drawedge[ELside=r](Player1,ExDown){$-6$}

\end{picture}

\begin{picture}(70,35)(16,-50)

\put(20,-24){\Large{$A_0$}}
\node[Nw=2.0,Nh=2.0,Nmr=1.0](Entry)(18,-32.0){}

\node[Nw=4.0,Nh=4.0,Nmr=0.0](Player2)(32.0,-32.0){$v_1$}
\node[Nw=4.0,Nh=4.0,Nmr=0.0](Player2Up)(46.0,-24.0){$v_2$}
\node[Nw=4.0,Nh=4.0,Nmr=0.0](Player2Down)(46.0,-40.0){$v_3$}

\node[Nw=4.0,Nh=4.0,Nmr=2.0](Player1)(60,-32.0){$v_4$}

\node[Nw=16.0,Nh=8.0,Nmr=1.0](Box)(82,-24.0){$A_1$}
\node[Nw=2.0,Nh=2.0,Nmr=1.0](EntryBox)(74,-24.0){}
\node[Nw=2.0,Nh=2.0,Nmr=1.0](BoxEx1)(90,-22.0){}
\node[Nw=2.0,Nh=2.0,Nmr=1.0](BoxEx2)(90,-26.0){}

\node[Nw=4.0,Nh=4.0,Nmr=0.0](Player2Last)(74.0,-40.0){$v_5$}

\drawbpedge[ELpos=70](BoxEx2,-45,3,Player1,0,40){$0$} 

\drawbpedge[ELpos=70](BoxEx1,-20,7,Player2Last,0,25){$6$}

\drawedge[curvedepth=4](Player2Last,Player1){$-5$}


\node[Nw=2.0,Nh=2.0,Nmr=1.0](Exit)(96,-32.0){}

\node[Nw=78.0,Nh=30,Nmr=7.475](Frame)(57,-32.0){}

\drawedge(Entry,Player2){$0$}
\drawedge(Player2,Player2Up){$1$}
\drawedge(Player2,Player2Down){$-1$}
\drawedge(Player2Down,Player1){$3$}
\drawedge(Player2Up,Player1){$7$}
\drawedge(Player1,EntryBox){$0$}
\drawedge[ELside=r](Player1,Player2Last){$5$}

\end{picture}

\caption{RSM with two modules ($A_0$ and $A_1$).
Player~1 controls the round vertices and the rest of the vertices are controlled by player~2.}\label{fig:StraExample}
\end{center}
\end{figure}
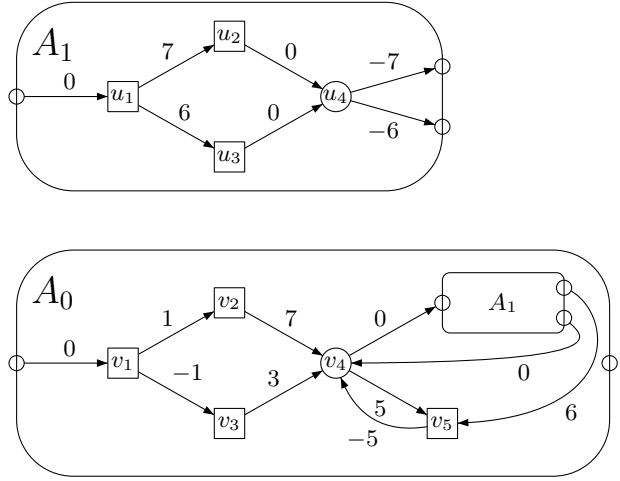

\begin{figure}[H]
\begin{center}

\begin{picture}(70,20)(16,-42)



\node[Nw=4.0,Nh=4.0,Nmr=0.0](Player2)(32.0,-32.0){$u_1$}
\node[Nw=4.0,Nh=4.0,Nmr=0.0](Player2Up)(46.0,-24.0){$u_2$}
\node[Nw=4.0,Nh=4.0,Nmr=0.0](Player2Down)(46.0,-40.0){$u_3$}

\node[Nw=4.0,Nh=4.0,Nmr=2.0](Player1)(60,-32.0){$u_4$}

\put(4,-33){\Large{$G_{A_1}$}}

\node[Nw=4.0,Nh=4.0,Nmr=2.0](Entry)(18,-32.0){$u_0$}

\node[Nw=4.0,Nh=4.0,Nmr=2.0](ExUp)(74,-28.0){$u_5$}
\node[Nw=4.0,Nh=4.0,Nmr=2.0](ExDown)(74,-36.0){$u_6$}

\drawedge(Entry,Player2){$0$}
\drawedge(Player2,Player2Up){$7$}
\drawedge(Player2,Player2Down){$6$}
\drawedge(Player2Down,Player1){$0$}
\drawedge(Player2Up,Player1){$0$}
\drawedge(Player1,ExUp){$-7$}
\drawedge[ELside=r](Player1,ExDown){$-6$}

\end{picture}

\begin{picture}(70,20)(16,-40)

\node[Nw=4.0,Nh=4.0,Nmr=2.0](Entry)(18,-32.0){$v_0$}

\put(4,-33){\Large{$G_{A_0}$}}

\node[Nw=4.0,Nh=4.0,Nmr=0.0](Player2)(32.0,-32.0){$v_1$}
\node[Nw=4.0,Nh=4.0,Nmr=0.0](Player2Up)(46.0,-24.0){$v_2$}
\node[Nw=4.0,Nh=4.0,Nmr=0.0](Player2Down)(46.0,-40.0){$v_3$}

\node[Nw=4.0,Nh=4.0,Nmr=2.0](Player1)(60,-32.0){$v_4$}

\node[Nw=5.0,Nh=5.0,Nmr=0.0](EntryBox)(74,-24.0){$b_{A_1}$}
\node[Nw=4.0,Nh=4.0,Nmr=2.0](BoxEx1)(90,-22.0){$v_6$}
\node[Nw=4.0,Nh=4.0,Nmr=2.0](BoxEx2)(90,-26.0){$v_7$}

\node[Nw=4.0,Nh=4.0,Nmr=0.0](Player2Last)(74.0,-40.0){$v_5$}

\drawbpedge[ELpos=70](BoxEx2,-45,3,Player1,0,40){$0$} 

\drawbpedge[ELpos=70](BoxEx1,-20,7,Player2Last,0,25){$6$}

\drawedge[curvedepth=4](Player2Last,Player1){$-5$}




\drawedge(EntryBox,BoxEx1){$-1$}
\drawedge[ELside=r](EntryBox,BoxEx2){$1$}

\drawedge(Entry,Player2){$0$}
\drawedge(Player2,Player2Up){$1$}
\drawedge(Player2,Player2Down){$-1$}
\drawedge(Player2Down,Player1){$3$}
\drawedge(Player2Up,Player1){$7$}
\drawedge(Player1,EntryBox){$0$}
\drawedge[ELside=r](Player1,Player2Last){$5$}

\end{picture}

\caption{$G_{A_0}$ and $G_{A_1}$ are the corresponding signature games for module $A_0$ and $A_1$ from Figure~\ref{fig:StraExample}.}\label{fig:StraToSigExample}
\end{center}
\end{figure}
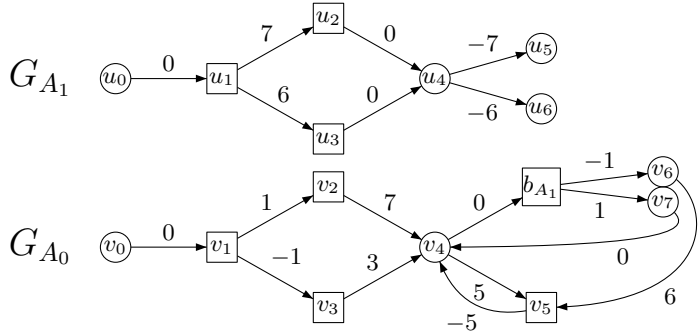
\end{examp}

}

\begin{lem}\label{lemm:NoMinusInfInSig}
Let $\wrg$ be a WRG.
If $\tau$ is a cycle independent modular winning strategy for the objective
$\LimInfAvg\geq 0$, then for every module $A_i$ 
the following assertions hold:
\begin{itemize}
\item $\tau_i$ is a winning strategy for the signature game $(G_{A_i},\vec{\nu}^i,\tau)$; and
\item the integer coefficients of $\vec{\nu}^i$ are at least $-2\cdot W \cdot |V_i|$.
\end{itemize}
\end{lem}
\begin{proof}
The first fact of the lemma follows from the facts that every path 
according to $\tau_i$ has sum of weights at least $\vec{\nu}^i$ 
and does not contain negative cycles (by the definition of the 
signature game given $\tau$, since $\tau$ is a winning strategy). 
The second fact of the lemma follows from the fact that every path 
according to $\tau_i$ does not contain negative cycles, as 
$\tau_i$ is a cycle independent modular winning strategy.
\hfill\qed
\end{proof}

\begin{lem}\label{lem:SignatureStrategy}
Let $\wrg$ be a WRG.
Let $\tau= \Set{\tau_i}_{i=1}^n$ be a cycle independent modular winning strategy 
in $\wrg$ for the objective $\LimInfAvg\geq 0$.
Let $\sigma = \Set{\sigma_i}_{i=1}^n$ be a modular strategy such that $\sigma_i$ 
is a winning strategy for the signature game $(G_{A_i},\vec{\nu}^i,\tau)$.
Then for every play $\rho^\sigma$, consistent with $\sigma$, which starts 
from the entry node of a module $A_i$ to a node $n$ in the same module 
(and possibly goes through box nodes), there exists a play $\rho^\tau$, 
consistent with $\tau$, from the same entry node to the same node $n$, 
such that $w(\rho^\sigma) \geq w(\rho^\tau)$.
In addition, the path $\rho^\sigma$ does not contain a negative proper cycle.
\end{lem}
\begin{proof}
The proof is by induction on the additional stack height of $\rho^\sigma$.
\begin{itemize}
\item \emph{Base case:} Additional stack height is $0$.
In this case the play $\rho^\sigma$ have only edges from $A_i$, and 
the weight of the play is identical to the weight of the same play in 
$(G_{A_i},\vec{\nu}^i,\tau)$.
The play $\rho^\sigma$ does not visit a vertex with threshold $+\infty$,
otherwise $\sigma_i$ would not be a winning strategy in the 
signature game $(G_{A_i},\vec{\nu}^i,\tau)$.
Hence, by definition, there exists a play $\rho^\tau$, consistent with $\tau$ from the entry node to 
$n$ with weight at most $\nu_n^i$.
Since $\sigma_i$ is a winning strategy in the signature game, and $\rho^\sigma$ is consistent with $\sigma_i$, 
we get that $w(\rho^\sigma) \geq \nu_n^i$.
Therefore $w(\rho^\sigma) \geq w(\rho^\tau)$.
Since $\sigma_i$ is a winning strategy in the signature game, and the path $\rho^\sigma$ is consistent with 
$\sigma_i$ also in graph $G_{A_i}$, we get that 
$\rho^\sigma$ does not contain negative cycles.

\item \emph{Inductive step:} Additional stack height $>0$.
For simplicity, we first assume that $\rho^\sigma$ goes only through one box 
node in the module $A_i$ (in the first stack level).
Let node $b$ be that box, and let node $u\in A_i$ be the return node in that path.
Let $\rho_{b,u}^\sigma$ be the sub-play from the entry node of $b$ to node $u$.
Recall that $w^{\tau}_{b,u}$ is the minimal weight among all plays 
consistent with $\tau$ between $b$ and $u$.
Let $A_j$ be the module invoked by $b$, and let $u'$ be the exit node that 
leads to the return node $u$ in $A_i$.
As the additional stack height from the entry node of $A_j$ to $u'$ is 
strictly smaller than the additional stack height of $\rho^\sigma$, 
it follows from the inductive hypothesis that there exists a path consistent 
with $\tau$ between these two nodes with weight at most $w(\rho_{b,u}^\sigma)$.
Hence $w^{\tau}_{b,u} \leq w(\rho_{b,u}^\sigma)$.
Thus, the weight of $\rho^\sigma$ is bounded from below by the induced 
path of $\rho^\sigma$ over the signature game $(G_{A_i},\vec{\nu}^i,\tau)$.
Thus, by the definition of the signature game there exists a path $\rho^\tau$ as desired.
In addition, by the inductive hypothesis, the path $\rho_{b,u}^\sigma$ does not contain 
proper negative cycles, and by the same arguments as above, there is also no negative proper 
cycle in module $A_i$.
The case where $\rho^\sigma$ goes through more then one box, is a straightforward extension 
of the argument presented above.
\end{itemize}
Thus we have the desired result.
\hfill\qed
\end{proof}

\begin{lem}\label{lem:MemorylessWinningStrategy} 
Let $\wrg$ be a WRG.
Let $\tau= \Set{\tau_i}_{i=1}^n$ be a cycle independent modular winning strategy 
in $\wrg$ for the objective $\LimInfAvg\geq0$.
Let $\sigma = \Set{\sigma_i}_{i=1}^n$ be a memoryless modular strategy such that $\sigma_i$ 
is a memoryless winning strategy for the signature game $(G_{A_i},\vec{\nu}^i,\tau)$.
Then $\sigma$ is a memoryless modular winning strategy in $\wrg$ 
for the objective $\LimInfAvg\geq 0$.
\end{lem}
\begin{proof}
Let $\wrg^\sigma$ be the player-2 WRG obtained by fixing the memoryless modular strategy 
$\sigma$ in $\wrg$.
Assume towards a contradiction that $\wrg^\sigma$ has a reachable non-decreasing 
negative cycle $C$, and let $\rho$ be a finite path that leads to the first vertex of $C$.
By Lemma~\ref{lem:SignatureStrategy} it follows that $C$ cannot be a proper cycle.

First, we argue that in $\wrg$ there exists a finite path $\rho^\tau$ 
from the first vertex of $\rho$ to the last vertex of $\rho$ 
that is consistent with $\tau$.
Indeed, by Lemma~\ref{lem:SignatureStrategy}, between every entry node 
and box node in $\rho$ there exists a path consistent with $\tau$,
and finally there also exists such a path between the last entry node and 
the last node of $\rho$.

Second, to achieve the contradiction we will show that $\tau$ is not a 
winning strategy.
Let $e_1$ be the first entry node in $C$ (it must exist as $C$ is not a proper cycle), 
and $n_i$ be the last (and the first) node in $C$ (note that $C$ is not a proper cycle, 
and hence this node is well defined).
Note that for every $m\in\Nat$, the path $C^m$ is a non-decreasing cycle that is 
consistent with $\sigma$ (as $\sigma$ is a memoryless strategy).
Let $\rho_{e_1,n_i}$ be the path from $e_1$ to $n_i$.
Let $\rho_{n_i,e_1}$ be the path from $n_i$ to first appearance of $e_1$ in $C$.
We consider the path $\rho^* = \rho_{e_1,n_i} \cdot C^m \cdot \rho_{n_i,e_1}$,
for $m = 2\cdot W\cdot (|\rho_{e_1,n_i}| + |\rho_{n_i,e_1}|)$.
This is a path that is (i)~consistent with $\sigma$, (ii)~begins and ends in the 
entry node $e_1$ of the same module (not necessarily in the same stack height).
Let $b_1, b_2, \dots, b_\ell$ be the boxes that occur in the path.
By Lemma~\ref{lem:SignatureStrategy}, for every $k$ there exists a path 
$\rho^\tau_{b_i,b_{i+1}}$ consistent with $\tau$ such that 
$w(\rho^\tau_{b_i,b_{i+1}}) \leq w(\rho^\sigma_{b_i,b_{i+1}})$.
Hence there exists a path $\rho^{\tau}_*$ that is consistent with $\tau$ 
from $e_1$ to $e_1$ such that the sum of the weights is negative.
As $\tau$ is a modular strategy, and $e_1$ is an entry node, it follows that the 
path $(\rho^{\tau}_*)^\omega$ is also consistent with $\tau$,
and has a negative mean-payoff.
In conclusion, we obtain that there exists a reachable negative non-decreasing cycle
in $\wrg$ consistent with $\tau$, and this contradicts that $\tau$ is a winning strategy.

Hence, every path consistent with $\sigma$ does not contain a negative non-decreasing
cycle.
By Lemma~\ref{lemm:NotEmptyIfCycle} it follows that $\sigma$ is a winning strategy 
in $\wrg$ for the objective $\LimInfAvg\geq 0$.
\hfill\qed
\end{proof}



\begin{lem}\label{lem:ModularMemorylessDeterminacy}
Let $\wrg$ be a WRG.
Player~1 has a modular winning strategy for the objective 
$\LimInfAvg\geq 0$ iff there exists a memoryless modular winning strategy 
for player~1 for the $\LimInfAvg\geq 0$ objective.
\end{lem}
\begin{proof}
The proof for the direction from right to left is trivial.
The opposite direction is obtained as follows: 
by Lemma~\ref{lem:ModularStrategyInfImpCycleIndStrategiesSup} it follows that 
if there is a modular winning strategy, then there is a cycle independent 
modular winning strategy;
and by Lemma~\ref{lem:MemorylessWinningStrategy} it follows that if there 
is a cycle independent modular winning strategy, then there is a memoryless 
modular winning strategy.
The desired result follows.
\hfill\qed
\end{proof}

We are now ready to prove the main result of this section.

\begin{thm}\label{thm:ModularStratInNP}
The problem of deciding if player~1 has a modular winning strategy in a 
WRG $\wrg$ for the objective $\LimInfAvg\geq 0$ is in NP.
\end{thm}
\begin{proof}
By Lemma~\ref{lem:ModularMemorylessDeterminacy} it is enough to guess a memoryless 
modular strategy (the memoryless modular strategy is the polynomial witness) 
and verify that it is indeed a winning strategy.
The verification can be achieved in polynomial time using the polynomial-time 
algorithms of Section~\ref{sect:PushDownProcesses} for WPSs with mean-payoff objectives.
\hfill\qed
\end{proof}

\noindent{\bf Strict inequalities and stack boundedness.}
Note that by Corollary~\ref{cor:ModularSupIffInf}, the results of Lemma~\ref{lem:ModularMemorylessDeterminacy} 
and Theorem~\ref{thm:ModularStratInNP} also hold for the $\LimSupAvg \geq 0$ objective. 
For modular strategies we only presented the result for 
mean-payoff objectives with non-strict inequalities.
The results for strict inequalities follow from an adaptation of the 
proofs for non-strict inequalities (for which we prove memoryless modular 
strategies are sufficient).
In particular, Corollary~\ref{cor:ModularSupIffInf} holds also for strict 
inequalities, that is, given a WRG $\wrg$, there exists a modular winning 
strategy for the objective $\LimSupAvg > 0$ iff there exists a modular winning 
strategy for the objective $\LimInfAvg > 0$.
Moreover, the results also follow for mean-payoff 
objectives with the stack boundedness condition
for the following reason: we observe that the manipulated 
operations never increase the stack height.
Thus from our results it follow that if there is a modular winning
strategy to ensure the mean-payoff objective along with stack
boundedness, then there is also a memoryless modular winning strategy.
Hence the NP upper bound follows for strict inequalities 
as well as for stack boundedness.

\smallskip\noindent{\bf WRGs with multi-entry modules.}
We now discuss the implication of our results for WRGs with modules
that have multiple entries.
The polynomial-time reduction from multi-entry recursive game graphs $\wrg$ 
to single-entry recursive game graphs $\wrg'$ presented in~\cite{AlurReach} 
preserves the winner of the game for reachability objectives,
and we observe that the reduction also preserves the winner for 
mean-payoff objectives.
Hence, Theorem~\ref{thm:ModularStratInNP} holds also when modules have 
multiple entry nodes.
However,  the reduction does not preserve memorylessness of modular strategies,
i.e., in general a modular memoryless winning strategy for player~1 in 
$\wrg'$ corresponds to a modular strategy that needs memory over $\wrg$.
In other words, for WRGs with multi-entry modules, 
though the reduction and Theorem~\ref{thm:ModularStratInNP} can be used to obtain 
the NP complexity, the reduction and Lemma~\ref{lem:ModularMemorylessDeterminacy} do not imply the 
existence of memoryless modular winning strategies.
In the following example we show that there exist WRGs with multi-entry modules,
where modular winning strategies exist, but there is no memoryless modular
winning strategy (and the example works even for reachability objectives,
and when all nodes are controlled by player~1).

\begin{figure}[H]
\begin{center}

\begin{picture}(70,35)(16,-50)



\node[Nw=4.0,Nh=4.0,Nmr=4.0](Player2)(46.0,-32.0){$u_1$}

\put(20,-26.5){\Large{$A_1$}}

\node[Nw=2.0,Nh=2.0,Nmr=1.0](EntryUp)(18,-28.0){}
\node[Nw=2.0,Nh=2.0,Nmr=1.0](EntryDown)(18,-36.0){}

\node[Nw=2.0,Nh=2.0,Nmr=1.0](ExUp)(74,-28.0){}
\node[Nw=2.0,Nh=2.0,Nmr=1.0](ExDown)(74,-36.0){}

\node[Nw=56.0,Nh=25,Nmr=7.475](Frame)(46,-32.0){}

\drawedge(EntryUp,Player2){}
\drawedge(EntryDown,Player2){}

\drawedge(Player2,ExUp){}
\drawedge(Player2,ExDown){}

\end{picture}

\begin{picture}(70,35)(16,-50)

\put(20,-24){\Large{$A_0$}}
\node[Nw=2.0,Nh=2.0,Nmr=1.0](Entry)(18,-32.0){}



\node[Nw=16.0,Nh=8.0,Nmr=1.0](Box)(44,-32.0){$A_1$}
\node[Nw=2.0,Nh=2.0,Nmr=1.0](EntryBox1)(36,-30.0){}
\node[Nw=2.0,Nh=2.0,Nmr=1.0](EntryBox2)(36,-34.0){}
\node[Nw=2.0,Nh=2.0,Nmr=1.0](BoxEx1)(52,-30.0){}
\node[Nw=2.0,Nh=2.0,Nmr=1.0](BoxEx2)(52,-34.0){}

\node[Nw=16.0,Nh=8.0,Nmr=1.0](sBox)(70,-32.0){$A_1$}
\node[Nw=2.0,Nh=2.0,Nmr=1.0](sEntryBox1)(62,-30.0){}
\node[Nw=2.0,Nh=2.0,Nmr=1.0](sEntryBox2)(62,-34.0){}
\node[Nw=2.0,Nh=2.0,Nmr=1.0](sBoxEx1)(78,-30.0){}
\node[Nw=2.0,Nh=2.0,Nmr=1.0](sBoxEx2)(78,-34.0){}






\node[Nw=2.0,Nh=2.0,Nmr=1.0](Exit)(96,-32.0){}

\node[Nw=78.0,Nh=30,Nmr=7.475](Frame)(57,-32.0){}

\node[Nw=4.0,Nh=4.0,Nmr=4.0](BadSink)(87.0,-29.0){$v_1$}
\node[Nw=4.0,Nh=4.0,Nmr=4.0](GoodSink)(87.0,-35.0){$v_2$}

\drawloop[loopdiam=2](BadSink){$-1$}
\drawloop[loopdiam=2,loopangle=-90](GoodSink){$1$}

\drawedge(Entry,EntryBox1){}
\drawedge(BoxEx2,sEntryBox2){}
\drawedge(sBoxEx2,BadSink){}
\drawedge(sBoxEx1,GoodSink){}

\drawedge[curvedepth=5](BoxEx1,BadSink){}


\end{picture}

\caption{RSM with two modules ($A_0$ and $A_1$).
Player~1 controls all the vertices.}\label{fig:MultiEntriesExample}
\end{center}
\end{figure}
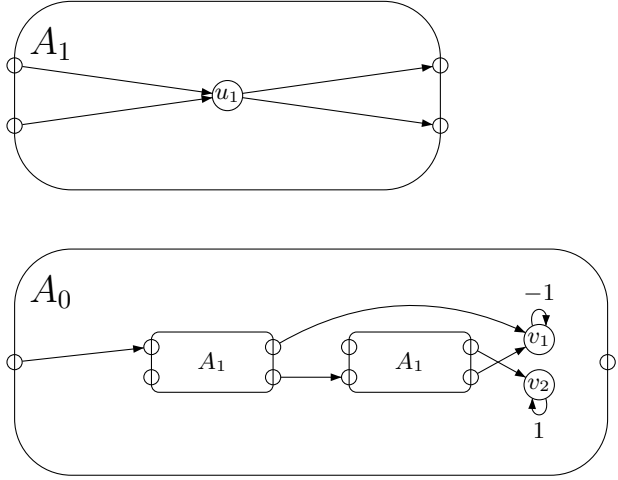

\begin{examp}
Consider the WRG shown in Figure~\ref{fig:MultiEntriesExample}, where
player~1 has to make decisions only in node $u_1$ in module $A_1$ 
(and all the other nodes have out-degree $1$).
A player-1 modular winning strategy in node $u_1$ is to go to the upper exit node 
if the module is invoked from the lower entry, and otherwise player~1 chooses 
the lower exit node.
With this strategy the play reaches $v_2$ and the mean-payoff is $1$.
On the other hand, any play according to a player-1 memoryless modular strategy, namely, 
either always go to the upper node or always go to the lower exit node, reaches $v_1$ and has mean-payoff $-1$.
Hence, in this example, player~1 has a modular winning strategy but not memoryless modular 
winning strategy (even for the objective is reachability to $v_2$).
\end{examp}

\subsection{NP-hardness of the modular winning strategy problem}
In this section we establish the NP-hardness of the modular 
winning strategy problem.
Our hardness result will be for one-player WRGs (player-1 WRGs), 
where every module will have single exits, and the weights are 
$\Set{-1,0,+1}$.
In other words, our hardness result shows that even a very simple 
version of the problem (single exit one-player WRGs with constant 
weights) is NP-hard.

\newcommand{\cl}{\mathit{cl}}

\smallskip\noindent{\bf Reduction.} 
We present a reduction from the 3-SAT problem (satisfiability of a CNF 
formula where every clause has exactly three distinct literals).
Consider a 3-SAT formula $\varphi(x_1,x_2,\dots,x_n) = \bigwedge_{i=1}^m \cl_i$,
over $n$ variables $x_1,x_2,\dots,x_n$, and $m$ clauses $\cl_1, \cl_2, \dots,\cl_m$.
A literal is a variable $x_i$ or its negation $\neg x_i$.
We construct a player-1 WRG as follows:
$\wrg_\varphi = \atuple{A_0,x_1, \neg x_1, x_2, \neg x_2, \dots, x_n, \neg x_n, \cl_1, 
\cl_2, \dots, \cl_m}$ 
in the following way: there is an initial module $A_0$, there is a module for every literal, 
and for every clause.
We now describe the modules.

\smallskip\noindent{\em Module $A_0$.}
The module invokes in an infinite loop in sequence the modules $\cl_1, \cl_2, \dots \cl_m$, and 
all the transitions in this module have weight zero.

\smallskip\noindent{\em Module for clause $\cl_i$.}
There is an edge from the entry node of module $\cl_i$ 
to a box that invokes module $y$, for every literal $y$ that appears in 
the clause $\cl_i$.
There is also an edge from the return node of $y$ to the exit node of $\cl_i$.
All the weights in this module are zero.

\smallskip\noindent{\em Module for literal $y_i$.}
The entry node of $y_i$ has outdegree two (left edge and right edge).
The left edge is the $\False$ edge, which leads to the exit node, and has a weight $-1$.
The right edge is the $\True$ edge, which leads to a box that invokes a 
call for module $\neg y_i$, and its weight is $-1$.
The return of the box leads to the exit node and the edge weight is $+2$.
The reduction is illustrated pictorially in Figure~\ref{fig-lower}.
\begin{figure}[t]
   \begin{center}
      \setlength{\unitlength}{0.00032371in}
\begingroup\makeatletter\ifx\SetFigFontNFSS\undefined%
\gdef\SetFigFontNFSS#1#2#3#4#5{%
  \reset@font\fontsize{#1}{#2pt}%
  \fontfamily{#3}\fontseries{#4}\fontshape{#5}%
  \selectfont}%
\fi\endgroup%
{\renewcommand{\dashlinestretch}{30}
\begin{picture}(15552,7975)(0,-10)
\put(7534.000,-13295.500){\arc{41217.457}{4.3671}{5.0577}}
\blacken\path(661.983,6165.517)(559.000,6097.000)(682.110,6108.994)(661.983,6165.517)
\put(559,5962){\ellipse{202}{202}}
\put(1448,5962){\ellipse{202}{202}}
\put(2753,5962){\ellipse{202}{202}}
\put(3923,5962){\ellipse{202}{202}}
\put(5228,5962){\ellipse{202}{202}}
\put(12023,5962){\ellipse{202}{202}}
\put(13339,5962){\ellipse{202}{202}}
\put(14509,5962){\ellipse{202}{202}}
\put(109,1462){\ellipse{202}{202}}
\put(2134,2272){\ellipse{202}{202}}
\put(3439,2272){\ellipse{202}{202}}
\put(5689,1462){\ellipse{202}{202}}
\put(2179,1417){\ellipse{202}{202}}
\put(3439,1417){\ellipse{202}{202}}
\put(2134,472){\ellipse{202}{202}}
\put(3439,472){\ellipse{202}{202}}
\put(9604,1417){\ellipse{202}{202}}
\put(11854,787){\ellipse{202}{202}}
\put(13159,787){\ellipse{202}{202}}
\put(15184,1822){\ellipse{202}{202}}
\put(15443,5962){\ellipse{202}{202}}
\thicklines
\path(109,2767)(5689,2767)(5689,22)
	(109,22)(109,2767)
\path(9604,2767)(15184,2767)(15184,22)
	(9604,22)(9604,2767)
\thinlines
\path(12034,6322)(13339,6322)(13339,5602)
	(12034,5602)(12034,6322)
\path(1459,6322)(2764,6322)(2764,5602)
	(1459,5602)(1459,6322)
\path(3934,6322)(5239,6322)(5239,5602)
	(3934,5602)(3934,6322)
\path(2134,2677)(3439,2677)(3439,1957)
	(2134,1957)(2134,2677)
\path(2134,1777)(3439,1777)(3439,1057)
	(2134,1057)(2134,1777)
\path(2134,832)(3439,832)(3439,112)
	(2134,112)(2134,832)
\path(11854,1147)(13159,1147)(13159,427)
	(11854,427)(11854,1147)
\thicklines
\path(559,7447)(15454,7447)(15454,4477)
	(559,4477)(559,7447)
\thinlines
\path(649,6007)(1369,6007)
\blacken\path(1249.000,5977.000)(1369.000,6007.000)(1249.000,6037.000)(1249.000,5977.000)
\path(2854,6007)(3844,6007)
\blacken\path(3724.000,5977.000)(3844.000,6007.000)(3724.000,6037.000)(3724.000,5977.000)
\path(5329,5962)(5869,5962)
\blacken\path(5749.000,5932.000)(5869.000,5962.000)(5749.000,5992.000)(5749.000,5932.000)
\path(13474,5962)(14419,5962)
\blacken\path(14299.000,5932.000)(14419.000,5962.000)(14299.000,5992.000)(14299.000,5932.000)
\path(199,1507)(2044,2272)
\blacken\path(1944.641,2198.326)(2044.000,2272.000)(1921.660,2253.750)(1944.641,2198.326)
\path(199,1462)(2044,1462)
\blacken\path(1924.000,1432.000)(2044.000,1462.000)(1924.000,1492.000)(1924.000,1432.000)
\path(199,1372)(2044,517)
\blacken\path(1922.509,540.236)(2044.000,517.000)(1947.737,594.675)(1922.509,540.236)
\path(3529,472)(5509,1417)
\blacken\path(5413.624,1338.238)(5509.000,1417.000)(5387.780,1392.387)(5413.624,1338.238)
\path(3529,1417)(5509,1462)
\blacken\path(5389.713,1429.281)(5509.000,1462.000)(5388.349,1489.266)(5389.713,1429.281)
\path(3574,2272)(5554,1552)
\blacken\path(5430.972,1564.815)(5554.000,1552.000)(5451.477,1621.203)(5430.972,1564.815)
\path(9694,1417)(15049,1867)
\blacken\path(14931.934,1827.057)(15049.000,1867.000)(14926.909,1886.846)(14931.934,1827.057)
\path(9739,1372)(11719,832)
\blacken\path(11595.335,834.631)(11719.000,832.000)(11611.122,892.517)(11595.335,834.631)
\path(13249,832)(15049,1732)
\blacken\path(14955.085,1651.502)(15049.000,1732.000)(14928.252,1705.167)(14955.085,1651.502)
\dashline{60.000}(6544,5917)(11314,5917)
\put(2539,3037){\makebox(0,0)[lb]{\smash{{\SetFigFontNFSS{7}{8.4}{\rmdefault}{\mddefault}{\updefault}$\cl_i$}}}}
\put(14329,1102){\makebox(0,0)[lb]{\smash{{\SetFigFontNFSS{7}{8.4}{\rmdefault}{\mddefault}{\updefault}+2}}}}
\put(10414,832){\makebox(0,0)[lb]{\smash{{\SetFigFontNFSS{7}{8.4}{\rmdefault}{\mddefault}{\updefault}-1}}}}
\put(12304,1732){\makebox(0,0)[lb]{\smash{{\SetFigFontNFSS{7}{8.4}{\rmdefault}{\mddefault}{\updefault}-1}}}}
\put(1819,5872){\makebox(0,0)[lb]{\smash{{\SetFigFontNFSS{7}{8.4}{\rmdefault}{\mddefault}{\updefault}$\cl_1$}}}}
\put(4294,5872){\makebox(0,0)[lb]{\smash{{\SetFigFontNFSS{7}{8.4}{\rmdefault}{\mddefault}{\updefault}$\cl_2$}}}}
\put(12394,5872){\makebox(0,0)[lb]{\smash{{\SetFigFontNFSS{7}{8.4}{\rmdefault}{\mddefault}{\updefault}$\cl_m$}}}}
\put(2539,2182){\makebox(0,0)[lb]{\smash{{\SetFigFontNFSS{7}{8.4}{\rmdefault}{\mddefault}{\updefault}$y_j$}}}}
\put(2494,1282){\makebox(0,0)[lb]{\smash{{\SetFigFontNFSS{7}{8.4}{\rmdefault}{\mddefault}{\updefault}$y_\ell$}}}}
\put(2539,337){\makebox(0,0)[lb]{\smash{{\SetFigFontNFSS{7}{8.4}{\rmdefault}{\mddefault}{\updefault}$y_k$}}}}
\put(12169,697){\makebox(0,0)[lb]{\smash{{\SetFigFontNFSS{7}{8.4}{\rmdefault}{\mddefault}{\updefault}$\neg y_i$}}}}
\put(11989,3037){\makebox(0,0)[lb]{\smash{{\SetFigFontNFSS{7}{8.4}{\rmdefault}{\mddefault}{\updefault}$y_i$}}}}
\put(7174,7717){\makebox(0,0)[lb]{\smash{{\SetFigFontNFSS{7}{8.4}{\rmdefault}{\mddefault}{\updefault}$A_0$}}}}
\end{picture}
}
   \end{center}
  \caption{NP-hardness reduction}
  \label{fig-lower}
\end{figure}
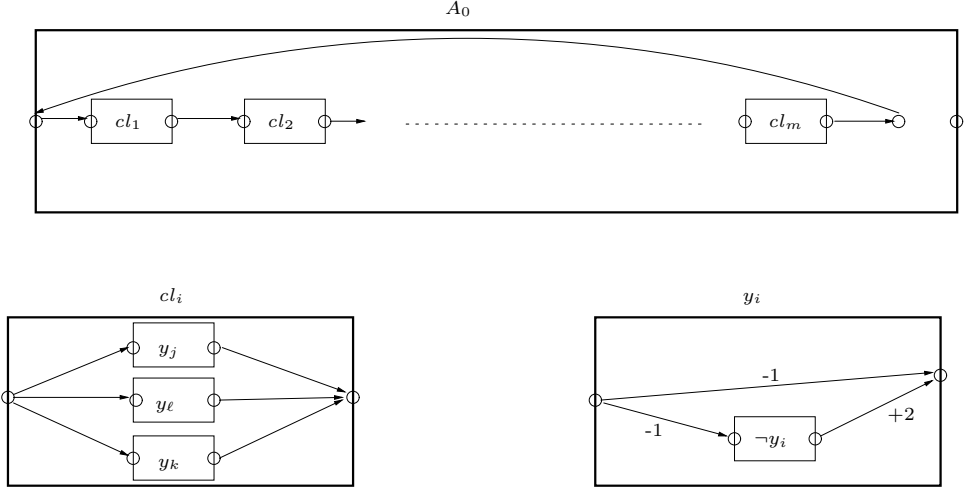

\begin{Obs}\label{obs:ObsOnConstruction}
Every path (that is generated by a modular strategy) from the entry node of 
the module $y_i$ to its exit node, has a total weight of at most $0$, 
hence every path (that is generated by a modular strategy) from the entry 
node of module $\cl_i$ to its exit node, has a total weight of at most $0$.
\end{Obs}

\begin{lem}\label{lemm:WinsIffSat}
There exists a modular winning strategy for player~1 in $\wrg_\varphi$ 
for the objective $\LimInfAvg\geq 0$ iff $\varphi$ is satisfiable.
\end{lem}
\begin{proof}
We first observe that every modular strategy in $\wrg_\varphi$ is a memoryless
modular strategy.
Observe that modular strategy for player~1 is a selection of a 
literal for every clause module, and selecting either the $\True$ or $\False$ edge 
for every literal module.
We now present both directions of the proof.

\begin{itemize}
\item {\em Modular winning strategy implies satisfiability.}
Let $\tau$ be a modular winning strategy for player 1, 
such that the literal $y^\tau_i$ is chosen for clause $\cl_i$.
First, towards a contradiction, let us assume that there exists 
$i \in \RangeSet{1}{m}$ such that $\tau$ selects the $\False$ edge for 
module $y^\tau_i$.
Hence the weight of the path inside the module $\cl_i$ is negative.
Thus, due to Observation~\ref{obs:ObsOnConstruction}, the mean-payoff value according 
to $\tau$ is negative.
Therefore $\tau$ selects the $\True$ edge for the module $y^\tau_i$.
Next, towards a contradiction, let us assume that there exist $i,j \in \RangeSet{1}{m}$ such 
that $y^\tau_i = \neg y^\tau_j$ and $\tau$ selects the $\True$ edge for both modules.
In this case, the play will never exit module $y^\tau_i$, and will go forever through 
edges with negative weights.
Therefore if $\tau$ selects the $\True$ edge for $y^\tau_i$, then it does not select 
the $\True$ edge for $\neg y^\tau_i$.
Due to the above, the assignment that assigns a true value to the literal 
$y^\tau_i$ in clause $\cl_i$ is a valid (non-conflicting) assignment 
that satisfies $\varphi$.

\item {\em Satisfiability implies modular winning strategy.}
Let $\overline{x}$ be a satisfying assignment (a non-conflicting 
assignment of truth values to variables) for the formula $\varphi$.
We construct a modular winning strategy $\tau_{\overline{x}}$ as follows.
In module $\cl_i$, the modular strategy invokes the 
module $y^{\overline{x}}_i$, where $y^{\overline{x}}_i$ is a 
literal for which $\overline{x}$ assigns a true value (since $\overline{x}$ 
is a satisfying assignment such a literal must exist).
In module $y_i$, follow the $\True$ edge if $\overline{x}$ assigns a 
true value to the literal $y_i$, and follow the $\False$ edge otherwise.
It is straightforward to verify that the mean-payoff of a 
play according to $\tau_{\overline{x}}$ is zero. 
\end{itemize}
This completes the proof.
\hfill\qed
\end{proof}
Observe that in the hardness reduction we have used positive weight $+2$
for simplicity, which can be split into two edges of weight $+1$ each.
Hence we have the following theorem.

\begin{thm}\label{thm:np-hard}
The modular winning strategy problem is NP-hard for one-player WRGs
(player-1 WRGs) with single exit for every module and objective 
$\LimInfAvg\geq 0$ with edge weights in $\Set{-1,0,+1}$.
\end{thm}

\noindent{\bf Strict inequalities and stack boundedness.}
We first observe that the above reduction also holds for $\LimSupAvg\geq0$ 
objective.
Moreover, whenever the 3-SAT formula $\varphi$ is satisfiable, then the 
witness memoryless modular strategy along with mean-payoff objective also
ensures stack boundedness. 
Hence the hardness result follows for mean-payoff objectives with non-strict
inequalities as well as for stack boundedness.
The result for strict inequality is obtained as follows: we modify the above
reduction by changing the weight of the edge back to the entry node of $A_0$
from~0 to ~1.
Then if the formula $\varphi$ is satisfiable, then the average payoff for 
memoryless modular strategies is at least $\frac{1}{|V|}$, where $|V|$ is the 
number of vertices, and if the formula $\varphi$ is not satisfiable, then the 
mean-payoff under all memoryless modular strategies is at most~0.
Hence the hardness follows also for mean-payoff objectives with strict 
inequalities.
We have the following theorem summarizing the results for modular strategies.

\begin{thm}
The following assertions hold for WRGs with objectives 
$\Phi \bowtie 0$, for $\bowtie \in \Set{\geq,>}$, 
$\Phi\in \Set{\LimSupAvg,\LimInfAvg}$, as well as objectives
$\Phi\bowtie 0$ along with stack boundedness.
\begin{enumerate}
\item If there is a modular winning strategy, then there is a memoryless
modular winning strategy.

\item The decision problem of whether there is a memoryless modular 
winning strategy is NP-complete. 

\item The decision problem is NP-hard for player-1 WRGs with single exit
for every module and edge weights in $\Set{-1,0,+1}$.
\end{enumerate}
\end{thm}

\section{Conclusion}
In this work we study for the first time mean-payoff objectives in 
pushdown games and present a complete characterization of computational 
and strategy complexity.
We show that pushdown systems (one-player pushdown games) 
with mean-payoff objectives under global strategies 
can be solved in polynomial time, whereas pushdown games 
with mean-payoff objectives under global strategies are 
undecidable.
For modular strategies both pushdown systems and pushdown games with 
mean-payoff objectives are NP-complete.
We also show 
that global strategies for mean-payoff objectives in general 
require infinite memory even in pushdown systems; 
whereas memoryless strategies
suffice for modular strategies for mean-payoff objectives.
An interesting direction of future work would be to consider such games 
with multi-dimensional mean-payoff objectives.



\begin{thebibliography}{10}

\bibitem{ABK11}
S.~Almagor, U.~Boker, and O.~Kupferman.
\newblock What's decidable about weighted automata?
\newblock In {\em ATVA}, pages 482--491, 2011.

\bibitem{ornaACM}
S.~Almagor, U.~Boker, and O.~Kupferman.
\newblock Formally Reasoning About Quality.
\newblock In Journal of the ACM, to appear.

\bibitem{ABEGRY05}
R.~Alur, M.~Benedikt, K.~Etessami, P.~Godefroid, T.~W. Reps, and M.~Yannakakis.
\newblock Analysis of recursive state machines.
\newblock {\em ACM Trans. Program. Lang. Syst.}, 27(4):786--818, 2005.

\bibitem{AHK02}
R.~Alur, T.A. Henzinger, and O.~Kupferman.
\newblock Alternating-time temporal logic.
\newblock {\em Journal of the ACM}, 49:672--713, 2002.

\bibitem{alur12}
R. Alur, K. Etessami, and M. Yannakakis.
\newblock Analysis of recursive state machines.
\newblock In {\em CAV}, pages 207--220, 2001.

\bibitem{AlurParity}
R.~Alur, S.~La Torre, and P.~Madhusudan.
\newblock Modular strategies for infinite games on recursive graphs.
\newblock In {\em CAV}, pages 67--79, 2003.

\bibitem{AlurReach}
R.~Alur, S.~La Torre, and P.~Madhusudan.
\newblock Modular strategies for recursive game graphs.
\newblock {\em Theor. Comput. Sci.}, 354(2):230--249, 2006.

\bibitem{BCHJ09}
R.~Bloem, K.~Chatterjee, T.~A. Henzinger, and B.~Jobstmann.
\newblock Better quality in synthesis through quantitative objectives.
\newblock In {\em CAV}, pages 140--156, 2009.

\bibitem{BGHJ09}
R.~Bloem, K.~Greimel, T.~A. Henzinger, and B.~Jobstmann.
\newblock Synthesizing robust systems.
\newblock In {\em FMCAD}, pages 85--92, 2009.

\bibitem{BCHK11}
U.~Boker, K.~Chatterjee, T.~A. Henzinger, and O.~Kupferman.
\newblock Temporal specifications with accumulative values.
\newblock In {\em LICS}, pages 43--52, 2011.


\bibitem{BBFK08}
T.~Br{\'a}zdil, V.~Brozek, V.~Forejt, and A.~Kucera.
\newblock Reachability in recursive {M}arkov decision processes.
\newblock {\em Inf. Comput.}, 206(5):520--537, 2008.

\bibitem{BBKO11}
T.~Br{\'a}zdil, V.~Brozek, A.~Kucera, and J.~Obdrz{\'a}lek.
\newblock Qualitative reachability in stochastic {BPA} games.
\newblock {\em Inf. Comput.}, 209(8):1160--1183, 2011.

\bibitem{Buchi64}
J.R. B{\"u}chi.
\newblock Regular canonical systems.
\newblock Archive for Mathematical Logic 6(3), pp. 91–111, 1964.

\bibitem{BuchiLandweber69}
J.R. B{\"u}chi and L.H. Landweber.
\newblock Solving sequential conditions by finite-state strategies.
\newblock {\em Transactions of the AMS}, 138:295--311, 1969.

\bibitem{CDH10}
K.~Chatterjee, L.~Doyen, and T.~A. Henzinger.
\newblock Quantitative languages.
\newblock {\em ACM Trans. Comput. Log.}, 11(4), 2010.


\bibitem{CPV15}
K.~Chatterjee, A.~Pavlogiannis, and Y.~Velner.
\newblock Quantitative interprocedural analysis.
\newblock In \emph{POPL}, 2015.

\bibitem{C_acm_15}
K. Chatterjee, T. A. Henzinger, B. Jobstmann and R. Singh.
\newblock Measuring and Synthesizing Systems in Probabilistic Environments.
\newblock Journal of the ACM, 2015.

\bibitem {chen2002}
H. Chen and D. Wagner.
\newblock Mops: an infrastructure for examining security properties of software. \newblock In Proceedings of ACM Conference on Computer and Communications Security, pages 235–244, 2002.

\bibitem{CLRS-Book}
T.~H. Cormen, C.~E. Leiserson, R.~L. Rivest, and C.~Stein.
\newblock {\em Introduction to Algorithms}.
\newblock The MIT Press, 2001.

\bibitem{InterfaceTheories}
L.~de~Alfaro and T.A. Henzinger.
\newblock Interface theories for component-based design.
\newblock In {\em EMSOFT}, LNCS 2211, pages 148--165. Springer, 2001.

\bibitem{DM10}
M.~Droste and I.~Meinecke.
\newblock Describing average- and longtime-behavior by weighted {MSO} logics.
\newblock In {\em MFCS}, pages 537--548, 2010.

\bibitem{EM79}
A.~Ehrenfeucht and J.~Mycielski.
\newblock Positional strategies for mean payoff games.
\newblock {\em Int. Journal of Game Theory}, 8(2):109--113, 1979.

\bibitem{EJ88}
E.A. Emerson and C.~Jutla.
\newblock The complexity of tree automata and logics of programs.
\newblock In {\em FOCS'88}, pages 328--337. IEEE, 1988.

\bibitem{EJ91}
E.A. Emerson and C.~Jutla.
\newblock Tree automata, mu-calculus and determinacy.
\newblock In {\em FOCS}, pages 368--377. IEEE, 1991.

\bibitem{EY05}
K.~Etessami and M.~Yannakakis.
\newblock Recursive {M}arkov decision processes and recursive stochastic games.
\newblock In {\em ICALP'05}, LNCS 3580, Springer, pages 891--903, 2005.

\bibitem{EY09}
K.~Etessami and M.~Yannakakis.
\newblock Recursive {M}arkov chains, stochastic grammars, and monotone systems
  of nonlinear equations.
\newblock {\em J. ACM}, 56(1), 2009.

\bibitem{EY15}
K.~Etessami and M.~Yannakakis.
\newblock Recursive Markov Decision Processes and Recursive Stochastic Games. J. \newblock {\em J. ACM}, 2015.

\bibitem{WilkeBook}
E.~Gr{\"a}del, W.~Thomas, and T.~Wilke, editors.
\newblock {\em Automata, Logics, and Infinite Games: A Guide to Current
  Research}, LNCS 2500. Springer, 2002.

\bibitem{GH82}
Y.~Gurevich and L.~Harrington.
\newblock Trees, automata, and games.
\newblock In {\em STOC'82}, pages 60--65. ACM Press, 1982.

\bibitem{henz2002}
T.A. Henzinger, R. Jhala, R. Majumdar, G.C. Necula, G. Sutre and W. Weimer.
\newblock Temporal-safety proofs for systems code.
\newblock In {\em CAV'02}, pages 526--538, 2002.

\bibitem{FairSimulation}
T.~A. Henzinger, O.~Kupferman, and S.~Rajamani.
\newblock Fair simulation.
\newblock {\em Information and Computation}, 173:64--81, 2002.

\bibitem{Jur98}
M.~Jurdzinski.
\newblock Deciding the winner in parity games is in {UP} $\cap$ co-{UP}.
\newblock {\em Information Processing Letters}, 68(3):119--124, 1998.

\bibitem{Karp78}
R.M. Karp.
\newblock A characterization of the minimum cycle mean in a digraph.
\newblock {\em Discrete Mathematics}, 23:309--311, 1978.

\bibitem{Krob}
D~Krob.
\newblock The equality problem for rational series with multiplicities in the
  tropical semiring is undecidable.
\newblock In {\em ICALP}, pages 101--112, 1992.

\bibitem{LigLip69}
T.~A. Liggett and S.~A. Lippman.
\newblock Stochastic games with perfect information and time average payoff.
\newblock {\em Siam Review}, 11:604--607, 1969.

\bibitem{McNaughton93}
R.~McNaughton.
\newblock Infinite games played on finite graphs.
\newblock {\em Annals of Pure and Applied Logic}, 65:149--184, 1993.

\bibitem{PnueliRosner89}
A.~Pnueli and R.~Rosner.
\newblock On the synthesis of a reactive module.
\newblock In {\em POPL}, pages 179--190. ACM Press, 1989.

\bibitem{RamadgeWonham87}
P.~J. Ramadge and W.~M. Wonham.
\newblock Supervisory control of a class of discrete-event processes.
\newblock {\em SIAM Journal of Control and Optimization}, 25(1):206--230, 1987.

\bibitem{mooly}
T. Reps, S. Horwitz and M. Sagiv.
\newblock Precise Interprocedural Dataflow Analysis via Graph Reachability
\newblock In {\em POPL}, pages 49--61. ACM Press, 1995.

\bibitem{Thomas97}
W.~Thomas.
\newblock Languages, automata, and logic.
\newblock In {\em Handbook of Formal Languages}, volume 3, Beyond Words,
  chapter~7, pages 389--455. Springer, 1997.

\bibitem{Wal00}
I.~Walukiewicz.
\newblock Model checking {CTL} properties of pushdown systems.
\newblock In {\em FSTTCS}, pages 127--138, 2000.

\bibitem{Wal01}
I.~Walukiewicz.
\newblock Pushdown processes: Games and model-checking.
\newblock {\em Inf. Comput.}, 164(2):234--263, 2001.

\bibitem{Yan90}
M.~Yannakakis.
\newblock Graph-theoretic methods in database theory.
\newblock In {\em PODS}, pages 230--242, 1990.

\bibitem{Zie98}
W.~Zielonka.
\newblock Infinite games on finitely coloured graphs with applications to
  automata on infinite trees.
\newblock In {\em Theoretical Computer Science}, volume 200(1-2), pages
  135--183, 1998.

\bibitem{ZP95}
U.~Zwick and M.~Paterson.
\newblock The complexity of mean payoff games on graphs.
\newblock {\em Theoretical Computer Science}, 158:343--359, 1996.

\end{thebibliography}
\end{document}